\documentclass{article}
\usepackage[utf8]{inputenc}
\usepackage{graphicx} 
\usepackage{amsmath} 
\usepackage{amsmath} 
\usepackage{amsfonts} 
\usepackage{amssymb} 
\usepackage{amsthm} 
\usepackage{xfrac} 
\usepackage{algorithmic} 
\usepackage{physics} 
\usepackage{multicol} 
\usepackage{tablefootnote} 
\usepackage[margin=1in]{geometry}
\usepackage{eqnarray}

\newtheorem{lemma}{Lemma} 
\newtheorem{corollary}{Corollary}
\newtheorem{proposition}{Proposition}

\usepackage{tikz} 
\usetikzlibrary{quantikz2}
\usepackage{nicematrix} 
\usepackage{xfrac} 

\newcommand{\zero}{\mathbf{0}} 
\newcommand{\eye}{\mathbf{1}} 
\newcommand{\R}{\mathbf{R}} 

\usepackage{wrapfig} 

\usepackage{bm} 

\title{Extended quantum circuit diagrams}
\author{William Schober\thanks{william.schober@usi.ch} \\ Faculty of Informatics -- Università della Svizzera italiana, 6900 Lugano, Switzerland}
\date{\today}

\begin{document}

\addtocontents{toc}{\protect\setcounter{tocdepth}{0}}

\maketitle

\begin{abstract}
    We present a formulation of quantum circuit diagrams based on the exponential map which provides a new way to calculate graphically with circuits. We present a sound list of rewrite rules for this formulation and demonstrate a variety of example calculations. 
\end{abstract}

\section{Introduction}

Due to the steady progress in NISQ computer hardware over the last decade, the need for efficient compilers has prompted the development of a number of tools for quantum circuit optimization. Of particular interest are diagrammatic languages, which replace matrix-based calculation methods with diagrammatic ones. Diagrammatic calculations consist of manipulating diagrams with a collection of `rewrite rules', diagrammatic axioms which equate any two diagrams that can be transformed into one another by a series of rewrites. Quantum circuit diagrams are an obvious candidate for diagrammatic calculation; two circuits are equivalent if they represent the same unitary matrix, and the rewrite rules are a collection of circuit equivalences. 

Various diagrammatic languages have been developed for circuits. The linear-optics-inspired approach developed at Inria was shown to be complete \cite{clement_complete_2023,clement_quantum_2023} for universal circuits, and later axiomatically minimal as well \cite{clement_minimal_2024}. Weaker diagrammatic languages for various non-universal gate sets have also been studied, including the CNOT-dihedral group of order 16 $\langle\text{CNOT}, X, T\rangle$ \cite{amy_finite_2018}, the Clifford group $\langle H, S, \text{CZ}\rangle$ \cite{selinger_generators_2015}, the real stabilizer group $\langle H, Z, \text{CZ}\rangle$ \cite{makary_generators_2021}, and universal gate sets on a fixed number of qubits \cite{bian_generators_2023-1,bian_generators_2023}. Stronger diagrammatic languages have been studied as well, with the foremost being the ZX-calculus\footnote{along with its sister calculi like ZH-calculus and ZW-calculus.}, which can represent arbitrary linears maps, not just unitary ones \cite{van_de_wetering_zx-calculus_2020,coecke_picturing_2017}.

The following list of rewrite rules for quantum circuit diagrams differs from previous approaches in two main ways. First, multi-controlled gates are a native part of the notation, and second, all of its definitions rely on the exponential map. The former means that circuits involving multi-controlled gates are easier to work with, and the latter means a variety of things: Hamiltonians are more explicitly visible, matrix diagonalization can be done graphically, and speaking more speculatively, the notation may be more amenable to handling diagrammatic derivatives like those used for backpropagation in QML. Nevertheless all of these frameworks ultimately agree on the meaning of all gates, and every identity shown here can be used in any other formalism. 

In Section 2 we present the notation used in extended quantum circuit diagrams. In Section 3 we present its rewrite rules. In Appendix \ref{appendix:examplecalcs} we provide a variety of example calculations, focusing on reproducing known results to demonstrate usage of the language. In Appendix \ref{appendix:soundness} we provide soundness proofs for the rules. This formalism suggests a `semi-classical interpretation' of controlled gates which motivates a formal definition of a controlled operation, which is the subject of Appendix \ref{appendix:onthenotionofcontrol}. Appendix \ref{appendix:hamiltonians} is a brief recap on Hamiltonians for the Computer Science reader unfamiliar with them. Appendix \ref{appendix:referencesheet} is a reference sheet containing all of the notation and rules on one page.

\section{Notation}

Extended quantum circuit diagrams (EQCDs) are a graphical language comprised of three pieces of notation -- the \textbf{symbols} $(\bm{\bullet},\bm{\circ},\bm{\oplus},\bm{\ominus})$, \textbf{vertical lines}, and \textbf{numbers} -- along with a list of rewrite rules describing how these pieces interact with one another. Vertical lines in particular merit a slightly longer, more pedagogical section because they involve the most unusual definition.

As always, horizontal lines represent the identity process for a qubit. This `unit wire' is the default process used to connect symbols together in sequence, and its Hamiltonian is the all-zero matrix $\mathbf{0}$.
\begin{align*}
\begin{quantikz}
    & & &
\end{quantikz}
:= \eye = e^{i\pi \mathbf{0}}
\end{align*}

\subsection{Symbols $\bm{\bullet},\bm{\circ},\bm{\oplus},\bm{\ominus}$}

The following four symbols representing one-qubit gates are the building blocks from which everything else is constructed.\footnote{Technically we could make do with just $\bm{\bullet}$ and $\bm{\oplus}$ since the others can be generated from them via the commutation rules, but for the moment we're concentrating on defining a graphical language that is easy to use rather than notationally minimal.}
\begin{align*}
&\begin{quantikz}
    & \phase{} &
\end{quantikz}
:= e^{i\pi\dyad{1}} = Z
&\begin{quantikz}
    & \targ{} &
\end{quantikz}
&:= e^{i\pi\dyad{-}} = X \\
&\begin{quantikz}
    & \ophase{} &
\end{quantikz}
:= e^{i\pi\dyad{0}} = -Z
&\begin{quantikz}
    & \targm{} &
\end{quantikz}
&:= e^{i\pi\dyad{+}} = -X
\end{align*}
We will sometimes refer to the \textit{type} of a symbol. The $Z$-type (or dot-type) symbols $\bm{\bullet}$ and $\bm{\circ}$ have respective Hamiltonians $\dyad{1}$ and $\dyad{0}$, which are projectors onto the computational basis states $\ket{1}$ and $\ket{0}$ respectively. The $X$-type (or circle-type) symbols $\bm{\oplus}$ and $\bm{\ominus}$ have respective Hamiltonians $\dyad{-}$ and $\dyad{+}$, which are projectors onto the computational basis states $\ket{-}$ and $\ket{+}$ respectively.

\subsection{Vertical lines}

Vertical lines are used to link symbols together to form multi-qubit gates. When two symbols $\bigstar,\bm{\diamond} \in \{\bm{\bullet},\bm{\circ},\bm{\oplus},\bm{\ominus}\}$ with respective Hamiltonians $H_\bigstar,H_\diamond \in \{\dyad{1},\dyad{0},\dyad{-},\dyad{+}\}$ are connected by a vertical line, they together denote a two-qubit gate whose Hamiltonian is $H_\bigstar \otimes H_\diamond$. That is,
\begin{equation*}
\text{if} \quad
\begin{quantikz}
    & \str{} &
\end{quantikz}
= e^{i\pi H_\bigstar} \quad
\text{and} \quad
\begin{quantikz}
    & \diam{} &
\end{quantikz}
= e^{i\pi H_\diamond} \quad \text{then} \quad
\begin{quantikz}
    & \str{}\vqw{1} & \\
    & \diam{} &
\end{quantikz}
:= e^{i\pi H_\bigstar \otimes H_\diamond}.
\end{equation*}
This definition reproduces all of the controlled operations. Let's see some examples, starting with CNOT. CNOT is constructed graphically by joining $\bm{\bullet}$ (which has Hamiltonian $\dyad{1}$) to $\bm{\oplus}$ (which has Hamiltonian $\dyad{-}$), so it has the tensor product Hamiltonian $\dyad{1}\otimes\dyad{-}$.
\begin{equation*}
\begin{quantikz}[column sep=0.1cm,row sep=0.2cm]
    & \ctrl{1} & \\
    & \targ{} & 
\end{quantikz}
=
e^{i\pi\big(\dyad{1}\otimes\dyad{-}\big)}=e^{i\pi\dyad{1-}} = \begin{bmatrix}
   1 & 0 & 0 & 0 \\
   0 & 1 & 0 & 0 \\
   0 & 0 & 0 & 1 \\
   0 & 0 & 1 & 0 
\end{bmatrix}
\end{equation*}
In addition to this diagrammatic form, exponential form, and matrix form, CNOT also has a variety of `semi-classical interpretations'\footnote{`Semi-classical' because, unlike the classical CNOT, the notion of control and target wires is dependent on a local choice of basis, something that is discussed further in Appendix C.} as a controlled operation where one wire is read and the other is written to. 
\begin{equation}
\begin{quantikz}[column sep=0.1cm,row sep=0.2cm]
    & \ctrl{1} & \\
    & \targ{} & 
\end{quantikz}
=e^{i\pi\dyad{1-}} = \begin{array}{l}
\verb|if | \ket{1}_A \\
\verb| do | X_B
\end{array} 
= \begin{array}{l}
\verb|if | \ket{-}_B \\
\verb| do | Z_A
\end{array}
= \begin{array}{l}
\verb|if | \ket{1-}_{AB} \\
\verb| do | -\eye_{AB}
\end{array}
\end{equation}
Typically only the first of these semi-classical interpretations is discussed in literature (\verb|if |$\ket{1}_A$\verb| do |$X_B$), although the other two are mathematically equivalent. The admittance of these collections of semi-classical interpretations, and the resulting interchangeability of control and target wires, turns out to be a general feature of controlled operations; for more discussion, see Appendix \ref{appendix:onthenotionofcontrol}. Here in the main text we restrict ourselves to simply presenting these little blocks of `semi-classical logic' since they are frequently useful for easily simplifying calculations involving gates with many controls.

CZ is one of the few gates where the interchangeability of control and target is well-known, because it is diagonal in the computational basis. Graphically it is created by joining $\bm{\bullet}$ and $\bm{\bullet}$ with a vertical line. CZ has the tensor product Hamiltonian $\dyad{1}\otimes\dyad{1} = \dyad{11}$.

\begin{equation*}
\begin{quantikz}
    & \ctrl{1} & \\
    & \phase{} & 
\end{quantikz}
=
e^{i\pi\dyad{11}}=
\begin{array}{l}
\verb|if | \ket{1}_A \\
\verb| do | Z_B
\end{array} 
= \begin{array}{l}
\verb|if | \ket{1}_B \\
\verb| do | Z_A
\end{array}
= \begin{array}{l}
\verb|if | \ket{11}_{AB} \\
\verb| do | -\eye_{AB}
\end{array}
=
\begin{bmatrix}
   1 & 0 & 0 & 0 \\
   0 & 1 & 0 & 0 \\
   0 & 0 & 1 & 0 \\
   0 & 0 & 0 & -1 
\end{bmatrix}
\end{equation*}

Some more unusual controlled operations can be constructed with the vertical line.
\begin{equation*}
\begin{quantikz}[column sep=0.1cm,row sep=0.2cm]
    & \targ{}\vqw{1} & \\
    & \targ{} & 
\end{quantikz}
=
e^{i\pi\dyad{--}}=
\begin{array}{l}
\verb|if | \ket{-}_A \\
\verb| do | X_B
\end{array} 
= \begin{array}{l}
\verb|if | \ket{-}_B \\
\verb| do | X_A
\end{array}
= \begin{array}{l}
\verb|if | \ket{--}_{AB} \\
\verb| do | -\eye_{AB}
\end{array}
= {\displaystyle \frac{1}{2}}
\begin{bmatrix}
   1 & 1 & 1 & -1 \\
   1 & 1 & -1 & 1 \\
   1 & -1 & 1 & 1 \\
   -1 & 1 & 1 & 1 
\end{bmatrix}
\end{equation*}

\begin{equation*}
\begin{quantikz}[column sep=0.1cm,row sep=0.2cm]
    & \targm{}\vqw{1} & \\
    & \targ{} & 
\end{quantikz}
=
e^{i\pi\dyad{+-}}=
\begin{array}{l}
\verb|if | \ket{+}_A \\
\verb| do | X_B
\end{array} 
= \begin{array}{l}
\verb|if | \ket{-}_B \\
\verb| do | (-X)_A
\end{array}
= \begin{array}{l}
\verb|if | \ket{+-}_{AB} \\
\verb| do | -\eye_{AB}
\end{array}
= {\displaystyle \frac{1}{2}}
\begin{bmatrix}
   1 & 1 & -1 & 1 \\
   1 & 1 & 1 & -1 \\
   -1 & 1 & 1 & 1 \\
   1 & -1 & 1 & 1 
\end{bmatrix}
\end{equation*}

The definition of the vertical line extends inductively to any number of qubits in the straightforward way. If $n$ symbols on $n$ different wires are connected by a vertical line, they together denote an $n$-qubit controlled operation whose Hamiltonian is the $n$-fold tensor product of its constituent Hamiltonians. We will call such objects \textit{conditionals} since their behavior is dependent on satisfying one or more conditions. CNOT, CZ, and the two stranger gates above are all 2-qubit conditionals. Toffoli is a 3-qubit conditional that is created graphically by joining $\bm{\bullet}$, $\bm{\bullet}$ and $\bm{\oplus}$. Its Hamiltonian is $\dyad{1} \otimes \dyad{1} \otimes \dyad{-} = \dyad{11-}$. There is no need to keep track of the order in which the symbols were joined because the tensor product is associative: $\dyad{1} \otimes \big(\dyad{1} \otimes \dyad{-}\big) = \big(\dyad{1} \otimes \dyad{1}\big) \otimes \dyad{-}$. Hence there is no such notion of `join order', and the diagrammatic notation need not reflect it, instead using a simple vertical line through all symbols at once as normal.
\begin{equation*}
\begin{quantikz}[column sep=0.1cm,row sep=0.2cm]
    & \ctrl{1} & \\
    & \ctrl{1} & \\
    & \targ{} & 
\end{quantikz}
= e^{i\pi\dyad{11-}} = 
\begin{bmatrix}
    \eye & \zero & \zero & \zero \\
    \zero & \eye & \zero & \zero \\
    \zero & \zero & \eye & \zero \\
    \zero & \zero & \zero & X 
\end{bmatrix}
\end{equation*}

Toffoli admits seven different semi-classical interpretations, depending on which symbols are considered controls and which are considered targets.

\begin{equation*}
\begin{aligned}
\begin{quantikz}[column sep=0.1cm,row sep=0.2cm]
    & \ctrl{1} & \\
    & \ctrl{1} & \\
    & \targ{} & 
\end{quantikz}
&= \begin{array}{l}
\verb|if | \ket{1}_A \\
\verb| do | \text{CNOT}_{BC}
\end{array} 
= \begin{array}{l}
\verb|if | \ket{1}_B \\
\verb| do | \text{CNOT}_{AC}
\end{array}
= \begin{array}{l}
\verb|if | \ket{-}_C \\
\verb| do | \text{CZ}_{AB}
\end{array} \\
&= \begin{array}{l}
\verb|if | \ket{1-}_{BC} \\
\verb| do | Z_A
\end{array} 
= \begin{array}{l}
\verb|if | \ket{1-}_{AC} \\
\verb| do | Z_B
\end{array}
= \begin{array}{l}
\verb|if | \ket{11}_{AB} \\
\verb| do | X_C
\end{array}
= \begin{array}{l}
\verb|if | \ket{11-}_{ABC} \\
\verb| do | -\eye_{ABC}
\end{array}
\end{aligned}
\end{equation*}

Controlled unitaries are created by joining one or more symbols to an arbitrary unitary on $n$ qubits $U = e^{i\pi H}$. Letting $H = \sum \lambda_j \dyad{u_j}$ where $\ket{u_j}$ are eigenvectors of $H$ and $\lambda_j \in [0,2)$ w.l.o.g.,\footnote{The exponential map from the Hermitian matrices to the unitary matrices is known to be surjective because the Lie group $U(n)$ is compact and connected. The restriction of $H$'s eigenvalues to the interval $[0,2)$ makes the exponential map injective as well, so that every unitary $U = e^{i\pi H}$ is in one-to-one correspondence with a Hamiltonian $H$.} the standard controlled unitary formed by joining $\bm{\bullet}$ to $U$ has tensor product Hamiltonian $\dyad{1}\otimes H = \sum \lambda_j \dyad{1u_j}$. An explicit circuit is shown on the right for the case where $U = W\text{diag}(e^{i\pi\lambda_j})W^\dagger$ is a 2-qubit gate. For a definition of what the numbers $\lambda_j$ mean in the diagrammatic context, see the next section on numbers.
\begin{equation*}
\begin{quantikz}[column sep=0.1cm,row sep=0.2cm]
    & \ctrl{1} & \\
    & \gate{U} & 
\end{quantikz}
=
e^{i\pi\sum \lambda_j \dyad{1u_j}}
=
\begin{array}{l}
\verb|if | \ket{1}_A \\
\verb| do | U_B
\end{array} 
= \begin{array}{l}
\verb|if | \ket{u_j}_B \\
\verb| do | (Z^{\lambda_j})_A
\end{array}
=
\begin{bmatrix}
    \eye & \zero \\
    \zero & U 
\end{bmatrix}
=
\begin{quantikz}[column sep=0.35cm,row sep=0.15cm]
    & \ghost{H} & \ctrl{1} & \ctrl{1} & \ctrl{1} & \ctrl{1} & & \\
    & \gate[2]{W^\dagger}\gategroup[2,steps=6,style={dashed,rounded corners, inner xsep=2pt, inner ysep=2pt},label style={label position=below, yshift=-0.5cm}]{$U$} & \octrl[wire style={"\lambda_0"}]{1} & \octrl[wire style={"\lambda_1"}]{1} & \ctrl[wire style={"\lambda_2"}]{1} & \ctrl[wire style={"\lambda_3"}]{1} & \gate[2]{W} & \\
    & & \ocontrol{} & \control{} & \ocontrol{} & \control{} & &
\end{quantikz}
\end{equation*}

We can create stranger-looking controlled unitaries with the vertical line. Here is a controlled-$U$ that is controlled by a $\ket{-}$ rather than a $\ket{1}$.

\begin{equation*}
\begin{quantikz}[column sep=0.1cm,row sep=0.2cm]
    & \targ{}\vqw{1} & \\
    & \gate{U} & 
\end{quantikz}
=
e^{i\pi\sum \lambda_j \dyad{-u_j}}
=
\begin{array}{l}
\verb|if | \ket{-}_A \\
\verb| do | U_B
\end{array} 
= \begin{array}{l}
\verb|if | \ket{u_j}_B \\
\verb| do | (X^{\lambda_j})_A
\end{array}
\end{equation*}

Here is an operation with two symbols attached to $U$. Reading off from the diagram, since the attached unitaries are $\bm{\circ}$, $U$, and $\bm{\oplus}$, it has Hamiltonian $\dyad{0}\otimes\sum\lambda_j\dyad{u_j}\otimes\dyad{-}$. This gate has a variety of equivalent semi-classical interpretations, with the fifth being the closest to a standard one.

\begin{equation*}
\begin{aligned}
\begin{quantikz}[column sep=0.1cm,row sep=0.2cm]
    & \octrl{1} & \\
    & \gate{U} & \\
    & \targ{}\vqw{-1} &
\end{quantikz}
=
e^{i\pi\sum \lambda_j \dyad{0u_j-}}
&=
\begin{array}{l}
\verb|if | \ket{0}_A \\
\verb| do | {e^{i\pi\sum\lambda_j\dyad{u_j-}}}_{BC}
\end{array} 
=
\begin{array}{l}
\verb|if | \ket{u_j}_B \\
\verb| do | {\overline{\text{C}}\text{NOT}^{\lambda_j}}_{AC}
\end{array} 
=
\begin{array}{l}
\verb|if | \ket{-}_C \\
\verb| do | {e^{i\pi\sum\lambda_j\dyad{0u_j}}}_{AB}
\end{array} \\
&=
\begin{array}{l}
\verb|if | \ket{u_j-}_{BC} \\
\verb| do | {(-Z)^{\lambda_j}}_A
\end{array} 
=
\begin{array}{l}
\verb|if | \ket{0-}_{AC} \\
\verb| do | U_B
\end{array} 
=
\begin{array}{l}
\verb|if | \ket{0u_j}_{AB} \\
\verb| do | {X^{\lambda_j}}_C
\end{array}
=
\begin{array}{l}
\verb|if | \ket{0u_j-}_{ABC} \\
\verb| do | e^{i\pi\lambda_j}\eye_{ABC}
\end{array}
\end{aligned}
\end{equation*}

The most general kind of controlled operation need not involve any symbols explicitly at all. The definition given above extends almost as straightforwardly to joining any two arbitrary unitaries. Let $U,V$ be unitary operations on $n$ and $m$ qubits respectively, with respective Hamiltonians $H_U,H_V$. Joining $U$ and $V$ with a vertical line denotes an operation on $n+m$ qubits whose Hamiltonian is $H_U\otimes H_V$.

\begin{equation*}
\text{if} \quad
\begin{quantikz}
    & \gate{U} &
\end{quantikz}
= e^{i\pi H_U} \quad \text{and} \quad
\begin{quantikz}
    & \gate{V} &
\end{quantikz}
= e^{i\pi H_V} \quad \text{then} \quad
\begin{quantikz}
    & \gate{U}\vqw{1} & \\
    & \gate{V} &
\end{quantikz}
:= e^{i\pi H_U \otimes H_V}
\end{equation*}

For this definition to be well-defined\footnote{The author thanks P. Selinger and S. Perdrix for independently pointing out the problem of ill-definedness, and A. Clemént for providing an initial, alternate condition that also removes the problem of ill-definedness (not shown here because it is a stronger condition than the one presented). The author notes that one can also consider not making any restriction at all, in which case we would just need to label gates by their Hamiltonians rather than the resulting unitaries, since we lose injectivity of the exponential map. The resulting graphical language has some rather unusual properties that quantum circuits do not have, and hence is not discussed here.} we require the same condition as before, namely that all eigenphases of $U$ and $V$ be in the interval $[0,2\pi)$, i.e. that the eigenvalues of $H_U,H_V$ be reduced modulo 2 to the interval $[0,2)$. This loses no generality since the exponential map is still surjective over this domain. This restriction of its domain makes it injective as well, meaning every $U$ has a unique Hamiltonian $H$ such that $U = e^{i\pi H}$.

\subsection{Numbers}
Real numbers $\alpha$ written next to gates denote the $\alpha$th power/root of that gate. That is, for any symbol $\bigstar,\bm{\diamond} \in \{\bm{\bullet},\bm{\circ},\bm{\oplus},\bm{\ominus}\}$ with respective Hamiltonian $H_\bigstar,H_\diamond \in \{\dyad{1},\dyad{0},\dyad{-},\dyad{+}\}$
\begin{equation*}
\text{if} \quad
\begin{quantikz}
    & \str{} &
\end{quantikz}
= e^{i\pi H_\bigstar} \quad \text{then} \quad
\begin{quantikz}
    & \str{\alpha} &
\end{quantikz}
:= e^{i\pi \alpha H_\bigstar}.
\end{equation*}

Let's see some examples of common gates in this notation. We note that any gates without a power, like those from previous sections, are retroactively thought of as implicitly having a power of 1. We call such gates \textit{simple}.

\begin{equation*}
\begin{quantikz}
    & \gate{Z} &
\end{quantikz}
=
\begin{quantikz}
    & \phase{1} &
\end{quantikz}
=
\begin{quantikz}
    & \phase{} &
\end{quantikz} 
\hspace{4cm}
\begin{quantikz}
    & \gate{X} &
\end{quantikz}
=
\begin{quantikz}
    & \targ{1} &
\end{quantikz}
=
\begin{quantikz}
    & \targ{} &
\end{quantikz} 
\end{equation*}

\begin{equation*}
\begin{aligned}
&\begin{quantikz}
    & \gate{S} &
\end{quantikz}
=
\begin{quantikz}
    & \phase{\sfrac{1}{2}} &
\end{quantikz}
&\begin{quantikz}
    & \gate{S^\dagger} &
\end{quantikz}
=
\begin{quantikz}
    & \phase{\sfrac{-1}{2}} &
\end{quantikz}
&
&\begin{quantikz}
    & \gate{V} &
\end{quantikz}
=
\begin{quantikz}
    & \targ{\sfrac{1}{2}} &
\end{quantikz} 
&
&\begin{quantikz}
    & \gate{V^\dagger} &
\end{quantikz}
=
\begin{quantikz}
    & \targ{\sfrac{-1}{2}} &
\end{quantikz} 
\\
&\begin{quantikz}
    & \gate{T} &
\end{quantikz}
=
\begin{quantikz}
    & \phase{\sfrac{1}{4}} &
\end{quantikz}
&\begin{quantikz}
    & \gate{T^\dagger} &
\end{quantikz}
=
\begin{quantikz}
    & \phase{\sfrac{-1}{4}} &
\end{quantikz}
&
&
&
&
\\
&\begin{quantikz}
    & \gate{R(\theta)} &
\end{quantikz}
=
\begin{quantikz}
    & \phase{\theta/\pi} &
\end{quantikz}
&
&
&\begin{quantikz}
    & \gate{R_X(\theta)} &
\end{quantikz}
=
\begin{quantikz}
    & \targ{\theta/\pi} &
\end{quantikz}
&
&
\end{aligned}
\end{equation*}
\begin{equation*}
\begin{quantikz}
    & \gate{\text{Phase}(\theta)} &
\end{quantikz}
=
\begin{quantikz}
    & \phase{\theta/\pi} & \ophase{\theta/\pi} &
\end{quantikz}
=
\begin{quantikz}
    & \targ{\theta/\pi} & \targm{\theta/\pi} &
\end{quantikz}
\end{equation*}

Hadamard is a technically a derived operation in this language, but we include it as syntactic sugar since it's so common and useful in circuits -- it will even appear in the rewrite ruleset. There are many equivalent ways to define it in terms of the symbols. We give three of them here (two Euler decompositions and a diagonalization\footnote{For the curious reader, the three decompositions given here are, in order, $SXSe^{-i\pi/4}$, $e^{-i\pi/4}SXS$ (with the global phase incorporated into the gates, since this graphical language treats global phases as a derived part of the notation), and $V^\dagger T V Z V^\dagger T^\dagger V$.}), although we'll never have cause to actually use them, instead relying only the way $H$ interacts with the symbols via the rewrite ruleset.

\begin{equation*}
\begin{aligned}
\begin{quantikz}
    & \gate{H} &
\end{quantikz}
&:=
\begin{quantikz}
& & \phase{\sfrac{1}{2}} & \targ{\sfrac{1}{2}} & \phase{\sfrac{1}{4}} & \ophase{\sfrac{-1}{4}} & 
\end{quantikz}
=
\begin{quantikz}
& \ophase{\sfrac{-1}{4}} & \phase{\sfrac{1}{4}} & \targ{\sfrac{1}{2}} & \phase{\sfrac{1}{2}} & 
\end{quantikz} \\
&=
\begin{quantikz}
& \targ{\sfrac{1}{2}} & \phase{\sfrac{-1}{4}} & \targ{\sfrac{-1}{2}} & \phase{} & \targ{\sfrac{1}{2}} & \phase{\sfrac{1}{4}} & \targ{\sfrac{-1}{2}} &
\end{quantikz}
\Big(\quad =
\begin{quantikz}
& \gate{Y^{-1/4}} & \phase{} & \gate{Y^{1/4}} &
\end{quantikz}\Big)
\end{aligned}
\end{equation*}

Numbers extend in the obvious way to act as powers for arbitrary unitaries. Let $U$ be a unitary operation on $n$ qubits. Then for all $\alpha \in \R$,
\begin{equation*}
\text{if} \quad
\begin{quantikz}
    & \gate{U} &
\end{quantikz}
= e^{i\pi H_U} \quad \text{then} \quad
\begin{quantikz}
    & \gate{U} & \wire[l][1]["\alpha"{above,pos=0.2,scale=1.3}]{a}
\end{quantikz}
:=
e^{i\pi \alpha H_U}
=
\begin{quantikz}
    & \gate{U^\alpha} &
\end{quantikz}
\end{equation*}

Said another way, $U^\alpha$ is a unitary operation that has the same eigenvectors as $U$, and whose eigenphases have been uniformly scaled $\lambda_j \mapsto \alpha\lambda_j$. That is, if we write $U$ in its spectral decomposition $U = W\Lambda W^\dagger$ where $\Lambda = \text{diag}(e^{i\pi\lambda_j})$, then $\alpha$ written next to $U$ denotes $U^\alpha = W\Lambda^\alpha W^\dagger = W\text{diag}(e^{i\pi\alpha\lambda_j})W^\dagger$.
\begin{equation*}
\text{if} \quad
\begin{quantikz}
    & \gate{U} &
\end{quantikz}
=
\begin{quantikz}
    & \gate{W^\dagger} & \gate{\Lambda} & \gate{W} &
\end{quantikz}
\quad \text{then} \quad
\begin{quantikz}
    & \gate{U} & \wire[l][1]["\alpha"{above,pos=0.2,scale=1.3}]{a}
\end{quantikz}
:=
\begin{quantikz}
    & \gate{U^\alpha} &
\end{quantikz}
=
\begin{quantikz}
    & \gate{W^\dagger} & \gate{\Lambda^\alpha} & \gate{W} &
\end{quantikz}
\end{equation*}

\subsubsection*{Combining vertical lines and numbers}

Now that we're equipped with a definition for powers/roots of multi-qubit gates, let's combine the graphical notions of numbers and vertical lines to see some examples of common multi-qubit gates as they appear in this notation. The matrix $R_k$ shown here is the one used in the standard circuit implementation of the Quantum Fourier Transform.
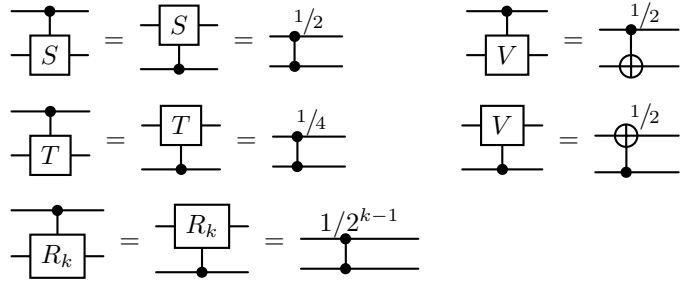
\begin{equation*}
\begin{aligned}
&\begin{quantikz}
    & \ctrl{1} & \\
    & \gate{S} &
\end{quantikz}
=
\begin{quantikz}
    & \gate{S} & \\
    & \ctrl{-1} &
\end{quantikz}
=
\begin{quantikz}
    & \phase{\sfrac{1}{2}} & \\
    & \ctrl{-1} &
\end{quantikz}
&\begin{quantikz}
    & \ctrl{1} & \\
    & \gate{V} &
\end{quantikz}
=
\begin{quantikz}
    & \phase{\sfrac{1}{2}} & \\
    & \trg{-1} &
\end{quantikz}
\\
&\begin{quantikz}
    & \ctrl{1} & \\
    & \gate{T} &
\end{quantikz}
=
\begin{quantikz}
    & \gate{T} & \\
    & \ctrl{-1} &
\end{quantikz}
=
\begin{quantikz}
    & \phase{\sfrac{1}{4}} & \\
    & \ctrl{-1} &
\end{quantikz}
&\begin{quantikz}
    & \gate{V} & \\
    & \ctrl{-1} &
\end{quantikz}
=
\begin{quantikz}
    & \targ{\sfrac{1}{2}} & \\
    & \ctrl{-1} &
\end{quantikz}
\\
&\begin{quantikz}
    & \ctrl{1} & \\
    & \gate{R_k} &
\end{quantikz}
=
\begin{quantikz}
    & \gate{R_k} & \\
    & \ctrl{-1} &
\end{quantikz}
=
\begin{quantikz}
    & \phase{1/2^{k-1}} & \\
    & \ctrl{-1} &
\end{quantikz}
&
\end{aligned}
\end{equation*}

The example of the CV may have the reader puzzled; surely it's the $V$ that has the power of $\sfrac{1}{2}$ attached to it, since it's the square root of $\bm{\oplus}$, and not the $\bm{\bullet}$ on the top wire. Why then is the root always written at the top? In fact there is no semantic meaning to the vertical position of numbers; they float freely along vertical lines. This is a consequence of the property $A\otimes \alpha B = \alpha A \otimes B$ of the tensor product.\footnote{If there are multiple numbers on the same vertical line, they multiply together. By convention we will always write just one number per gate, to the right of the gate and usually at the top.}
\begin{equation*} 
\begin{quantikz}
    & \ctrl{1} & \\
    & \targ{\alpha} &
\end{quantikz}
= e^{i\pi\dyad{1}\otimes (\alpha \dyad{-})} = e^{i\pi(\alpha \dyad{1})\otimes\dyad{-}} =
\begin{quantikz}
    & \phase{\alpha} & \\
    & \trg{-1} &
\end{quantikz}
\end{equation*}

We pause here to remark that the notation is now powerful enough to represent any quantum circuit since the universal gate set $\{H, T, \text{CNOT}\}$ can be constructed from symbols, vertical lines, and numbers. It remains to be shown how to calculate graphically with these pieces. That is the subject of the rest of the paper.

\section{Rewrite rules}

The following rewrite ruleset can either be seen as a collection of algebraic consequences of the above definitions (this is how they were discovered) or, in the purist graphical style, taken simply as axioms for a graphical language. As yet we make no claims about the completeness\footnote{We suspect it is not complete for all circuits (even those without measurements) since it lacks an Euler decomposition rule, something that is common to all known languages that are complete for circuits. Exactly what subset of circuits it \textit{is} complete for without the Euler decomposition rule is an open question.} or minimality\footnote{It is not minimal since rule (g) is derivable from the others. Whether the ruleset becomes minimal if we remove (g) is an open question.} of this ruleset for circuits; only that it is sound\footnote{Soundness proofs are given in Appendix \ref{appendix:soundness}.} and powerful enough to do some interesting graphical calculations, showcased in Appendix \ref{appendix:examplecalcs}.

EQCDs obey the following `flipflop' symmetry principles, which allow the ruleset to presented more easily.

\begin{enumerate}
\item \textbf{(flip) All rules hold under arbitrary permutations of the symbols within types. Applying either of the maps $(\bm{\bullet}\leftrightarrow\bm{\circ})$, or $(\bm{\oplus}\leftrightarrow\bm{\ominus})$, or both, to any rule yields another valid rule.}
\item \textbf{(flop) All rules hold under joint interchange of symbols between types. Applying the map $(\bm{\bullet}\leftrightarrow\bm{\oplus}$ and $\bm{\circ}\leftrightarrow\bm{\ominus})$ to any rule yields another valid rule.}
\end{enumerate}

Following each rule we provide a short description in words of its meaning, usually to emphasize what features are gained by applying the symmetry principles to it.

\subsection{Rules concerning one qubit}

\begin{equation*}
\text{\textbf{(c)} commutation:  } \begin{quantikz}
    & \phase{\alpha} & \ophase{\beta} &
\end{quantikz}
=
\begin{quantikz}
    & \ophase{\beta} & \phase{\alpha} &
\end{quantikz}
\quad \forall \alpha,\beta\in\R
\end{equation*}

Rule (c) says that symbols of the same type commute freely regardless of their respective powers. That is to say that $\bm{\bullet}$ commutes with $\bm{\circ}$, and $\bm{\oplus}$ commutes with $\bm{\ominus}$.

\begin{equation*}
\text{\textbf{(ac)} anticommutation:  } \begin{quantikz}
    & \phase{\alpha} & \targ{} &
\end{quantikz}
=
\begin{quantikz}
    & \targ{} & \ophase{\alpha} &
\end{quantikz}
\quad \Big(\text{ and} \quad
\begin{quantikz}
    & \ophase{\alpha} & \targ{} &
\end{quantikz}
=
\begin{quantikz}
    & \targ{} & \phase{\alpha} &
\end{quantikz}\Big)
\quad \forall \alpha\in\R
\end{equation*}

Rule (ac) says that when a symbol moves past an opposite-typed \textit{numberless} symbol, it flips. A special case of this rule is when neither symbol has a number ($\alpha = 1$), in which case one of the two symbols is flipped, but not both; which symbol gets flipped is up the user's choice.

\begin{equation*}
\text{\textbf{(H)} Hadamard:  } \begin{quantikz}
    & \phase{\alpha} & \gate{H} &
\end{quantikz}
=
\begin{quantikz}
    & \gate{H} & \targ{\alpha} &
\end{quantikz}
\quad \Big(\text{ and} \quad
\begin{quantikz}
    & \ophase{\alpha} & \gate{H} &
\end{quantikz}
=
\begin{quantikz}
    & \gate{H} & \targm{\alpha} &
\end{quantikz}\Big)
\quad \forall \alpha\in\R
\end{equation*}

Rule (H) gives the Hadamard gate its defining behavior as the canonical one-qubit flop map.

\begin{equation*}
\text{\textbf{(i)} involution:  }
\begin{quantikz}
    & \phase{2} & 
\end{quantikz}
=
\begin{quantikz}
    & \gate{H} & \wire[l][1]["2"{above,pos=0.2,scale=1.3}]{a}
\end{quantikz}
=
\begin{quantikz}
    & & & &
\end{quantikz}
\end{equation*}

Rule (i) says that all four symbols, and Hadamard, are involutions. Applying any of them twice in a row is equivalent to doing nothing.

\subsection{Rules concerning vertical lines}

The next five rules concern the behavior of vertical lines. All unitary gates $U,V$ that appear imply that the rule holds true for any gates $U,V$ acting on the same $n$ qubits.

\begin{equation*}
\text{\textbf{(C0)} empty control:  }
\begin{quantikz}
& \gate{\eye}\vqw{1} & \\
& \gate{U} & 
\end{quantikz}
=
\begin{quantikz}
& \ghost{H} & \\
& \ghost{H} & 
\end{quantikz}
\hspace{6mm}
\begin{array}{l}
\verb|if false| \\
\verb| do | U
\end{array}
=
\verb| pass |
\quad \Big(= \begin{array}{l}
\verb|if | \ket{u_j} \\
\verb| pass|
\end{array} \forall j \Big)
\end{equation*}

Rule (C0) says that connecting a vertical line to an empty wire `eats' the entire controlled gate, deleting the vertical line and anything else it was connected to. The logic $\verb|if false do | U$ is obtained by considering the upper (bundle of) wire to be the control, while the equivalent logic $\verb|if | \ket{u_j} \verb| pass | \forall j$ is obtained by considering the lower (bundle of) wire to be the control. Both are equivalent to simply doing nothing.

\begin{equation*}
\textbf{(d)} \hspace{-1mm}\begin{array}{l}
\text{distributivity of controls} \\
\text{  over gate composition:}
\end{array}
\text{  }
\begin{quantikz}[column sep=0.1cm,row sep=0.25cm]
    & \ctrl{1} & \ctrl{1} & \\
    & \gate{U} & \gate{V} &
\end{quantikz}
=
\begin{quantikz}[column sep=0.1cm,row sep=0.3cm]
    & \ctrl{1} & \\
    & \gate{U\circ V} &
\end{quantikz}
\hspace{8mm}
\begin{array}{l}
\verb|if | \ket{1}_A \\
\verb| do | U_B \\
\verb|if | \ket{1}_A \\
\verb| do | V_B
\end{array}
= \begin{array}{l}
\verb|if | \ket{1}_A \\
\verb| do | U_B \\
\verb| do | V_B
\end{array}
\end{equation*}

Rule (d) says that applying a series of controlled unitaries that are all controlled by the same symbol (shown: \textbullet) is equivalent to applying a single operation controlled by that symbol, which implements those unitaries in series. That is, symbolically, $CU \circ CV = C(U\circ V)$. We emphasize here that $U\circ V$ denotes applying $U$ followed by $V$. Since circuit diagrams are read left-to-right, we prefer to keep this convention throughout all of the graphical notation. In terms of a matrix multiplication, $U\circ V$ denotes $VU$.

\begin{equation*}
\text{\textbf{(e)} expansion:  }
\begin{quantikz}[column sep=0.1cm,row sep=0.3cm]
    & & \\
    & \gate{U} &
\end{quantikz}
=
\begin{quantikz}[column sep=0.1cm,row sep=0.25cm]
    & \octrl{1} & \ctrl{1} & \\
    & \gate{U} & \gate{U} &
\end{quantikz}
\hspace{8mm}
\verb| do | U_B
= \begin{array}{l}
\verb|if | \ket{0}_A \\
\verb| do | U_B \\
\verb|elif | \ket{1}_A \\
\verb| do | U_B
\end{array}
\end{equation*}

Rule (e) says that an unconditional operation is equivalent to a series of controlled unitaries, controlled by a complete basis of the control wires.

\begin{equation*}
\text{\textbf{(f)} forced commutation:  }
\begin{quantikz}[column sep=0.1cm,row sep=0.25cm]
    & \octrl{1} & \ctrl{1} & \\
    & \gate{U} & \gate{V} &
\end{quantikz}
=
\begin{quantikz}[column sep=0.1cm,row sep=0.25cm]
    & \ctrl{1} & \octrl{1} & \\
    & \gate{V} & \gate{U} &
\end{quantikz}
\hspace{8mm}
\begin{array}{l}
\verb|if | \ket{0}_A \\
\verb| do | U_B \\
\verb|elif | \ket{1}_A \\
\verb| do | V_B
\end{array}
= \begin{array}{l}
\verb|if | \ket{1}_A \\
\verb| do | V_B \\
\verb|elif | \ket{0}_A \\
\verb| do | U_B
\end{array}
\end{equation*}

Rule (f) says that controlled operations on the same set of wires commute if they have orthogonal (flipped) symbols on at least one wire, regardless of whether $U$ and $V$ commute or not.

\begin{equation*}
\text{\textbf{(g)} generalized commutation:  }
\text{   if} \quad
\begin{quantikz}[column sep=0.1cm,row sep=0.25cm]
    & \gate{U} & \gate{V} &
\end{quantikz}
=
\begin{quantikz}[column sep=0.1cm,row sep=0.25cm]
    & \gate{V'} & \gate{U} &
\end{quantikz}
\quad \text{then} \quad
\begin{quantikz}[column sep=0.1cm,row sep=0.25cm]
    & & \ctrl{1} & \\
    & \gate{U} & \gate{V} &
\end{quantikz}
=
\begin{quantikz}[column sep=0.1cm,row sep=0.25cm]
    & \ctrl{1} & & \\
    & \gate{V'} & \gate{U} &
\end{quantikz}
\end{equation*}

Rule (g) is derived from rules (d-f) and says how gate commutation relations are preserved by vertical lines in certain situations.

\begin{equation*}
\text{\textbf{(s)} SWAP:  }
\begin{quantikz}
& \permute{2,1} & \ghost{H} \\
& & \ghost{H}
\end{quantikz}
=
\begin{quantikz}
& \ctrl{1} & \targ{} & \ctrl{1} & \\
& \targ{} & \ctrl{-1} & \targ{} &
\end{quantikz}
\end{equation*}

Rule (s) defines the SWAP gate. SWAP takes the same form for almost any choice of symbols used to define it, shown in Appendix \ref{appendix:soundness}.

\begin{equation*}
\text{\textbf{(t)} twist:  }
\begin{quantikz}
& \permute{2,1} & \gate[2]{U_{AB}} & \\
& & \ghost{U_{AB}} &
\end{quantikz}
=
\begin{quantikz}
& \gate[2]{U_{BA}} & \permute{2,1} & \\
& \ghost{U_{BA}} & &
\end{quantikz}
\end{equation*}

Rule (t) says that any gate moving past a SWAP gets turned upside-down, i.e. has the roles of its qubits $A,B$ exchanged.

\subsection{Rules concerning numbers}

The next two rules concern the behavior of numbers. In short, numbers add horizontally and multiply vertically.

\begin{equation*}
\text{\textbf{(n+)} number addition:  }
\begin{quantikz}
& \gate{U} & \wire[l][1]["\alpha"{above,pos=0.2,scale=1.3}]{a} & & \gate{U} & \wire[l][1]["\beta"{above,pos=0.2,scale=1.3}]{a}
\end{quantikz}
=
\begin{quantikz}
& \gate{U} & & \wire[l][1]["\alpha+\beta"{above,pos=0.2,scale=1.3}]{a}
\end{quantikz}
\end{equation*}

Rule (n+) says that powers of the same gate in series add.

\begin{equation*}
\text{\textbf{(n*)} number multiplication:  }
\begin{quantikz}
& \gate{U}\vqw{1} & \wire[l][1]["\alpha"{above,pos=0.2,scale=1.3}]{a} \\
& \gate{V} & \wire[l][1]["\beta"{above,pos=0.2,scale=1.3}]{a}
\end{quantikz}
=
\begin{quantikz}
& \gate{U}\vqw{1} & \wire[l][1]["\alpha\beta"{above,pos=0.2,scale=1.3}]{a} \\
& \gate{V} &
\end{quantikz}
=
\begin{quantikz}
& \gate{U}\vqw{1} & \\
& \gate{V} & \wire[l][1]["\alpha\beta"{above,pos=0.2,scale=1.3}]{a} 
\end{quantikz}
\end{equation*}

Rule (n*) says that powers of gates connected along vertical lines multiply together. By convention we will only write one number per gate, to its right, and usually at the top.

\subsection{Rules concerning measurement}

These final two rules concern the notion of circuit equivalence for circuits involving measurements. Neither are original to this work; both are common principles in quantum computing \cite{nielsen_quantum_2010}. We take a rather conservative approach here and do not discuss discarding qubits at all, instead introducing only measurement, and with just enough structure to prove some communication primitives.

\begin{equation*}
\text{\textbf{(pdm)} principle of deferred measurement:  }
\begin{quantikz}
    & \meter{}\vcw{1} \\
    & \gate{U} &
\end{quantikz}
=
\begin{quantikz}
    & \ctrl{1} & \meter{} \\
    & \gate{U} & 
\end{quantikz}
\end{equation*}

Rule (pdm) is the principle of deferred measurement, and says that mid-circuit measurements can be delayed until the end of the circuit. Any classically controlled unitaries depending on the measurement result become quantum controlled unitaries. Measurements are assumed to be in the computational basis, and any resulting classically controlled gates are conditional on a result of 1, unless otherwise specified. 

\begin{equation*}
\text{\textbf{(mb)} change of measurement basis:  }
\begin{quantikz}
    & \meter{U}
\end{quantikz}
=
\begin{quantikz}
    & \gate{U^\dagger} & \meter{}
\end{quantikz}
\end{equation*}

Rule (mb) says that measuring in the basis $\{\ket{u_j}\}$ is equivalent to applying the unitary operation $U^\dagger$ followed by measuring in the computational basis, where $\ket{u_j}$ are the columns of $U$. Physically it is essentially just the statement that rotating your measurement apparatus clockwise is equivalent to rotating your state counterclockwise, since their relative orientation is what matters. We will denote measurements in the $\{\ket{u_j}\}$ basis by labelling the measurement box with $U$.

\section{Conclusion}
We presented a formulation of quantum circuit diagrams based on the exponential map which provides a new way to calculate graphically with circuits. We gave a sound list of rewrite rules for this formulation and demonstrate a variety of example calculations in the appendices. It remains to be shown what class of circuits this ruleset is complete for; we suspect it needs (at least) an Euler decomposition rule to be complete for all circuits. The primary advantage of this formulation compared to other circuit graphical languages is that it handles multi-controlled operations more easily, something that has caused some difficulty for graphical languages in the past.\footnote{For example, in ZX-calculus multi-controlled gates like Toffoli typically don't have an efficient representation and require a large number of spiders to represent. This was partly the motivation for the development of the closely related ZH-calculus.\cite{kissinger_conversation}} It also allows one to diagonalize circuits graphically, something we haven't seen in any previous work. As a final advantage, Hamiltonians are more explicit in this formalism since EQCDs are defined using the exponential map, which makes it easy to read off from a diagram what mathematical operation would need to be physically implemented in hardware to implement a given gate or subcircuit.

\section*{Acknowledgements}
The author would like to thank A. Clément, P. Selinger, N. Delorme, S. Wesley, S. Perdrix, and L. Laneve for helpful discussions.

\bibliography{main}{}
\bibliographystyle{unsrt}

\newpage
\appendix

\addtocontents{toc}{\protect\setcounter{tocdepth}{3}}
\tableofcontents

\newpage

\section{Example calculations}\label{appendix:examplecalcs}

In this section we showcase a variety of example calculations to demonstrate usage of EQCDs.

\subsection{Basic identities}

The following are a collection of simple identities that are useful shorthand in other graphical calculations, as well as miscellaneous calculations that don't fit into another category.

\subsubsection{Rule (i'): CNOT, CZ, Toffoli, CCZ, and arbitrary simple conditionals are involutions}

Applying rule (d) to either wire followed by rule (i) shows that CNOT$^2 = \eye$.

\begin{equation*}
\begin{quantikz}
& \phase{2} & \\
& \trg{-1} & 
\end{quantikz}
\stackrel{\text{(n+)}}{=}
\begin{quantikz}
& \ctrl{1} & \ctrl{1} & \\
& \targ{} & \targ{} &
\end{quantikz}
\stackrel{\text{(d)}}{=}
\begin{quantikz}
& \ctrl{1} & \\
& \gate{\bm{\oplus} \bm{\oplus}} & 
\end{quantikz}
\stackrel{\text{(i)}}{=}
\begin{quantikz}
& \ctrl{1} & \\
& \gate{\eye} & 
\end{quantikz}
\stackrel{\text{(C0)}}{=}
\begin{quantikz}
& \ghost{H} & \\
& & 
\end{quantikz} \quad \textbf{(i')}
\end{equation*}

The same construction for CZ, Toffoli, and CCZ demonstrates that all three are also involutions. The same construction also applies to any simple conditional; this proof therefore says that all simple conditionals are involutions. We'll call this derived rule (i').

\subsubsection{CNOTs and Paulis}
It's well-known that $Z$ commutes through the control of a CNOT, and $X$ commutes through the target. Here we see it's simply because the involved symbols are the same, and hence by rule (g), they commute.

\begin{equation*}
\begin{aligned}
\begin{quantikz}[column sep=0.1cm,row sep=0.15cm]
    & \gate{Z} & \ctrl{1} & \\
    & & \targ{} &
\end{quantikz}
:=
\begin{quantikz}[column sep=0.1cm,row sep=0.25cm]
    & \phase{} & \ctrl{1} & \\
    & & \targ{} &
\end{quantikz}
&\stackrel{\text{(g),(c)}}{=}
\begin{quantikz}[column sep=0.1cm,row sep=0.25cm]
    & \ctrl{1} & \phase{} & \\
    & \targ{} & &
\end{quantikz}
=:
\begin{quantikz}[column sep=0.1cm,row sep=0.15cm]
    & \ctrl{1} & \gate{Z} & \\
    & \targ{} & &
\end{quantikz} \\
\begin{quantikz}[column sep=0.1cm,row sep=0.15cm]
    & & \ctrl{1} & \\
    & \gate{X} & \targ{} &
\end{quantikz}
:=
\begin{quantikz}[column sep=0.1cm,row sep=0.25cm]
    & & \ctrl{1} & \\
    & \targ{} & \targ{} &
\end{quantikz}
&\stackrel{\text{(g),(c)}}{=}
\begin{quantikz}[column sep=0.1cm,row sep=0.25cm]
    & \ctrl{1} & & \\
    & \targ{} & \targ{} &
\end{quantikz}
=:
\begin{quantikz}[column sep=0.1cm,row sep=0.15cm]
    & \ctrl{1} & & \\
    & \targ{} & \gate{X} &
\end{quantikz}
\end{aligned}
\end{equation*}

\subsubsection{Rule (p): Parity}
Applying the expansion rule to the parity operator $Z\otimes Z$ gives a useful identity.

\begin{equation*}
\begin{quantikz}
& \control{} & \\
& \control{} & 
\end{quantikz}
\stackrel{\text{(e)}}{=}
\begin{quantikz}
& \octrl{1}\gategroup[2,steps=2,style={dashed,rounded
corners, inner
xsep=2pt, inner
ysep=2pt}]{$1\otimes Z$} & \ctrl{1} & & \ctrl{1}\gategroup[2,steps=2,style={dashed,rounded
corners, inner
xsep=2pt, inner
ysep=2pt}]{$Z\otimes 1$} & \ctrl{1} & \qw \\
& \control{} & \control{} & & \ocontrol{} & \control{} & \qw 
\end{quantikz}
\stackrel{\text{(f)}}{=}
\begin{quantikz}
& \octrl{1} & & \ctrl{1}\gategroup[2,steps=2,style={dashed,rounded
corners, inner
xsep=2pt, inner
ysep=2pt}]{cancel} & \ctrl{1} & & \ctrl{1} & \\
& \control{} & & \control{} & \control{} & & \ocontrol{} &
\end{quantikz}
\stackrel{\text{(i')}}{=}
\begin{quantikz}[column sep=0.25cm]
& \octrl{1} & \ctrl{1} & \qw \\
& \control{} & \ocontrol{} & \qw 
\end{quantikz}
\end{equation*}

It says that two unconnected black dots (or two unconnected white dots, by the same reasonaing) can be replaced by a square whose opposite corners are colored white and black respectively. Note that the orientation doesn't matter, since these two conditionals commute by rule (f). We will call this derived rule (p). 

\begin{equation*}
\begin{quantikz}
& \control{} & \\
& \control{} & 
\end{quantikz}
=
\begin{quantikz}[column sep=0.25cm]
& \octrl{1} & \ctrl{1} & \qw \\
& \control{} & \ocontrol{} & \qw 
\end{quantikz}
=
\begin{quantikz}[column sep=0.25cm]
& \ctrl{1} & \octrl{1} & \qw \\
& \ocontrol{} & \control{} & \qw 
\end{quantikz}
=
\begin{quantikz}
& \ocontrol{} & \\
& \ocontrol{} & 
\end{quantikz} \quad \textbf{(p)}
\end{equation*}

The same proof extends to the parity operator on $n$ qubits $\bigotimes^n Z$. It takes the form of a series of the $2^{n-1}$ odd simple conditionals; those simple conditionals with a $\bm{\bullet}/\bm{\circ}$ on every wire and an odd number of $\bm{\bullet}$ symbols. For example, on three qubits, we get the Z-type conditionals that project onto 001, 010, 100, and 111.

\begin{equation*}
\begin{quantikz}
& \phase{} & \\
& \phase{} & \\
& \phase{} &
\end{quantikz}
=
\begin{quantikz}
& \octrl{2} & \octrl{2} & \ctrl{2} & \ctrl{2} & \\
& \ophase{} & \phase{} & \ophase{} & \phase{} & \\
& \phase{} & \ophase{} & \ophase{} & \phase{} &
\end{quantikz}
\end{equation*}

\subsubsection{CNOT and Paulis, part 2: error propagation}
When a $Z$ moves past the target of a CNOT, it is copied to the control wire. Likewise for an $X$ moving past the control wire of a CNOT. This is most well-known in the context of error propagation in stabilizer codes, where the Pauli gate is thought of as phase- or bit-flip error respectively that is spread through the circuit by CNOTs if not corrected. These identities are consequences of the parity rule (p) above.

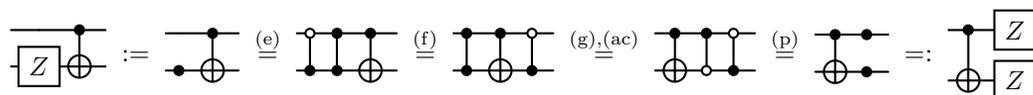
\begin{equation*}
\begin{quantikz}[column sep=0.1cm,row sep=0.15cm]
    & & \ctrl{1} & \\
    & \gate{Z} & \targ{} &
\end{quantikz}
:=
\begin{quantikz}[column sep=0.1cm,row sep=0.25cm]
    & & \ctrl{1} & \\
    & \phase{} & \targ{} &
\end{quantikz}
\stackrel{\text{(e)}}{=}
\begin{quantikz}[column sep=0.1cm,row sep=0.25cm]
    & \octrl{1} & \ctrl{1} & \ctrl{1} & \\
    & \phase{} & \phase{} & \targ{} &
\end{quantikz}
\stackrel{\text{(f)}}{=}
\begin{quantikz}[column sep=0.1cm,row sep=0.25cm]
    & \ctrl{1} & \ctrl{1} & \octrl{1} & \\
    & \phase{} & \targ{} & \phase{} &
\end{quantikz}
\stackrel{\text{(g),(ac)}}{=}
\begin{quantikz}[column sep=0.1cm,row sep=0.25cm]
    & \ctrl{1} & \ctrl{1} & \octrl{1} & \\
    & \targ{} & \ophase{} & \phase{} &
\end{quantikz}
\stackrel{\text{(p)}}{=}
\begin{quantikz}[column sep=0.1cm,row sep=0.25cm]
    & \ctrl{1} & \phase{} & \\
    & \targ{} & \phase{} & 
\end{quantikz}
=:
\begin{quantikz}[column sep=0.1cm,row sep=0.15cm]
    & \ctrl{1} & \gate{Z} & \\
    & \targ{} & \gate{Z} &
\end{quantikz}
\end{equation*}

The $X$ version of this identity is just the flopped $Z$ version, but it can of course be proven in the same way, using a flopped rule (p).

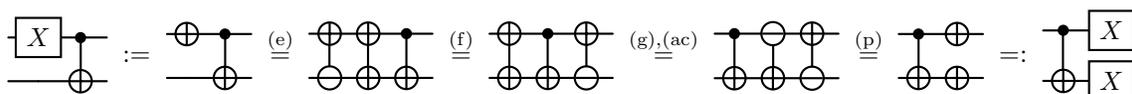
\begin{equation*}
\begin{quantikz}[column sep=0.1cm,row sep=0.15cm]
    & \gate{X} & \ctrl{1} & \\
    & & \targ{} &
\end{quantikz}
:=
\begin{quantikz}[column sep=0.1cm,row sep=0.25cm]
    & \targ{} & \ctrl{1} & \\
    & & \targ{} &
\end{quantikz}
\stackrel{\text{(e)}}{=}
\begin{quantikz}[column sep=0.1cm,row sep=0.25cm]
    & \targ{}\vqw{1} & \targ{}\vqw{1} & \ctrl{1} & \\
    & \targm{} & \targ{} & \targ{} &
\end{quantikz}
\stackrel{\text{(f)}}{=}
\begin{quantikz}[column sep=0.1cm,row sep=0.25cm]
    & \targ{}\vqw{1} & \ctrl{1} & \targ{}\vqw{1} & \\
    & \targ{} & \targ{} & \targm{} &
\end{quantikz}
\stackrel{\text{(g),(ac)}}{=}
\begin{quantikz}[column sep=0.1cm,row sep=0.25cm]
    & \ctrl{1} & \targm{}\vqw{1} & \targ{}\vqw{1} & \\
    & \targ{} & \targ{} & \targm{} &
\end{quantikz}
\stackrel{\text{(p)}}{=}
\begin{quantikz}[column sep=0.1cm,row sep=0.25cm]
    & \ctrl{1} & \targ{} & \\
    & \targ{} & \targ{} &
\end{quantikz}
=:
\begin{quantikz}[column sep=0.1cm,row sep=0.15cm]
    & \ctrl{1} & \gate{X} & \\
    & \targ{} & \gate{X} &
\end{quantikz}
\end{equation*}

What is less often discussed is that another solution is obtained by pushing all of the change onto the CNOT rather than the Pauli. The proof is actually simpler, and doing it for $X$ gives a fairly intuitive identity saying that negating the control before applying CNOT is equivalent to controlling the CNOT on 0, followed by negating the control.

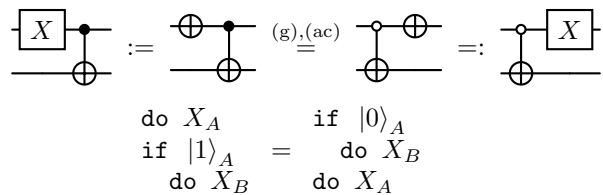
\begin{equation*}
\begin{aligned}
\begin{quantikz}[column sep=0.1cm,row sep=0.15cm]
    & \gate{X} & \ctrl{1} & \\
    & & \targ{} &
\end{quantikz}
:=
\begin{quantikz}[column sep=0.1cm,row sep=0.25cm]
    & \targ{} & \ctrl{1} & \\
    & & \targ{} &
\end{quantikz}
&\stackrel{\text{(g),(ac)}}{=}
\begin{quantikz}[column sep=0.1cm,row sep=0.25cm]
    & \octrl{1} & \targ{} & \\
    & \targ{} & &
\end{quantikz}
=:
\begin{quantikz}[column sep=0.1cm,row sep=0.15cm]
    & \octrl{1} & \gate{X} & \\
    & \targ{} & &
\end{quantikz} \\
\begin{array}{l}
\verb|do | X_A \\
\verb|if |\ket{1}_A \\
\verb|  do | X_B
\end{array}
&=
\begin{array}{l}
\verb|if |\ket{0}_A \\
\verb|  do | X_B \\
\verb|do | X_A
\end{array}
\end{aligned}
\end{equation*}

What is never discussed is the equivalent identity for $Z$, since normal circuits lack the notation for it. 

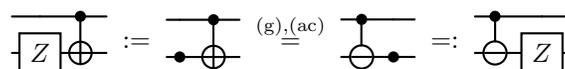
\begin{equation*}
\begin{quantikz}[column sep=0.1cm,row sep=0.15cm]
    & & \ctrl{1} & \\
    & \gate{Z} & \targ{} &
\end{quantikz}
:=
\begin{quantikz}[column sep=0.1cm,row sep=0.25cm]
    & & \ctrl{1} & \\
    & \phase{} & \targ{} &
\end{quantikz}
\stackrel{\text{(g),(ac)}}{=}
\begin{quantikz}[column sep=0.1cm,row sep=0.25cm]
    & \ctrl{1} & & \\
    & \targm{} & \phase{} &
\end{quantikz}
=:
\begin{quantikz}[column sep=0.1cm,row sep=0.15cm]
    & \ctrl{1} & & \\
    & \targm{} & \gate{Z} &
\end{quantikz}
\end{equation*}

This strange-looking identity is in fact no stranger than the previous one, since it has only been flopped. Reading it as is asks the reader to be able to read CNOT with its control and target roles reversed. The identity says that flipping the `control bit' $(+\leftrightarrow-)$ before applying CNOT is equivalent to controlling the CNOT on a +, followed by flipping the `control bit' (quotes added to emphasize that the control bit is now being thought of as being on the bottom wire). In other words, the two pieces of semi-classical logic being equated here are

\begin{equation*}
\begin{array}{l}
\verb|do | Z_B \\
\verb|if |\ket{-}_B \\
\verb|  do | Z_A
\end{array}
=
\begin{array}{l}
\verb|if |\ket{+}_B \\
\verb|  do | Z_A \\
\verb|do | Z_B
\end{array}
\end{equation*}

\subsubsection{Rule (5CX): Partially overlapping CNOTs}
The following identity concerning the commutation relation between partially overlapping CNOTs is well-known, appearing at least as far back as 2011 \cite{garcia-escartin_equivalent_2011}, and probably long before, because CNOT is the most common 2-qubit gate and is often the target of optimization routines. 

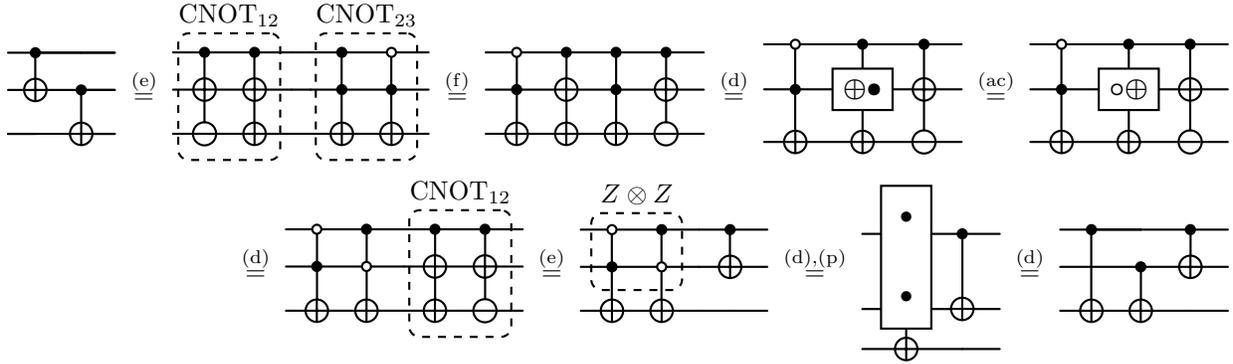
\begin{equation*}
\begin{aligned}
\begin{quantikz}[column sep=0.2cm]
& \ctrl{1} & \qw & \qw \\
& \targ{} & \ctrl{1} & \qw \\
& \qw & \targ{} & \qw
\end{quantikz}
\stackrel{\text{(e)}}{=}
\begin{quantikz}
& \ctrl{1}\gategroup[3,steps=2,style={dashed,rounded corners, inner xsep=2pt, inner ysep=2pt}]{CNOT$_{12}$} & \ctrl{1} & & & \ctrl{1}\gategroup[3,steps=2,style={dashed,rounded corners, inner xsep=2pt, inner ysep=2pt}]{CNOT$_{23}$} & \octrl{1} & \\
& \targ{}\vqw{1} & \targ{}\vqw{1} & & & \ctrl{1} & \ctrl{1} & \\
& \targm{} & \targ{} & & & \targ{} & \targ{} &
\end{quantikz}
\stackrel{\text{(f)}}{=}
\begin{quantikz}
& \octrl{1} & \ctrl{1} & \ctrl{1} & \ctrl{1} & \\
& \ctrl{1} & \targ{}\vqw{1} & \ctrl{1} & \targ{}\vqw{1} & \\
& \targ{} & \targ{} & \targ{} & \targm{} &
\end{quantikz}
\stackrel{\text{(d)}}{=}
\begin{quantikz}
& \octrl{1} & \ctrl{1} & \ctrl{1} & \\
& \ctrl{1} & \gate{\bm{\oplus} \bm{\bullet}}\vqw{1} & \targ{}\vqw{1} & \\
& \targ{} & \targ{} & \targm{} &
\end{quantikz}
\stackrel{\text{(ac)}}{=}
\begin{quantikz}
& \octrl{1} & \ctrl{1} & \ctrl{1} & \\
& \ctrl{1} & \gate{\bm{\circ} \bm{\oplus}}\vqw{1} & \targ{}\vqw{1} & \\
& \targ{} & \targ{} & \targm{} &
\end{quantikz} 
\\
\stackrel{\text{(d)}}{=}
\begin{quantikz}
& \octrl{1} & \ctrl{1} & & \ctrl{1}\gategroup[3,steps=2,style={dashed,rounded corners, inner xsep=2pt, inner ysep=2pt}]{CNOT$_{12}$} & \ctrl{1} & \\
& \ctrl{1} & \ocontrol{}\vqw{1} & & \targ{}\vqw{1} & \targ{}\vqw{1} & \\
& \targ{} & \targ{} & & \targ{} & \targm{} &
\end{quantikz}
\stackrel{\text{(e)}}{=}
\begin{quantikz}
& \octrl{1}\gategroup[2,steps=2,style={dashed,rounded
corners, inner xsep=0pt, inner
ysep=1pt}]{$Z\otimes Z$} & \ctrl{1} & & \ctrl{1} & \\
& \ctrl{1} & \ocontrol{}\vqw{1} & & \targ{} & \\
& \targ{} & \targ{} & & &
\end{quantikz}
\stackrel{\text{(d),(p)}}{=}
\begin{quantikz}[row sep=0.1cm]
& \gate[2]{\verticaltext{\textbullet\ \textbullet}} & \ctrl{1} & \\
& \ghost{Z} & \targ{} & \\
& \targ{}\vqw{-1} & &
\end{quantikz}
\stackrel{\text{(d)}}{=}
\begin{quantikz}
& \ctrl{2} & \qw & \ctrl{1} & \\
& & \ctrl{1} & \targ{} & \\
& \targ{} & \targ{} & & 
\end{quantikz}
\end{aligned}
\end{equation*}

Rearranging, we get another useful identity. In the circuit axiomitization approach taken by Clément et al \cite{clement_minimal_2024} this identity is considered an axiom; we follow their most recent notation and call this derived rule (5CX), for `five CNOT'.\footnote{To be precise, our rule (5CX) is the flopped version of the one shown in \cite{clement_minimal_2024}, but we consider all rules related to each other by flipflop symmetries to be the same!}

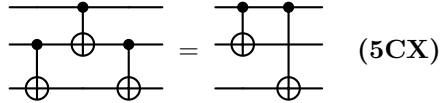
\begin{equation*}
\begin{quantikz}[column sep=0.2cm]
& & \ctrl{1} & & \\
& \ctrl{1} & \targ{} & \ctrl{1} & \\
& \targ{} & & \targ{} &
\end{quantikz}
=
\begin{quantikz}[column sep=0.2cm]
& \ctrl{1} & \ctrl{2} & \\
& \targ{} & & \\
& & \targ{} &
\end{quantikz}  \quad \textbf{(5CX)}
\end{equation*}

\subsubsection{Rule (t'): SWAP properties}
SWAP has some standard properties which we prove explicitly here: it is an involution,

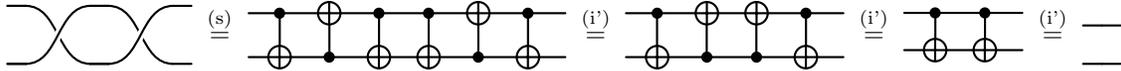
\begin{equation*}
\begin{quantikz}
& \permute{2,1} & \permute{2,1} & \ghost{H} \\
& & & \ghost{H}
\end{quantikz}
\stackrel{\text{(s)}}{=}
\begin{quantikz}
& \ctrl{1} & \targ{} & \ctrl{1} & \ctrl{1} & \targ{} & \ctrl{1} & \\
& \targ{} & \ctrl{-1} & \targ{} & \targ{} & \ctrl{-1} & \targ{} &
\end{quantikz}
\stackrel{\text{(i')}}{=}
\begin{quantikz}
& \ctrl{1} & \targ{} & \targ{} & \ctrl{1} & \\
& \targ{} & \ctrl{-1} & \ctrl{-1} & \targ{} &
\end{quantikz}
\stackrel{\text{(i')}}{=}
\begin{quantikz}
& \ctrl{1} & \ctrl{1} & \\
& \targ{} & \targ{} &
\end{quantikz}
\stackrel{\text{(i')}}{=}
\begin{quantikz}
& \ghost{H} & \\
& &
\end{quantikz}
\end{equation*}

...and anything sandwiched by SWAPs gets turned upside-down.

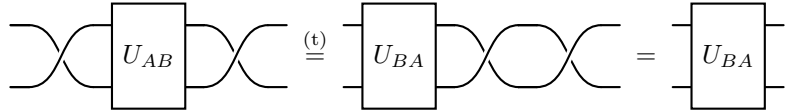
\begin{equation*}
\begin{quantikz}
& \permute{2,1} & \gate[2]{U_{AB}} & \permute{2,1} & \\
& & \ghost{U_{AB}} & &
\end{quantikz}
\stackrel{\text{(t)}}{=}
\begin{quantikz}
& \gate[2]{U_{BA}} & \permute{2,1} & \permute{2,1} & \\
& \ghost{U_{BA}} & & &
\end{quantikz}
=
\begin{quantikz}
& \gate[2]{U_{BA}} & \\
& \ghost{U_{BA}} & 
\end{quantikz}
\end{equation*}

These rules are simple and common enough that we will collectively refer to them as derived rule (t').

\subsubsection{QFT$_2$ squared}
The quantum fourier transform is a convenient testbed for circuit calculations because its standard circuit implementation involves all of our pieces: symbols,\footnote{Dots explicitly, circles implicitly through the presence of $H$ gates.} vertical lines, and numbers. The \textit{square} of the quantum fourier transform is particularly nice to work with since it's easy to algebraically convince oneself that it implements the integer negation function $\ket{k}\mapsto\ket{-k}$ in the cyclic group $\mathbf{Z}_{2^n}$ on the computational basis.

\begin{equation*}
(\text{QFT}_n)^2 : \ket{k} \mapsto
\begin{cases}
    \ket{0} & \text{if } k=0 \\
    \ket{2^n-k} & \text{if } k\neq 0
\end{cases}
\end{equation*}

Since this function is implementable with a classical logic circuit, it follows that it should always be possible to reduce the involved quantum circuit to an embedded classical circuit, i.e. a quantum circuit that only involves things like $X$, CNOT, and Toffoli, but not things like $Z$, $H$, or CZ.

We note again that the gate $R_k$ appearing in the standard circuit implementation of the QFT is, in this notation, a $\bm{\bullet}$ with power $1/2^{k-1}$.

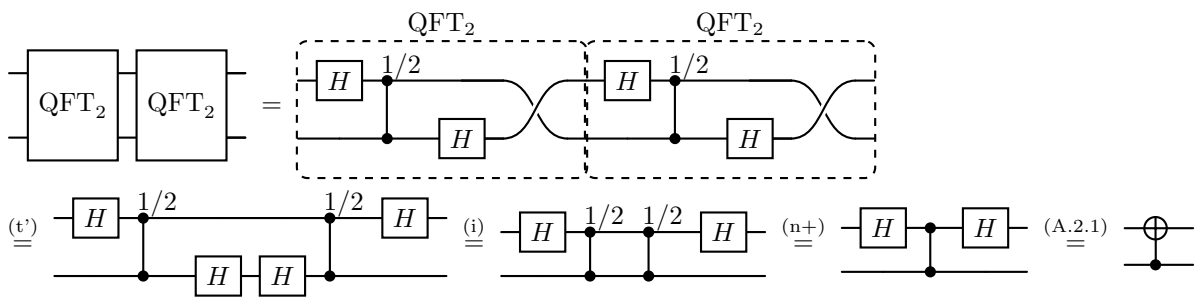
\begin{equation*}
\begin{aligned}
&\begin{quantikz}
& \gate[2]{\text{QFT}_2} & \gate[2]{\text{QFT}_2} & \qw \\
& & & \qw
\end{quantikz}
=
\begin{quantikz}
& \gate{H}\gategroup[2,steps=4,style={dashed,rounded
corners}]{QFT$_2$} & \phase{1/2} & & \permute{2,1} & & \gate{H}\gategroup[2,steps=4,style={dashed,rounded
corners}]{QFT$_2$} & \phase{1/2} & & \permute{2,1} & \\
& & \ctrl{-1} & \gate{H} & & & & \ctrl{-1} & \gate{H} & &
\end{quantikz}
\\
&\stackrel{\text{(t')}}{=}
\begin{quantikz}
& \gate{H} & \phase{1/2} & & & \phase{1/2} & \gate{H} & \\
& & \ctrl{-1} & \gate{H} & \gate{H} & \ctrl{-1} & &
\end{quantikz}
\stackrel{\text{(i)}}{=}
\begin{quantikz}
& \gate{H} & \phase{1/2} & \phase{1/2} & \gate{H} & \\
& & \ctrl{-1} & \ctrl{-1} & & 
\end{quantikz}
\stackrel{\text{(n+)}}{=}
\begin{quantikz}
& \gate{H} & \phase{} & \gate{H} & \\
& & \ctrl{-1} & & 
\end{quantikz}
\stackrel{\text{(A.2.1)}}{=}
\begin{quantikz}
& \targ{} & \\
& \ctrl{-1} &
\end{quantikz}
\end{aligned}
\end{equation*}

\subsubsection{QFT$_3$ squared}

At once step in this calculation we identify the (QFT$_2$)$^2$ as a subcircuit and replace it with a CNOT.

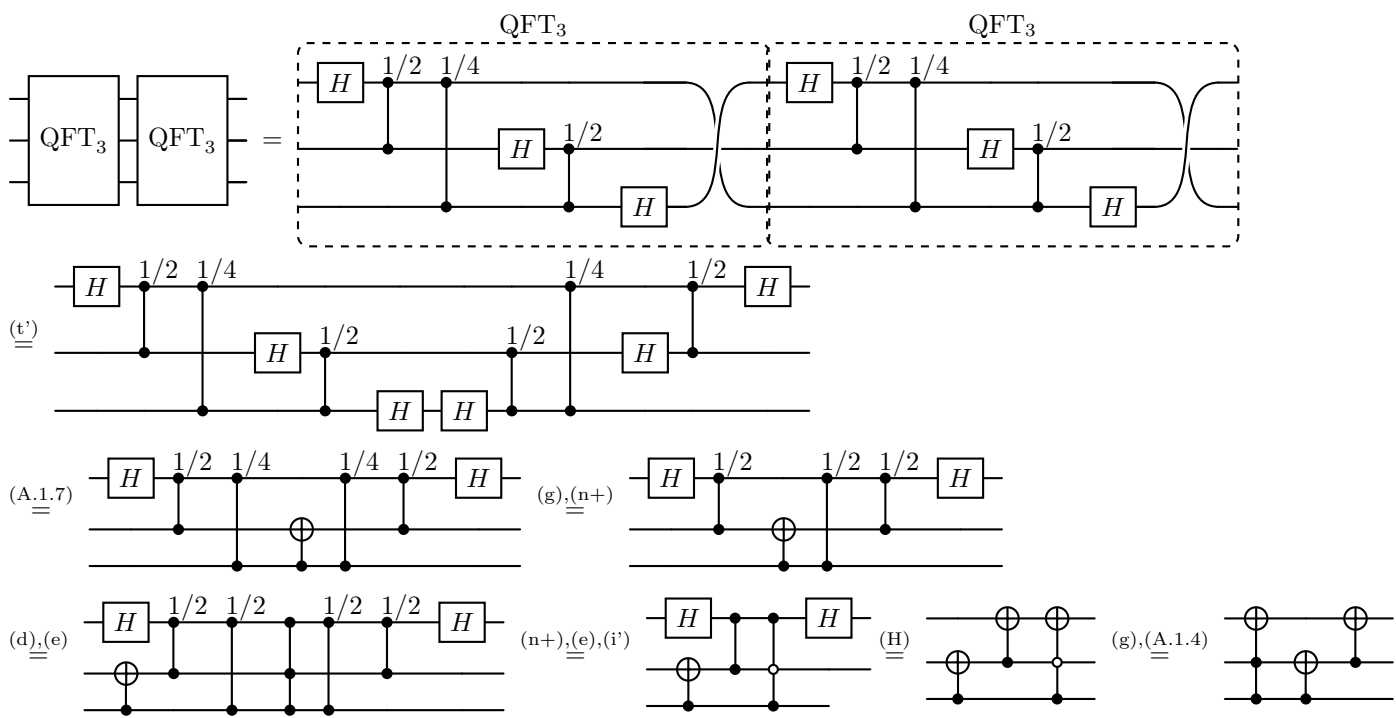
\begin{equation*}
\begin{aligned}
&\begin{quantikz}
& \gate[3]{\text{QFT}_3} & \gate[3]{\text{QFT}_3} & \qw \\
& & & \qw \\
& & & \qw
\end{quantikz}
=
\begin{quantikz}
& \gate{H}\gategroup[3,steps=7,style={dashed,rounded
corners}]{QFT$_3$} & \phase{1/2} & \phase{1/4} & & & & \permute{3,2,1} & & \gate{H}\gategroup[3,steps=7,style={dashed,rounded
corners}]{QFT$_3$} & \phase{1/2} & \phase{1/4} & & & & \permute{3,2,1} & \\
& & \ctrl{-1} & & \gate{H} & \phase{1/2} & & & & & \ctrl{-1} & & \gate{H} & \phase{1/2} & & & \\
& & & \ctrl{-2} & & \ctrl{-1} & \gate{H} & & & & & \ctrl{-2} & & \ctrl{-1} & \gate{H} & & 
\end{quantikz}
\\
&\stackrel{\text{(t')}}{=}
\begin{quantikz}
& \gate{H} & \phase{1/2} & \phase{1/4} & & & & & & \phase{1/4} & & \phase{1/2} & \gate{H} & \\
& & \ctrl{-1} & & \gate{H} & \phase{1/2} & & & \phase{1/2} & & \gate{H} & \ctrl{-1} & & \\
& & & \ctrl{-2} & & \ctrl{-1} & \gate{H} & \gate{H} & \ctrl{-1} & \ctrl{-2} & & & & 
\end{quantikz}
\\
&\stackrel{\text{(A.1.7)}}{=}
\begin{quantikz}
& \gate{H} & \phase{1/2} & \phase{1/4} & & \phase{1/4} & \phase{1/2} & \gate{H} & \\
& & \ctrl{-1} & & \targ{} & & \ctrl{-1} & & \\
& & & \ctrl{-2} & \ctrl{-1} & \ctrl{-2} & & &
\end{quantikz}
\stackrel{\text{(g),(n+)}}{=}
\begin{quantikz}
& \gate{H} & \phase{1/2} & & \phase{1/2} & \phase{1/2} & \gate{H} & \\
& & \ctrl{-1} & \targ{} & & \ctrl{-1} & & \\
& & & \ctrl{-1} & \ctrl{-2} & & &
\end{quantikz}
\\
&\stackrel{\text{(d),(e)}}{=}
\begin{quantikz}
& \gate{H} & \phase{1/2} & \phase{1/2} & \phase{} & \phase{1/2} & \phase{1/2} & \gate{H} & \\
& \targ{} & \ctrl{-1} & & \phase{} & & \ctrl{-1} & & \\
& \ctrl{-1} & & \ctrl{-2} & \ctrl{-2} & \ctrl{-2} & & & 
\end{quantikz}
\stackrel{\text{(n+),(e),(i')}}{=}
\begin{quantikz}
& \gate{H} & \phase{} & \phase{} & \gate{H} & \\
& \targ{} & \ctrl{-1} & \octrl{-1} & & \\
& \ctrl{-1} & & \ctrl{-1} &
\end{quantikz}
\stackrel{\text{(H)}}{=}
\begin{quantikz}
& & \targ{} & \targ{} & \\
& \targ{} & \ctrl{-1} & \octrl{-1} & \\
& \ctrl{-1} & & \ctrl{-1} & 
\end{quantikz}
\stackrel{\text{(g),(A.1.4)}}{=}
\begin{quantikz}
& \targ{} & & \targ{} & \\
& \ctrl{-1} & \targ{} & \ctrl{-1} & \\
& \ctrl{-1} & \ctrl{-1} & & 
\end{quantikz}
\end{aligned}
\end{equation*}

\subsubsection{$H$ gadget}

The following $H$ gadget presented in \cite{heyfron_efficient_2018} uses an fSWAP, a measurement in the $H$ basis ($\ket{\pm}$), and a classically controlled $X$.

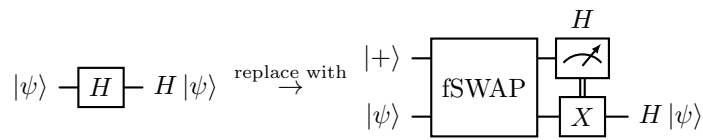
\begin{equation*}
\begin{quantikz}
\lstick{$\ket{\psi}$} & \gate{H} & \rstick{$H\ket{\psi}$}
\end{quantikz}
\stackrel{\text{replace with}}{\rightarrow}
\begin{quantikz}
\lstick{$\ket{+}$} & \gate[2]{\text{fSWAP}} & \meter{H}\vcw{1} \\
\lstick{$\ket{\psi}$} & \ghost{\text{fSWAP}} & \gate{X} & \rstick{$H\ket{\psi}$}
\end{quantikz}
\end{equation*}

If you did not know what this gadget did (i.e. someone gave you the circuit on the right without labelling its output), one way to discern its functionality is to reduce it to a known circuit. Here we show what that looks like by reducing it with EQCDs to just an $H$ gate, which clearly has the action $\ket{\psi} \mapsto H\ket{\psi}$. At every step of the calculation we give the rule being applied.

\begin{equation*}
\begin{aligned}
&\begin{quantikz}
\lstick{$\ket{+}$} & \gate[2]{\text{fSWAP}} & \meter{H}\vcw{1} \\
\lstick{$\ket{\psi}$} & \ghost{\text{fSWAP}} & \gate{X} & 
\end{quantikz}
:=
\begin{quantikz}
\lstick{$\ket{+}$} & \permute{2,1} & \ctrl{1} & \meter{H}\vcw{1} \\
\lstick{$\ket{\psi}$} & & \phase{} & \gate{X} &
\end{quantikz}
\stackrel{\text{(pdm)}}{=}
\begin{quantikz}
\lstick{$\ket{+}$} & \permute{2,1} & \ctrl{1} & \trg{1} & \meter{H} \\
\lstick{$\ket{\psi}$} & & \phase{} & \targ{} & &
\end{quantikz}
\\
&\stackrel{\text{(mb)}}{=}
\begin{quantikz}
\lstick{$\ket{+}$} & \permute{2,1} & \ctrl{1} & \trg{1} & \gate{H} & \meter{} \\
\lstick{$\ket{\psi}$} & & \phase{} & \targ{} & & &
\end{quantikz}
\stackrel{\text{(g),(H)}}{=}
\begin{quantikz}
\lstick{$\ket{+}$} & \permute{2,1} & \ctrl{1} & \gate{H} & \ctrl{1} & \meter{} \\
\lstick{$\ket{\psi}$} & & \phase{} & & \targ{} & &
\end{quantikz}
\\
&\stackrel{\text{(g),(H)}}{=}
\begin{quantikz}
\lstick{$\ket{+}$} & \permute{2,1} & \gate{H} & \trg{1} & \ctrl{1} & \meter{} \\
\lstick{$\ket{\psi}$} & & & \phase{} & \targ{} & &
\end{quantikz}
\stackrel{\text{(t)}}{=}
\begin{quantikz}
\lstick{$\ket{+}$} & & \permute{2,1} & \trg{1} & \ctrl{1} & \meter{} \\
\lstick{$\ket{\psi}$} & \gate{H} & & \phase{} & \targ{} & &
\end{quantikz}
\\
&\stackrel{\text{(s)}}{=}
\begin{quantikz}
\lstick{$\ket{+}$} & & \trg{1} & \ctrl{1} & \trg{1} & \trg{1} & \ctrl{1} & \meter{} \\
\lstick{$\ket{\psi}$} & \gate{H} & \phase{} & \targ{} & \phase{} & \phase{} & \targ{} & &
\end{quantikz}
\stackrel{\text{(i')}}{=}
\begin{quantikz}
\lstick{$\ket{+}$} & & \trg{1} & \ctrl{1} & \ctrl{1} & \meter{} \\
\lstick{$\ket{\psi}$} & \gate{H} & \phase{} & \targ{} &  \targ{} & &
\end{quantikz}
\\
&\stackrel{\text{(i')}}{=}
\begin{quantikz}
\lstick{$\ket{+}$} & & \trg{1} & \meter{} \\
\lstick{$\ket{\psi}$} & \gate{H} & \phase{} & &
\end{quantikz}
\stackrel{(1)}{=}
\begin{quantikz}
\lstick{$\ket{+}$} & & \meter{} \\
\lstick{$\ket{\psi}$} & \gate{H} & &
\end{quantikz}
\stackrel{\text{(remove ancilla)}}{=}
\begin{quantikz}
\lstick{$\ket{\psi}$} & \gate{H} &
\end{quantikz}
\end{aligned}
\end{equation*}

\subsubsection{Section 12.1.1 in the Dodo book}
Section 12.1.1 of the Dodo book \cite{coecke_picturing_2017} showcases a calculation in which the ZX-calculus is used to simplify the following circuit.\footnote{Technically in the Dodo book the top two ancilla qubits are discarded rather than measured, and the bottom two data qubits are measured, but it makes no difference to the circuit analysis.}

\begin{equation*}
\begin{quantikz}
\lstick{$\ket{0}$} & \ctrl{1} & \targ{1/4} & \phase{-1/2} & \targ{1/4} & \ctrl{1} & & \ctrl{1} & & \meter{} \\
\lstick{$\ket{0}$} & \targ{} & & & \ctrl{1} & \targ{} & \phase{-1/4} & \targ{} & \ctrl{1} & \meter{} \\
\lstick{$\ket{0}$} & \ctrl{1} & \targ{-1/4} & \ctrl{1} & \targ{} & \ctrl{1} & \targ{1/4} & \ctrl{1} & \targ{} & 
 \\
\lstick{$\ket{0}$} & \targ{} & & \targ{} & & \targ{} & & \targ{} & & 
\end{quantikz}
\end{equation*}

Regarding the feasibility of doing this simplification graphically, the authors say: ``If we try to keep the quantum gates intact, we're pretty much stuck, but what if we temporarily forget that the chunk in the middle is made up of quantum gates, and just treat it like any other ZX-diagram? That's when some ZX-magic happens!"

With standard circuit techniques, it's indeed true that we're pretty much stuck. But by using EQCDs we can perform the same calculation. We now demonstrate this, step-by-step identical to the calculation shown on p. 682-683 of the Dodo book, but using EQCDs instead of ZX-calculus.

First the authors consider the following pieces of the circuit:

\begin{equation*}
\begin{quantikz}
& \ctrl{1} & \targ{-1/4} & \ctrl{1} &  \\
& \targ{} & & \targ{} &
\end{quantikz}
\quad \quad \quad
\begin{quantikz}
& \ctrl{1} & & \ctrl{1} &  \\
& \targ{} & \phase{-1/4} & \targ{} &
\end{quantikz}
\quad \quad \quad
\begin{quantikz}
& \ctrl{1} & \targ{1/4} & \ctrl{1} &  \\
& \targ{} & & \targ{} &
\end{quantikz}
\end{equation*}

These forms are simplified into alternate (equivalent) ones which are symmetric with respect to the SWAP operation. In ZX-calculus this is done by noticing a 4-cycle and applying the strong complementarity rule. We do the same using the expansion and commutation rules. 

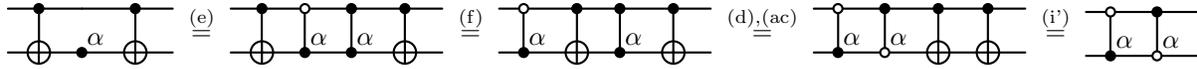
\begin{equation*}
\begin{quantikz}
& \ctrl{1} & & \ctrl{1} &  \\
& \targ{} & \phase{\alpha} & \targ{} &
\end{quantikz}
\stackrel{\text{(e)}}{=}
\begin{quantikz}
& \ctrl{1} & \octrl{1} & \ctrl{1} & \ctrl{1} &  \\
& \targ{} & \phase{\alpha} & \phase{\alpha} & \targ{} &
\end{quantikz}
\stackrel{\text{(f)}}{=}
\begin{quantikz}
& \octrl{1} & \ctrl{1} & \ctrl{1} & \ctrl{1} &  \\
& \phase{\alpha} & \targ{} & \phase{\alpha} & \targ{} &
\end{quantikz}
\stackrel{\text{(d),(ac)}}{=}
\begin{quantikz}
& \octrl{1} & \ctrl{1} & \ctrl{1} & \ctrl{1} &  \\
& \phase{\alpha} & \ophase{\alpha} & \targ{} & \targ{} &
\end{quantikz}
\stackrel{\text{(i')}}{=}
\begin{quantikz}
& \octrl{1} & \ctrl{1} & \\
& \phase{\alpha} & \ophase{\alpha} & 
\end{quantikz}
\end{equation*}

Similarly for the flopped case. We will collectively call these derived rules (z).

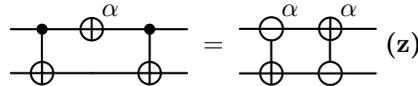
\begin{equation*}
\begin{quantikz}
& \ctrl{1} & \targ{\alpha} & \ctrl{1} &  \\
& \targ{} & & \targ{} &
\end{quantikz}
=
\begin{quantikz}
& \targm{\alpha} & \targ{\alpha} & \\
& \targ{}\vqw{-1} & \targm{}\vqw{-1} & 
\end{quantikz} \textbf{(z)}
\end{equation*}

This completes the steps shown on p. 682 (in slightly more general form, using arbitrary $\alpha$). Next we substitute these identities into the circuit.

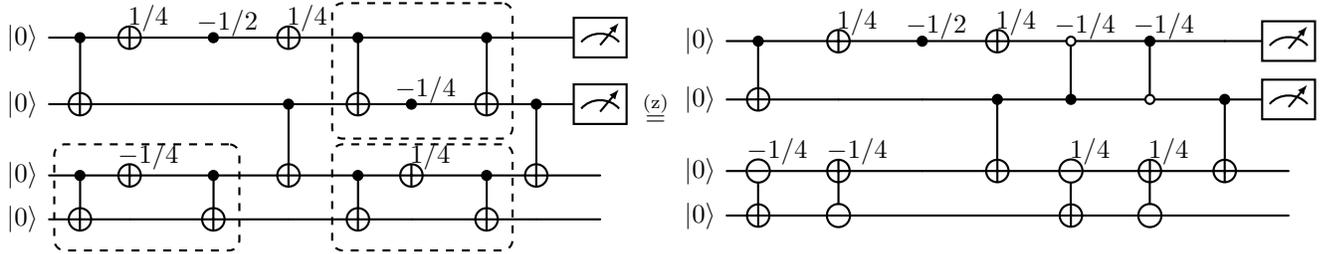
\begin{equation*}
\begin{quantikz}
\lstick{$\ket{0}$} & \ctrl{1} & \targ{1/4} & \phase{-1/2} & \targ{1/4} & \ctrl{1}\gategroup[2,steps=3,style={dashed,rounded corners, inner xsep=2pt, inner ysep=2pt}]{} & & \ctrl{1} & & \meter{} \\
\lstick{$\ket{0}$} & \targ{} & & & \ctrl{1} & \targ{} & \phase{-1/4} & \targ{} & \ctrl{1} & \meter{} \\
\lstick{$\ket{0}$} & \ctrl{1}\gategroup[2,steps=3,style={dashed,rounded corners, inner xsep=2pt}]{} & \targ{-1/4} & \ctrl{1} & \targ{} & \ctrl{1}\gategroup[2,steps=3,style={dashed,rounded corners, inner xsep=2pt}]{} & \targ{1/4} & \ctrl{1} & \targ{} & \\
\lstick{$\ket{0}$} & \targ{} & & \targ{} & & \targ{} & & \targ{} & & 
\end{quantikz}
\stackrel{\text{(z)}}{=}
\begin{quantikz}
\lstick{$\ket{0}$} & \ctrl{1} & \targ{1/4} & \phase{-1/2} & \targ{1/4} & \ophase{-1/4} & \phase{-1/4} & & \meter{} \\
\lstick{$\ket{0}$} & \targ{} & & & \ctrl{1} & \ctrl{-1} & \octrl{-1} & \ctrl{1} & \meter{} \\
\lstick{$\ket{0}$} & \targm{-1/4}\vqw{1} & \targ{-1/4}\vqw{1} & & \targ{} & \targm{1/4}\vqw{1} & \targ{1/4}\vqw{1} & \targ{} & \\
\lstick{$\ket{0}$} & \targ{} & \targm{} & & & \targ{} & \targm{} & &
\end{quantikz}
\end{equation*}

In the Dodo book the authors here note that the top two ancilla wires can be disconnected from the bottom two by applying the spider fusion and complementarity rules to cut off a pair of spider legs. Here those legs are the two CNOTs connecting wires 2 and 3. 

We do the same by observing that the CNOTs in question commute with all four gates separating them, essentially because on every overlapping wire they share symbols of the same type. On wire 2 all four gates in question have $\circ/\bullet$ (Z-type) symbols, and on wire 3 they all have $\ominus/\oplus$ (X-type) symbols. Hence the CNOTs can be freely moved next to each other, after which they cancel, disconnecting the two halves of the circuit.

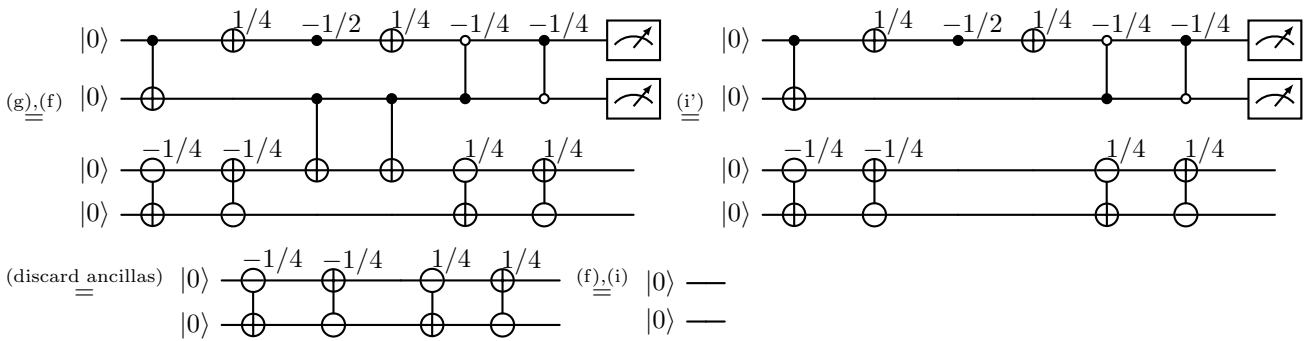
\begin{equation*}
\begin{aligned}
&\stackrel{\text{(g),(f)}}{=}
\begin{quantikz}
\lstick{$\ket{0}$} & \ctrl{1} & \targ{1/4} & \phase{-1/2} & \targ{1/4} & \ophase{-1/4} & \phase{-1/4} & \meter{} \\
\lstick{$\ket{0}$} & \targ{} & & \ctrl{1} & \ctrl{1} & \ctrl{-1} & \octrl{-1} & \meter{} \\
\lstick{$\ket{0}$} & \targm{-1/4}\vqw{1} & \targ{-1/4}\vqw{1} & \targ{} & \targ{} & \targm{1/4}\vqw{1} & \targ{1/4}\vqw{1} & \\
\lstick{$\ket{0}$} & \targ{} & \targm{} & & & \targ{} & \targm{} &
\end{quantikz}
\stackrel{\text{(i')}}{=}
\begin{quantikz}
\lstick{$\ket{0}$} & \ctrl{1} & \targ{1/4} & \phase{-1/2} & \targ{1/4} & \ophase{-1/4} & \phase{-1/4} & \meter{} \\
\lstick{$\ket{0}$} & \targ{} & & & & \ctrl{-1} & \octrl{-1} & \meter{} \\
\lstick{$\ket{0}$} & \targm{-1/4}\vqw{1} & \targ{-1/4}\vqw{1} & & & \targm{1/4}\vqw{1} & \targ{1/4}\vqw{1} & \\
\lstick{$\ket{0}$} & \targ{} & \targm{} & & & \targ{} & \targm{} &
\end{quantikz}
\\
&\stackrel{\text{(discard ancillas)}}{=}
\begin{quantikz}
\lstick{$\ket{0}$} & \targm{-1/4}\vqw{1} & \targ{-1/4}\vqw{1} & & \targm{1/4}\vqw{1} & \targ{1/4}\vqw{1} & \\
\lstick{$\ket{0}$} & \targ{} & \targm{} & & \targ{} & \targm{} & 
\end{quantikz}
\stackrel{\text{(f),(i)}}{=}
\begin{quantikz}
\lstick{$\ket{0}$} & \ghost{H} & \\
\lstick{$\ket{0}$} & & 
\end{quantikz}
\end{aligned}
\end{equation*}

After discarding the two ancilla wires we're left with a pair of operations that are conjugates of one another.\footnote{Technically conjugating would change the order of the two conditional operations in addition to negating their signs, but since all of the symbols in question are of the same type, all four conditional operations commute.} They cancel, completeting the calculation without any ZX-magic.

This example may make the reader wonder if EQCDs are simply ZX-diagrams in disguise; in fact they are not. ZX-calculus is a strictly stronger graphical language in the sense that it can represent arbitrary linear maps, including non-unitary maps that do not correspond to circuits. By design EQCDs can only represent circuits.

\subsection{Diagonalizations}

An interesting feature of EQCDs is that they can be used to diagrammatically diagonalize circuits. This includes individual gates like CNOT or fSWAP, subcircuits built from small sequences of gates (inside a larger circuit), or even full circuits/entire algorithms like those implementing phase estimation, Hamiltonian simulation, or Grover's algoirithm.

Every unitary matrix $U\in U(2^n)$ has the form $U = V\text{diag}(e^{i\pi\lambda_j})V^\dagger$ by the spectral theorem, where $V$ is a matrix whose columns are eigenvectors of $U$, and $\Lambda := \text{diag}(e^{i\pi\lambda_j})$ is a diagonal matrix of their associated eigenvalues. Diagrammatically diagonalizing a circuit for $U$ amounts to finding an explicit circuit decomposition of $U$ into the three components $V^\dagger$, $\Lambda$, and $V$. The result is an explicit circuit for $V$, the basis transformation from the computational basis into the eigenbasis of $U$, an explicit circuit for $\Lambda$ (whose form is always the same since it's a diagonal matrix, but whose eigenvalues are usually not known ahead of time), and an explicit circuit for $V^\dagger$, which is just the mirror image of the circuit for $V$.

$V$ and $V^\dagger$ are conjugate transposes of one another. The conjugate transpose of a circuit $V$ is easily seen to be its right-to-left mirror image. All gates appear in reverse order, and all numbers are negated. That is, picking an arbitrary circuit representative for $V$,

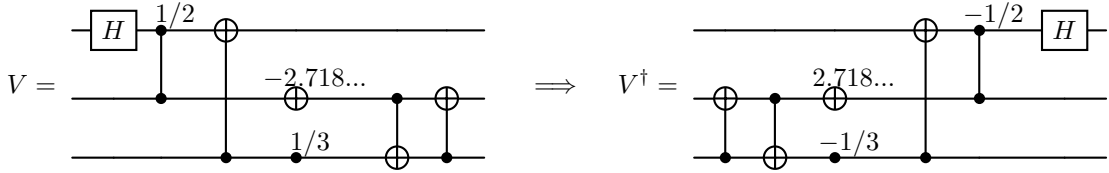
\begin{equation*}
V = 
\begin{quantikz}
& \gate{H} & \phase{1/2} & \targ{} & & & & \\
& & \ctrl{-1} & & \targ{-2.718...} & \ctrl{1} & \targ{} & \\
& & & \ctrl{-2} & \phase{1/3} & \targ{} & \ctrl{-1} &
\end{quantikz}
\quad \implies \quad 
V^\dagger =
\begin{quantikz}
& & & & \targ{} & \phase{-1/2} & \gate{H} & \\
& \targ{} & \ctrl{1} & \targ{2.718...} & & \ctrl{-1} & & \\
& \ctrl{-1} & \targ{} & \phase{-1/3} & \ctrl{-2} & & &
\end{quantikz}
\end{equation*}

It's easily verified that $VV^\dagger = V^\dagger V = \eye$. The eigenvalue matrix $\Lambda = \text{diag}(e^{i\pi\lambda_j})$ has the form

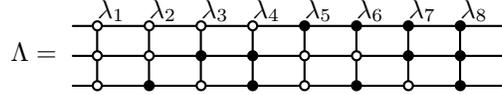
\begin{equation*}
\Lambda = 
\begin{quantikz}
& \ophase{\lambda_1} & \ophase{\lambda_2} & \ophase{\lambda_3} & \ophase{\lambda_4} & \phase{\lambda_5} & \phase{\lambda_6} & \phase{\lambda_7} & \phase{\lambda_8} & \\
& \octrl{-1} & \octrl{-1} & \ctrl{-1} & \ctrl{-1} & \octrl{-1} & \octrl{-1} & \ctrl{-1} & \ctrl{-1} & \\
& \octrl{-1} & \ctrl{-1} & \octrl{-1} & \ctrl{-1} & \octrl{-1} & \ctrl{-1} & \octrl{-1} & \ctrl{-1} & \\
\end{quantikz}
\end{equation*}

where each of the $2^n$ conditionals implements one of $U$'s $2^n$ eigenvalue transformations.

\begin{equation*}
\begin{quantikz}
& \ophase{\lambda_1} & \\
& \octrl{-1} & \\
& \octrl{-1} &
\end{quantikz}
=
e^{i\pi\lambda_1\dyad{000}}
=
\begin{bmatrix}
   e^{i\pi\lambda_1} & 0 &  & 0 \\
   0 & 1 &  & 0 \\
    &  & \ddots & 0 \\
   0 & 0 & 0 & 1 
\end{bmatrix},
\quad \quad
\begin{quantikz}
& \ophase{\lambda_2} & \\
& \octrl{-1} & \\
& \ctrl{-1} &
\end{quantikz}
=
e^{i\pi\lambda_1\dyad{001}}
=
\begin{bmatrix}
   1 & 0 &  & 0 \\
   0 & e^{i\pi\lambda_2} &  & 0 \\
    &  & \ddots & 0 \\
   0 & 0 & 0 & 1 
\end{bmatrix},
\end{equation*}

\begin{equation*}
\dots, \quad \quad
\begin{quantikz}
& \phase{\lambda_8} & \\
& \ctrl{-1} & \\
& \ctrl{-1} &
\end{quantikz}
=
e^{i\pi\lambda_8\dyad{111}}
=
\begin{bmatrix}
   1 & 0 &  & 0 \\
   0 & 1 &  & 0 \\
    &  & \ddots & 0 \\
   0 & 0 & 0 & e^{i\pi\lambda_8} 
\end{bmatrix}.
\end{equation*}

So, given a starting circuit for $U$, diagrammatically diagonalizing $U$ means using the rewrite rules to separate $U$ into three parts: the mirror images $V$ and $V^\dagger$, which can contain any circuit elements as long as they are perfect mirror images of one another, sandwiching a central block $\Lambda$ containing only dots. $\Lambda$ can contain vertical lines and numbers of any amount and in any configuration, since it will always be rewritable into the above form. What is relevant is that it not contain any $\bm{\oplus}$ or $\bm{\ominus}$ symbols, since these are non-diagonal.

In this subsection we give a few simple examples of diagrammatic gate diagonalizations.

\subsubsection{CNOT}

Applying rule (g) to $H$ and a conditional reveals that it flops symbols exactly as if they were not attached to other symbols by a vertical line.

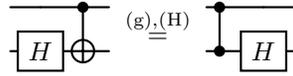
\begin{equation*}
\begin{quantikz}[column sep=0.1cm,row sep=0.25cm]
    & & \ctrl{1} & \\
    & \gate{H} & \targ{} &
\end{quantikz}
\stackrel{\text{(g),(H)}}{=}
\begin{quantikz}[column sep=0.1cm,row sep=0.25cm]
    & \ctrl{1} & & \\
    & \phase{} & \gate{H} &
\end{quantikz}
\end{equation*}

This fact both diagonalizes CNOT and also shows that CNOT and CZ are locally equivalent (i.e. related by single-qubit unitaries).

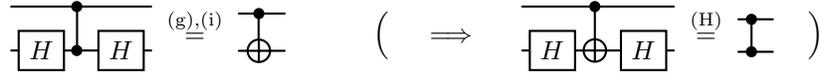
\begin{equation*}
\begin{quantikz}[column sep=0.1cm,row sep=0.25cm]
    & & \ctrl{1} & & \\
    & \gate{H} & \phase{} & \gate{H} &
\end{quantikz}
\stackrel{\text{(g),(i)}}{=}
\begin{quantikz}[column sep=0.1cm,row sep=0.25cm]
    & \ctrl{1} & \\
    & \targ{} & 
\end{quantikz}
\quad \quad \quad \Big(\quad \implies \quad
\begin{quantikz}[column sep=0.1cm,row sep=0.25cm]
    & & \ctrl{1} & & \\
    & \gate{H} & \targ{} & \gate{H} &
\end{quantikz}
\stackrel{\text{(H)}}{=}
\begin{quantikz}[column sep=0.1cm,row sep=0.25cm]
    & \ctrl{1} & \\
    & \phase{} & 
\end{quantikz}
\quad \Big)
\end{equation*}

Explicitly in matrices, this diagonalization is

\begin{equation*}
\begin{bmatrix}
   H & \zero \\
   \zero & H 
\end{bmatrix}
\begin{bmatrix}
   \eye & \zero \\
   \zero & Z 
\end{bmatrix}
\begin{bmatrix}
   H & \zero \\
   \zero & H 
\end{bmatrix}
=
\begin{bmatrix}
   \eye & \zero \\
   \zero & X 
\end{bmatrix}
\end{equation*}

CNOT can also be diagonalized in a `weaker' way by a controlled $H$ rather than $\eye\otimes H$. We call this a weaker diagonalization because it applies only to the subspace where CNOT is not diagonal. These weaker diagonalizations are sometimes useful in larger proofs, as we'll see soon with fSWAP.

\begin{equation*}
\begin{aligned}
\begin{quantikz}[column sep=0.1cm,row sep=0.25cm]
    & \ctrl{1} & \ctrl{1} & \ctrl{1} & \\
    & \gate{H} & \phase{} & \gate{H} &
\end{quantikz}
\stackrel{\text{(d),(H),(i)}}{=}
&\begin{quantikz}[column sep=0.1cm,row sep=0.25cm]
    & \ctrl{1} & \\
    & \targ{} & 
\end{quantikz}
\\
\begin{bmatrix}
   \eye & \zero \\
   \zero & H 
\end{bmatrix}
\begin{bmatrix}
   \eye & \zero \\
   \zero & Z 
\end{bmatrix}
\begin{bmatrix}
   \eye & \zero \\
   \zero & H 
\end{bmatrix}
=
&\begin{bmatrix}
   \eye & \zero \\
   \zero & X 
\end{bmatrix}
\end{aligned}
\end{equation*}

\subsubsection{SWAP}


SWAP is diagonalized by just applying the CNOT diagonalization.

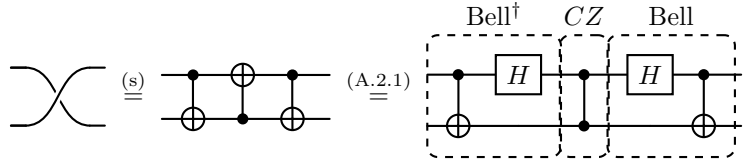
\begin{equation*}
\begin{quantikz}[column sep=0.2cm]
& \permute{2,1} & \ghost{H} \\
& & \ghost{H}
\end{quantikz}
\stackrel{\text{(s)}}{=}
\begin{quantikz}
& \ctrl{1} & \targ{} & \ctrl{1} & \\
& \targ{} & \ctrl{-1} & \targ{} &
\end{quantikz}
\stackrel{\text{(A.2.1)}}{=}
\begin{quantikz}
& \ctrl{1}\gategroup[2,steps=2,style={dashed,rounded corners}]{Bell$^\dagger$} & \gate{H} & & \ctrl{1}\gategroup[2,steps=1,style={dashed,rounded corners}]{$CZ$} & & \gate{H}\gategroup[2,steps=2,style={dashed,rounded corners}]{Bell} & \ctrl{1} & \\
& \targ{} & & & \control{} & & & \targ{} &
\end{quantikz}
\end{equation*}

Looking at the diagram, we recognize the circuit for SWAP's eigenbasis as the Bell basis transformation, whose columns are the Bell states $\{\ket{\Phi_+},\ket{\Psi_+},\ket{\Phi_-},\ket{\Psi_-}\}$, with associated eigenvalues given by the CZ, $\{1,1,1,-1\}$. Since none of these states are separable, SWAP has no semi-classical logic as a controlled operation controlled by one or more wires. If forcibly interpreted as a controlled operation, it instead implements the non-classical logic $(\verb|if | \ket{\Psi_-}\verb| do | -\eye)$, which we call `non-classical' because, in the language of appendix \ref{appendix:onthenotionofcontrol}, it has no partial eigenbases on either qubit.

SWAP also has a weaker diagonalization following from CNOT's weaker diagonalization.

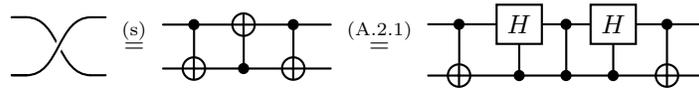
\begin{equation*}
\begin{quantikz}[column sep=0.2cm]
& \permute{2,1} & \ghost{H} \\
& & \ghost{H}
\end{quantikz}
\stackrel{\text{(s)}}{=}
\begin{quantikz}
& \ctrl{1} & \targ{} & \ctrl{1} & \\
& \targ{} & \ctrl{-1} & \targ{} &
\end{quantikz}
\stackrel{\text{(A.2.1)}}{=}
\begin{quantikz}
& \ctrl{1} & \gate{H} & \ctrl{1} & \gate{H} & \ctrl{1} & \\
& \targ{} & \ctrl{-1} & \control{} & \ctrl{-1} & \targ{} &
\end{quantikz}
\end{equation*}

This weaker basis transformation only transforms half of the computational basis into Bell states. Only the odd-parity computational basis states are transformed, leading to the eigenbasis $\{\ket{00},\ket{\Psi_+},\ket{11},\ket{\Psi_-}\}$ with eigenvalues $\{1,1,1,-1\}$. 

\subsubsection{fSWAP}

The fermionic SWAP gets its name from a specific physical context: when two identical fermions are physically swapped, their overall wavefunction picks up a minus sign. If we consider our qubits to be the presence or absence of a fermion (say, an electron), then the natural SWAP gate for this system must include a minus sign when both modes being swapped are occupied, i.e. for the $\ket{11}$ state. This gate is fSWAP.

Being closely related to SWAP, fSWAP can be diagonalized in much the same way. The only extra complication ia a CZ which must be pushed into the center. This is most easily done using the weaker diagonalization of CNOT shown above.

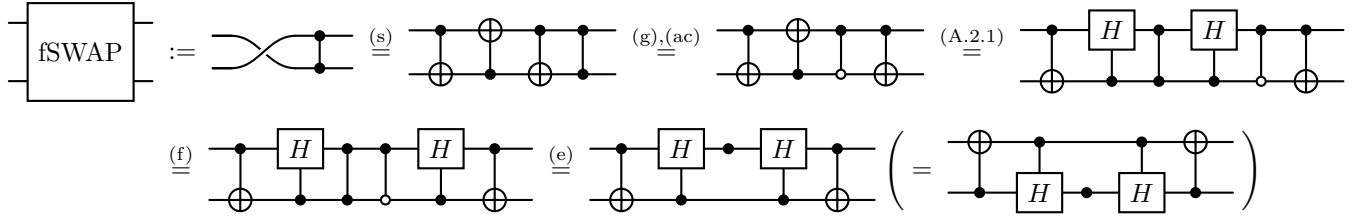
\begin{equation*}
\begin{aligned}
\begin{quantikz}
& \gate[2]{\text{fSWAP}} & \\
& & 
\end{quantikz}
&:=
\begin{quantikz}
& \permute{2,1} & \ctrl{1} & \\
& & \phase{} &
\end{quantikz}
\stackrel{\text{(s)}}{=}
\begin{quantikz}
& \ctrl{1} & \targ{} & \ctrl{1} & \ctrl{1} & \\
& \targ{} & \ctrl{-1} & \targ{} & \phase{} &
\end{quantikz}
\stackrel{\text{(g),(ac)}}{=}
\begin{quantikz}
& \ctrl{1} & \targ{} & \ctrl{1} & \ctrl{1} & \\
& \targ{} & \ctrl{-1} & \ophase{} & \targ{} &
\end{quantikz}
\stackrel{\text{(A.2.1)}}{=}
\begin{quantikz}
& \ctrl{1} & \gate{H} & \phase{} & \gate{H} & \ctrl{1} & \ctrl{1} & \\
& \targ{} & \ctrl{-1} & \ctrl{-1} & \ctrl{-1} & \ophase{} & \targ{} &
\end{quantikz}
\\
&\stackrel{\text{(f)}}{=}
\begin{quantikz}
& \ctrl{1} & \gate{H} & \phase{} & \ctrl{1} & \gate{H} & \ctrl{1} & \\
& \targ{} & \ctrl{-1} & \ctrl{-1} & \ophase{} & \ctrl{-1} & \targ{} &
\end{quantikz}
\stackrel{\text{(e)}}{=}
\begin{quantikz}
& \ctrl{1} & \gate{H} & \phase{} & \gate{H} & \ctrl{1} & \\
& \targ{} & \ctrl{-1} & & \ctrl{-1} & \targ{} &
\end{quantikz}
\Bigg(=
\begin{quantikz}
& \targ{} & \ctrl{1} & & \ctrl{1} & \targ{} & \\
& \ctrl{-1} & \gate{H} & \phase{} & \gate{H} & \ctrl{-1} &
\end{quantikz}\Bigg)
\end{aligned}
\end{equation*}

Inspecting the diagram, fSWAP's eigenbasis is the same as the weaker diagonalization of SWAP, though this time not by choice. It is again half of a Bell basis $\{\ket{00},\ket{\Psi_+},\ket{11},\ket{\Psi_-}\}$, but this time with one extra altered eigenvalue, $\{1,1,-1,-1\}$. Like SWAP, fSWAP has no semi-classical logic, and instead implements the non-classical logic

\begin{equation*}
\begin{array}{l}
\verb|if | \ket{\Psi_-} \\
\verb| do | -\eye \\
\verb|elif | \ket{11} \\
\verb| do | -\eye \\
\end{array} 
\quad\Big(\quad = \begin{array}{l}
\verb|if | \ket{\Psi_-} \verb| xor | \ket{11} \\
\verb| do | -\eye
\end{array}\Big).
\end{equation*}

fSWAP has another form which will useful in section A.3.3.

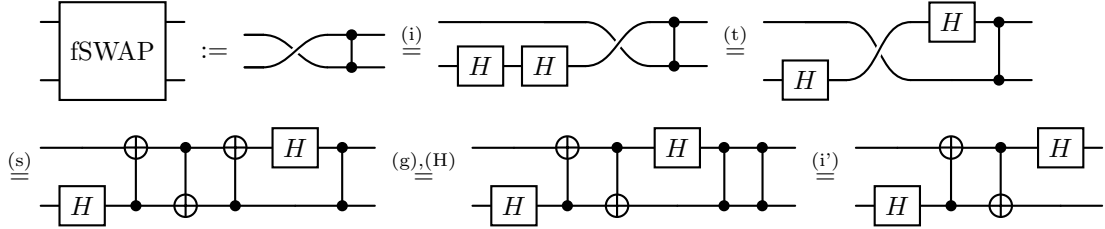
\begin{equation*}
\begin{aligned}
&\begin{quantikz}
& \gate[2]{\text{fSWAP}} & \\
& & 
\end{quantikz}
:=
\begin{quantikz}
& \permute{2,1} & \ctrl{1} & \\
& & \phase{} &
\end{quantikz}
\stackrel{\text{(i)}}{=}
\begin{quantikz}
& & & \permute{2,1} & \ctrl{1} & \\
& \gate{H} & \gate{H} & & \phase{} &
\end{quantikz}
\stackrel{\text{(t)}}{=}
\begin{quantikz}
& & \permute{2,1} & \gate{H} & \ctrl{1} & \\
& \gate{H} & & & \phase{} & 
\end{quantikz}
\\
\stackrel{\text{(s)}}{=}
&\begin{quantikz}
& & \targ{} & \ctrl{1} & \targ{} & \gate{H} & \ctrl{1} & \\
& \gate{H}  & \ctrl{-1} & \targ{} & \ctrl{-1} & & \phase{} &
\end{quantikz}
\stackrel{\text{(g),(H)}}{=}
\begin{quantikz}
& & \targ{} & \ctrl{1} & \gate{H} & \phase{} & \ctrl{1} & \\
& \gate{H}  & \ctrl{-1} & \targ{} & & \ctrl{-1} & \phase{} &
\end{quantikz}
\stackrel{\text{(i')}}{=}
\begin{quantikz}
& & \targ{} & \ctrl{1} & \gate{H} & \\
& \gate{H}  & \ctrl{-1} & \targ{} & &
\end{quantikz}
\end{aligned}
\end{equation*}

\subsubsection{iSWAP}

The imaginary SWAP is locally equivalent to fSWAP, so we will begin from there to diagonalize it. First we manipulate the diagram into a nearly symmetric form involving a SWAP sandwiched by dot-only conditionals, and then we push the dots toward the center of the circuit. \footnote{We thank Shao-Hen Chiew for bringing iSWAP to our attention by asking for a diagonalization of it in conversation at QIP2024.}

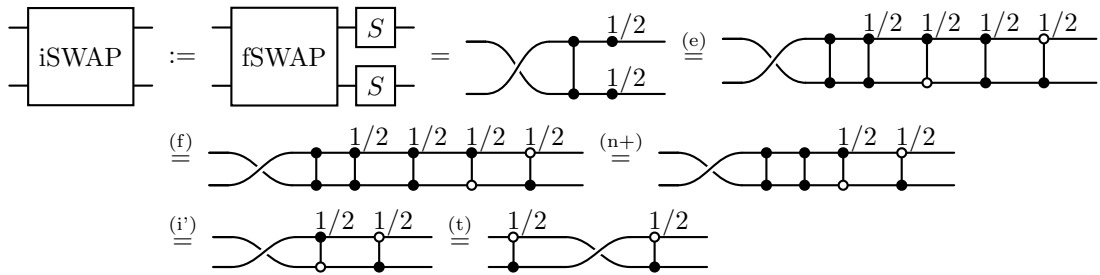
\begin{equation*}
\begin{aligned}
\begin{quantikz}
& \gate[2]{\text{iSWAP}} & \\
& & 
\end{quantikz}
&:=
\begin{quantikz}
& \gate[2]{\text{fSWAP}} & \gate{S} & \\
& & \gate{S} &
\end{quantikz}
=
\begin{quantikz}
& \permute{2,1} & \ctrl{1} & \phase{1/2} & \\
& & \phase{} & \phase{1/2} &
\end{quantikz}
\stackrel{\text{(e)}}{=}
\begin{quantikz}
& \permute{2,1} & \ctrl{1} & \phase{1/2} & \phase{1/2} & \phase{1/2} & \ophase{1/2} & \\
& & \phase{} & \ctrl{-1} & \octrl{-1} & \ctrl{-1} & \ctrl{-1} & \ghost{H}
\end{quantikz}
\\
&\stackrel{\text{(f)}}{=}
\begin{quantikz}
& \permute{2,1} & \ctrl{1} & \phase{1/2} & \phase{1/2} & \phase{1/2} & \ophase{1/2} & \\
& & \phase{} & \ctrl{-1} & \ctrl{-1} & \octrl{-1} & \ctrl{-1} &
\end{quantikz}
\stackrel{\text{(n+)}}{=}
\begin{quantikz}
& \permute{2,1} & \ctrl{1} & \ctrl{1} & \phase{1/2} & \ophase{1/2} & \\
& & \phase{} & \phase{} & \octrl{-1} & \ctrl{-1} &
\end{quantikz}
\\
&\stackrel{\text{(i')}}{=}
\begin{quantikz}
& \permute{2,1} & \phase{1/2} & \ophase{1/2} & \\
& & \octrl{-1} & \ctrl{-1} &
\end{quantikz}
\stackrel{\text{(t)}}{=}
\begin{quantikz}
& \ophase{1/2} & \permute{2,1}  & \ophase{1/2} & \\
& \ctrl{-1} & & \ctrl{-1} &
\end{quantikz}
\end{aligned}
\end{equation*}

Now we insert the diagonalization of SWAP and push the new terms toward the center.

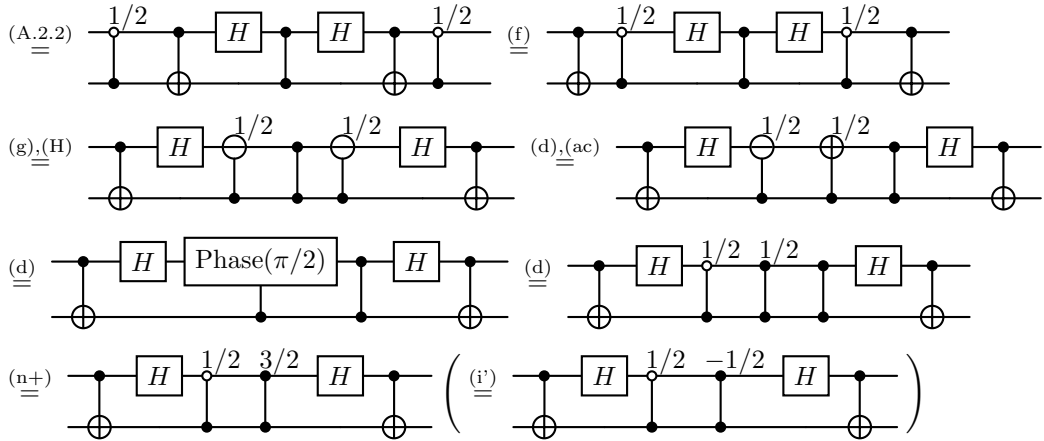
\begin{equation*}
\begin{aligned}
&\stackrel{\text{(A.2.2)}}{=}
\begin{quantikz}
& \ophase{1/2} & \ctrl{1} & \gate{H} & \ctrl{1} & \gate{H} & \ctrl{1} & \ophase{1/2} & \\
& \ctrl{-1} & \targ{} &  & \phase{} &  & \targ{} & \ctrl{-1} &
\end{quantikz}
\stackrel{\text{(f)}}{=}
\begin{quantikz}
& \ctrl{1} & \ophase{1/2} & \gate{H} & \ctrl{1} & \gate{H} & \ophase{1/2} & \ctrl{1} & \\
& \targ{} & \ctrl{-1} &  & \phase{} &  & \ctrl{-1} & \targ{} & 
\end{quantikz}
\\
&\stackrel{\text{(g),(H)}}{=}
\begin{quantikz}
& \ctrl{1} & \gate{H} & \targm{1/2} & \ctrl{1} & \targm{1/2} & \gate{H} & \ctrl{1} & \\
& \targ{} &  & \ctrl{-1} & \phase{} & \ctrl{-1} &  & \targ{} & 
\end{quantikz}
\stackrel{\text{(d),(ac)}}{=}
\begin{quantikz}
& \ctrl{1} & \gate{H} & \targm{1/2} & \targ{1/2} & \ctrl{1} & \gate{H} & \ctrl{1} & \\
& \targ{} &  & \ctrl{-1} & \ctrl{-1} & \phase{} &  & \targ{} & 
\end{quantikz}
\\
&\stackrel{\text{(d)}}{=}
\begin{quantikz}
& \ctrl{1} & \gate{H} & \gate{\text{Phase}(\pi/2)} & \ctrl{1} & \gate{H} & \ctrl{1} & \\
& \targ{} &  & \ctrl{-1} & \phase{} &  & \targ{} & 
\end{quantikz}
\stackrel{\text{(d)}}{=}
\begin{quantikz}
& \ctrl{1} & \gate{H} & \ophase{1/2} & \phase{1/2} & \ctrl{1} & \gate{H} & \ctrl{1} & \\
& \targ{} &  & \ctrl{-1} & \ctrl{-1} & \phase{} &  & \targ{} & 
\end{quantikz}
\\
&\stackrel{\text{(n+)}}{=}
\begin{quantikz}
& \ctrl{1} & \gate{H} & \ophase{1/2} & \phase{3/2} & \gate{H} & \ctrl{1} & \\
& \targ{} &  & \ctrl{-1} & \ctrl{-1} &  & \targ{} & 
\end{quantikz}
\Bigg(\stackrel{\text{(i')}}{=}
\begin{quantikz}
& \ctrl{1} & \gate{H} & \ophase{1/2} & \phase{-1/2} & \gate{H} & \ctrl{1} & \\
& \targ{} &  & \ctrl{-1} & \ctrl{-1} &  & \targ{} & 
\end{quantikz}\Bigg)
\end{aligned}
\end{equation*}

Like SWAP, iSWAP's eigenbasis is the Bell basis. It implements the non-classical logic 

\begin{equation*}
\begin{array}{l}
\verb|if | \ket{\Psi_+} \\
\verb| do | i\eye \\
\verb|elif | \ket{\Psi_-} \\
\verb| do | -i\eye \\
\end{array}.
\end{equation*}

\subsubsection{QFT$_2$}

The quantum fourier transform on 2 qubits involves a SWAP at the end, and can be diagonalized using a diagonalization of SWAP and some pushing of leftover dots toward the center. Like with the fSWAP, using the weaker diagonalization makes the process easier. 

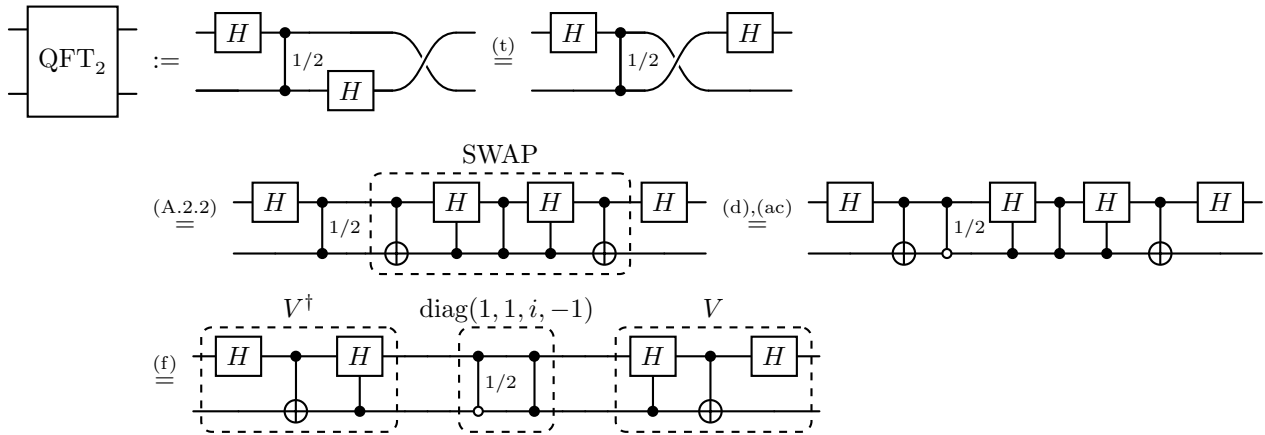
\begin{equation*}
\begin{aligned}
\begin{quantikz}
& \gate[2]{\text{QFT}_2} & \\
& & 
\end{quantikz}
&:=
\begin{quantikz}
& \gate{H} & \ctrl[wire style={"1/2"}]{1} & & & \permute{2,1} & \\
& \qw & \ctrl{-1} & & \gate{H} & &
\end{quantikz}
\stackrel{\text{(t)}}{=}
\begin{quantikz}
& \gate{H} & \ctrl[wire style={"1/2"}]{1} & \permute{2,1} & \gate{H} & \\
& \ghost{H} & \ctrl{-1} & & &
\end{quantikz}
\\
&\stackrel{\text{(A.2.2)}}{=}
\begin{quantikz}
& \gate{H} & \ctrl[wire style={"1/2"}]{1} & & & \ctrl{1}\gategroup[2,steps=5,style={dashed,rounded
corners, inner
xsep=2pt, inner
ysep=0pt}]{SWAP} & \gate{H} & \ctrl{1} & \gate{H} & \ctrl{1} & \gate{H} & \\
& & \control{} & & & \targ{} & \ctrl{-1} & \control{} & \ctrl{-1} & \targ{} & &
\end{quantikz}
\stackrel{\text{(d),(ac)}}{=}
\begin{quantikz}
& \gate{H} & \ctrl{1} & \ctrl[wire style={"1/2"}]{1} & & \gate{H} & \ctrl{1} & \gate{H} & \ctrl{1} & \gate{H} & \\
& & \targ{} & \ocontrol{} & & \ctrl{-1} & \control{} & \ctrl{-1} & \targ{} & &
\end{quantikz}
\\
&\stackrel{\text{(f)}}{=}
\begin{quantikz}[column sep=0.3cm,row sep=0.3cm]
& \gate{H}\gategroup[2,steps=3,style={dashed,rounded
corners, inner
xsep=2pt, inner
ysep=0pt}]{$V^\dagger$} & \ctrl{1} & \gate{H} & & & & \ctrl[wire style={"1/2"}]{1}\gategroup[2,steps=3,style={dashed,rounded
corners, inner
xsep=2pt, inner
ysep=0pt}]{diag$(1,1,i,-1)$} & & \control{} & & & & \gate{H}\gategroup[2,steps=3,style={dashed,rounded
corners, inner
xsep=2pt, inner
ysep=0pt}]{$V$} & \ctrl{1} & \gate{H} & \\
& & \targ{} & \ctrl{-1} & & & & \ocontrol{} & & \ctrl{-1} & & & & \ctrl{-1} & \targ{} & &
\end{quantikz}
\end{aligned}
\end{equation*}

\subsection{Communication primitives}
In this section we provide diagrammatic proofs of various communication primitives.

\subsubsection{Teleportation}

Teleportation is a protocol with two ancilla qubits (shown here initialized to $\ket{0}$), some unitary gates, two measurements, and two classical controlled corrections. The result of the protocol is that Alice's state $\ket{\psi}$ is deterministically teleported to Bob's qubit. This proof of the protocol's functionality proceeds by simplifying the circuit until the magic is gone and it's just an identity wire from Alice to Bob.

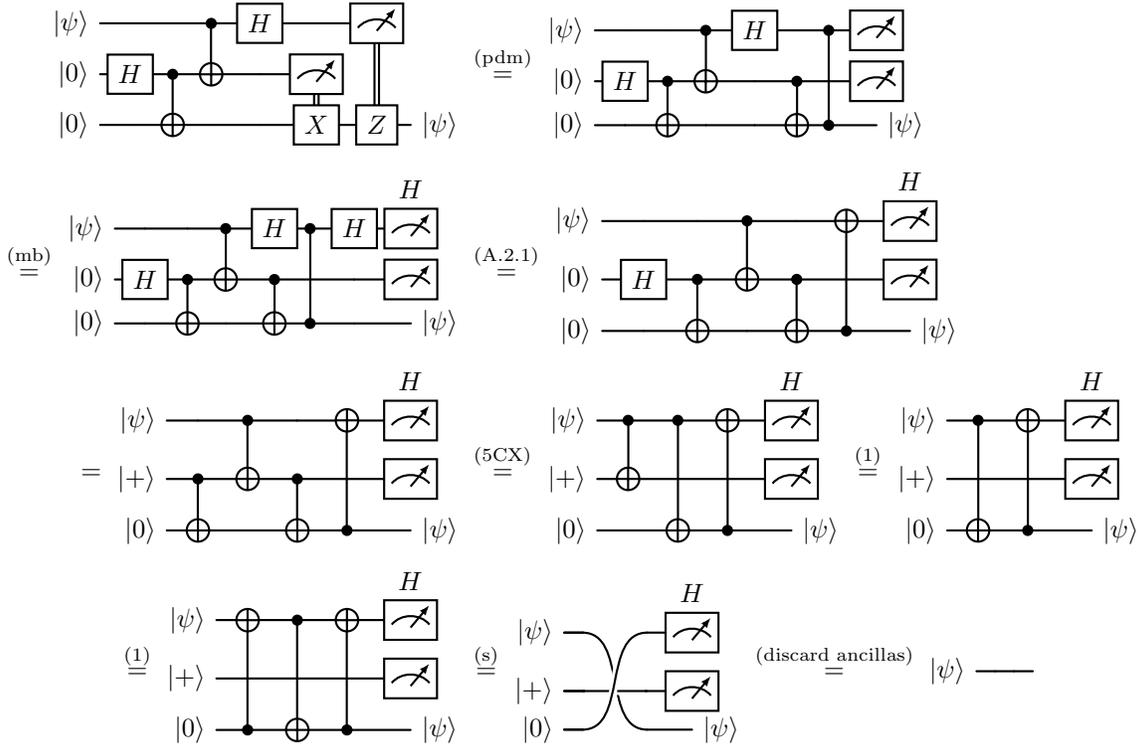
\begin{equation*}
\begin{aligned}
\begin{quantikz}[column sep=0.1cm,row sep=0.15cm]
\lstick{$\ket{\psi}$} & & & \ctrl{1} & \gate{H} & & \meter{} \vcw{2} \\
\lstick{$\ket{0}$} & \gate{H} & \ctrl{1} & \targ{} & & \meter{} \vcw{1} \\
\lstick{$\ket{0}$} & & \targ{} & & & \gate{X} & \gate{Z} & \rstick{$\ket{\psi}$}
\end{quantikz}
&\stackrel{\text{(pdm)}}{=}
\begin{quantikz}[column sep=0.1cm,row sep=0.15cm]
\lstick{$\ket{\psi}$} & & & \ctrl{1} & \gate{H} & & \phase{} & \meter{} \\
\lstick{$\ket{0}$} & \gate{H} & \ctrl{1} & \targ{} & & \ctrl{1} & & \meter{} \\
\lstick{$\ket{0}$} & & \targ{} & & & \targ{} & \ctrl{-2} & \rstick{$\ket{\psi}$}
\end{quantikz}
\\
\stackrel{\text{(mb)}}{=}
\begin{quantikz}[column sep=0.1cm,row sep=0.15cm]
\lstick{$\ket{\psi}$} & & & \ctrl{1} & \gate{H} & \phase{} & \gate{H} & \meter{H} \\
\lstick{$\ket{0}$} & \gate{H} & \ctrl{1} & \targ{} & \ctrl{1} & & & \meter{} \\
\lstick{$\ket{0}$} & & \targ{} & & \targ{} & \ctrl{-2} & & \rstick{$\ket{\psi}$}
\end{quantikz}
&\stackrel{\text{(A.2.1)}}{=}
\begin{quantikz}
\lstick{$\ket{\psi}$} & & & \ctrl{1} & & \targ{} & \meter{H} \\
\lstick{$\ket{0}$} & \gate{H} & \ctrl{1} & \targ{} & \ctrl{1} & & \meter{} \\
\lstick{$\ket{0}$} & & \targ{} & & \targ{} & \ctrl{-2} & \rstick{$\ket{\psi}$}
\end{quantikz}
\\
=
\begin{quantikz}
\lstick{$\ket{\psi}$} & & \ctrl{1} & & \targ{} & \meter{H} \\
\lstick{$\ket{+}$} & \ctrl{1} & \targ{} & \ctrl{1} & & \meter{} \\
\lstick{$\ket{0}$} & \targ{} & & \targ{} & \ctrl{-2} & \rstick{$\ket{\psi}$}
\end{quantikz}
&\stackrel{\text{(5CX)}}{=}
\begin{quantikz}
\lstick{$\ket{\psi}$} & \ctrl{1} & \ctrl{2} & \targ{} & \meter{H} \\
\lstick{$\ket{+}$} & \targ{} & & & \meter{} \\
\lstick{$\ket{0}$} & & \targ{} & \ctrl{-2} & \rstick{$\ket{\psi}$}
\end{quantikz}
\stackrel{(1)}{=}
\begin{quantikz}
\lstick{$\ket{\psi}$} & \ctrl{2} & \targ{} & \meter{H} \\
\lstick{$\ket{+}$} & & & \meter{} \\
\lstick{$\ket{0}$} & \targ{} & \ctrl{-2} & \rstick{$\ket{\psi}$}
\end{quantikz}
\\
\stackrel{(1)}{=}
\begin{quantikz}
\lstick{$\ket{\psi}$} & \targ{} & \ctrl{2} & \targ{} & \meter{H} \\
\lstick{$\ket{+}$} & & & & \meter{} \\
\lstick{$\ket{0}$} & \ctrl{-2} & \targ{} & \ctrl{-2} & \rstick{$\ket{\psi}$}
\end{quantikz}
&\stackrel{\text{(s)}}{=}
\begin{quantikz}
\lstick{$\ket{\psi}$} & \permute{3,2,1} & \meter{H} \\
\lstick{$\ket{+}$} & \ghost{H} & \meter{} \\
\lstick{$\ket{0}$} & & \rstick{$\ket{\psi}$}
\end{quantikz}
\stackrel{\text{(discard ancillas)}}{=}
\begin{quantikz}
\lstick{$\ket{\psi}$} & & &
\end{quantikz}
\end{aligned}
\end{equation*}

\subsubsection{Superdense coding}
Superdense coding is a protocol where Alice can send two classical bits a,b to Bob by sending only one qubit, by using a pre-established Bell pair. Alice does this by performing classically controlled $X$ and $Z$ operations on the qubit before sending it to Bob, after which he may measure in the Bell basis to recover a,b. Our first step is to `quantize' the classical bits a,b by replacing them with qubits in the computational basis states $\ket{a},\ket{b}$ which Alice measures in order to perform the classically controlled $X$ and $Z$ operations.

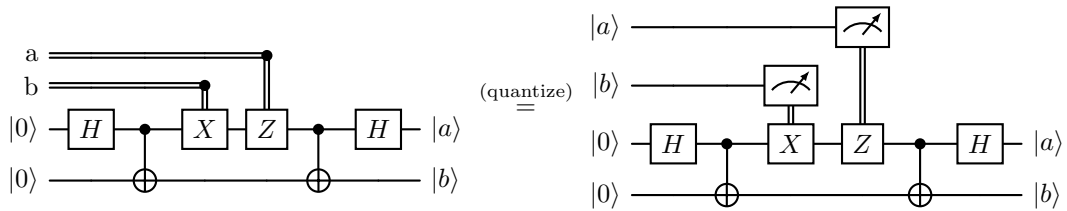
\begin{equation*}
\begin{quantikz}[wire types={c,c,q,q}]
\lstick{a} & & & &  \ctrl[vertical
wire=c]{2}  \\
\lstick{b} & & &  \ctrl[vertical
wire=c]{1}  \\
\lstick{$\ket{0}$} & \gate{H} & \ctrl{1} & \gate{X} & \gate{Z} & \ctrl{1} & \gate{H} & \rstick{$\ket{a}$} \\
\lstick{$\ket{0}$} & & \targ{} & & & \targ{} & & \rstick{$\ket{b}$}
\end{quantikz}
\stackrel{\text{(quantize)}}{=}
\begin{quantikz}
\lstick{$\ket{a}$} & & & & \meter{} \vcw{2} \\
\lstick{$\ket{b}$} & & & \meter{} \vcw{1} \\
\lstick{$\ket{0}$} & \gate{H} & \ctrl{1} & \gate{X} & \gate{Z} & \ctrl{1} & \gate{H} & \rstick{$\ket{a}$} \\
\lstick{$\ket{0}$} & & \targ{} & & & \targ{} & & \rstick{$\ket{b}$}
\end{quantikz}
\end{equation*}

Now we simplify the circuit until what remains is an embedded classical copy circuit, which copies a,b to the lowest two qubits.

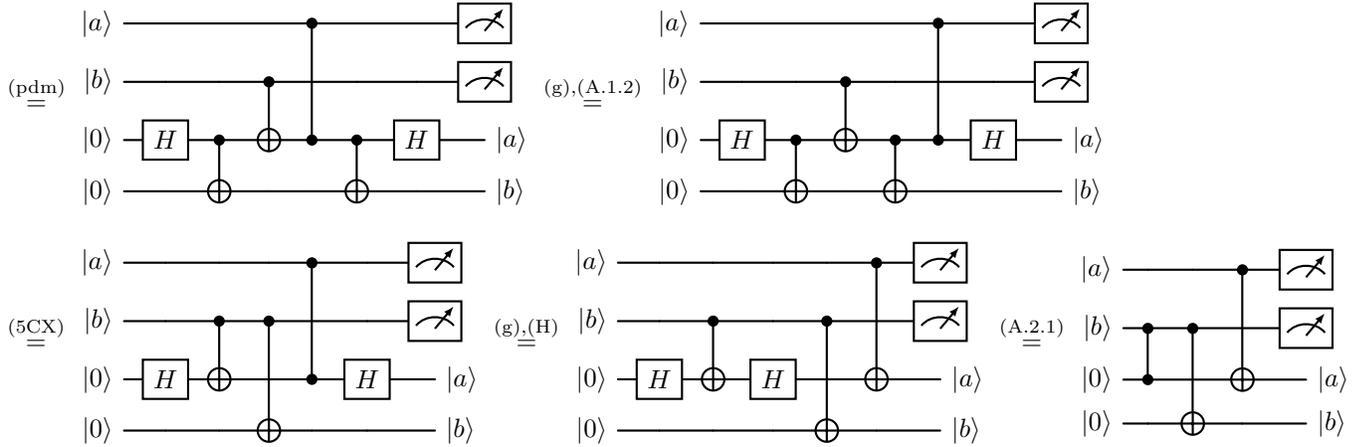
\begin{equation*}
\begin{aligned}
&\stackrel{\text{(pdm)}}{=}
\begin{quantikz}
\lstick{$\ket{a}$} & & & & \ctrl{2} & & & \meter{} \\
\lstick{$\ket{b}$} & & & \ctrl{1} & & & & \meter{} \\
\lstick{$\ket{0}$} & \gate{H} & \ctrl{1} & \targ{} & \phase{} & \ctrl{1} & \gate{H} & \rstick{$\ket{a}$} \\
\lstick{$\ket{0}$} & & \targ{} & & & \targ{} & & \rstick{$\ket{b}$}
\end{quantikz}
\stackrel{\text{(g),(A.1.2)}}{=}
\begin{quantikz}
\lstick{$\ket{a}$} & & & & & \ctrl{2} & & \meter{} \\
\lstick{$\ket{b}$} & & & \ctrl{1} & & & & \meter{} \\
\lstick{$\ket{0}$} & \gate{H} & \ctrl{1} & \targ{} & \ctrl{1} & \phase{} & \gate{H} & \rstick{$\ket{a}$} \\
\lstick{$\ket{0}$} & & \targ{} & & \targ{} & & & \rstick{$\ket{b}$}
\end{quantikz}
\\
&\stackrel{\text{(5CX)}}{=}
\begin{quantikz}
\lstick{$\ket{a}$} & & & & \ctrl{2} & & \meter{} \\
\lstick{$\ket{b}$} & & \ctrl{1} & \ctrl{2} & & & \meter{} \\
\lstick{$\ket{0}$} & \gate{H} & \targ{} & & \phase{} & \gate{H} & \rstick{$\ket{a}$} \\
\lstick{$\ket{0}$} & & & \targ{} & & & \rstick{$\ket{b}$}
\end{quantikz}
\stackrel{\text{(g),(H)}}{=}
\begin{quantikz}
\lstick{$\ket{a}$} & & & & & \ctrl{2} & \meter{} \\
\lstick{$\ket{b}$} & & \ctrl{1} & & \ctrl{2} & & \meter{} \\
\lstick{$\ket{0}$} & \gate{H} & \targ{} & \gate{H} & & \targ{} & \rstick{$\ket{a}$} \\
\lstick{$\ket{0}$} & & & & \targ{} & & \rstick{$\ket{b}$}
\end{quantikz}
\stackrel{\text{(A.2.1)}}{=}
\begin{quantikz}
\lstick{$\ket{a}$} & & & \ctrl{2} & \meter{} \\
\lstick{$\ket{b}$} & \ctrl{1} & \ctrl{2} & & \meter{} \\
\lstick{$\ket{0}$} & \phase{} & & \targ{} & \rstick{$\ket{a}$} \\
\lstick{$\ket{0}$} & & \targ{} & & \rstick{$\ket{b}$}
\end{quantikz}
\end{aligned}
\end{equation*}

The final step is the recognize that the CZ will not activate since one of its inputs is $\ket{0}$. It can simply be deleted, leaving behind a classical copy circuit.

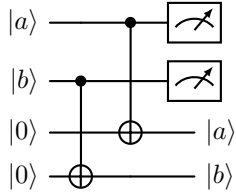
\begin{equation*}
\begin{aligned}
\begin{quantikz}
\lstick{$\ket{a}$} & & \ctrl{2} & \meter{} \\
\lstick{$\ket{b}$} & \ctrl{2} & & \meter{} \\
\lstick{$\ket{0}$} & & \targ{} & \rstick{$\ket{a}$} \\
\lstick{$\ket{0}$} & \targ{} & & \rstick{$\ket{b}$}
\end{quantikz}
\end{aligned}
\end{equation*}

We remark that we've left the output qubits $\ket{a},\ket{b}$ labelled throughout the calculation only for the reader's convenience; the proof does not rely on knowledge of them ahead of time, and of course if one were proving this protocol's functionality with this method they would not be known ahead of time; they would instead be discovered at the end of the calculation, and retroactively written down on the output of every prior step.

\subsubsection{Entanglement swapping}

The last in this series of communication primitive proofs, entanglement swapping is a protocol where entanglement can be established between Alice and Bob by performing a Bell measurement on half of a Bell pair produced by Alice and half of a Bell pair produced by Bob, followed by some classical corrections. The result is a Bell state $\ket{\Phi_+}$ shared by Alice and Bob, which this time we will not label until the end of the proof, when we `discover' its functionality.

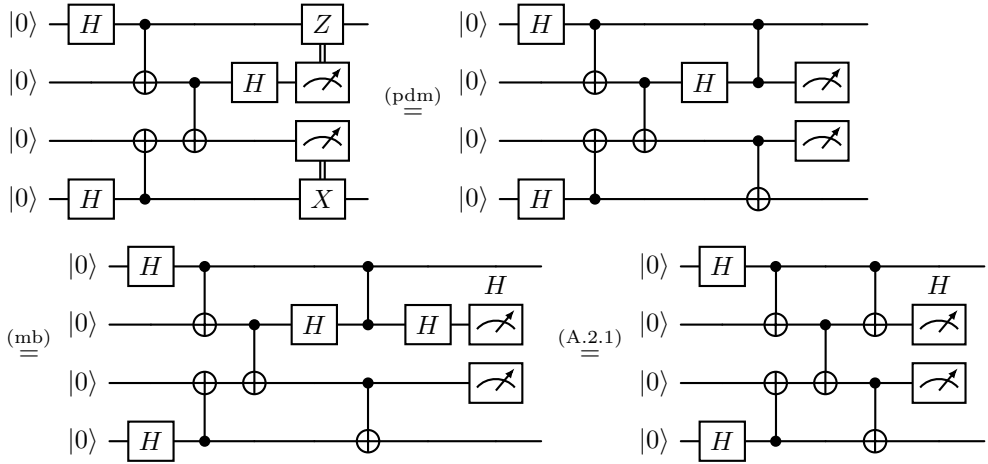
\begin{equation*}
\begin{aligned}
&\begin{quantikz}
\lstick{$\ket{0}$} & \gate{H} & \ctrl{1} & & & \gate{Z} &  \\
\lstick{$\ket{0}$} & & \targ{} & \ctrl{1} & \gate{H} & \meter{} \vcw{-1} \\
\lstick{$\ket{0}$} & & \targ{} & \targ{} & & \meter{} \vcw{1} \\
\lstick{$\ket{0}$} & \gate{H} & \ctrl{-1} & & & \gate{X} & 
\end{quantikz}
\stackrel{\text{(pdm)}}{=}
\begin{quantikz}
\lstick{$\ket{0}$} & \gate{H} & \ctrl{1} & & & \ctrl{1} & & \\
\lstick{$\ket{0}$} & & \targ{} & \ctrl{1} & \gate{H} & \phase{} & \meter{} \\
\lstick{$\ket{0}$} & & \targ{} & \targ{} & & \ctrl{1} & \meter{} \\
\lstick{$\ket{0}$} & \gate{H} & \ctrl{-1} & & & \targ{} & &
\end{quantikz}
\\
&\stackrel{\text{(mb)}}{=}
\begin{quantikz}
\lstick{$\ket{0}$} & \gate{H} & \ctrl{1} & & & \ctrl{1} & & & \\
\lstick{$\ket{0}$} & & \targ{} & \ctrl{1} & \gate{H} & \phase{} & \gate{H} & \meter{H} \\
\lstick{$\ket{0}$} & & \targ{} & \targ{} & & \ctrl{1} & & \meter{} \\
\lstick{$\ket{0}$} & \gate{H} & \ctrl{-1} & & & \targ{} & & & 
\end{quantikz}
\stackrel{\text{(A.2.1)}}{=}
\begin{quantikz}
\lstick{$\ket{0}$} & \gate{H} & \ctrl{1} & & \ctrl{1} & & \\
\lstick{$\ket{0}$} & & \targ{} & \ctrl{1} & \targ{} & \meter{H} \\
\lstick{$\ket{0}$} & & \targ{} & \targ{} & \ctrl{1} & \meter{} \\
\lstick{$\ket{0}$} & \gate{H} & \ctrl{-1} & & \targ{} & &
\end{quantikz}
\end{aligned}
\end{equation*}

At this step we can introduce a very convenient CNOT for free, which does not activate since its control ancilla is in the state $\ket{0}$, and which allows us to use rule (5CX) to replace a series of four CNOTs with just one.

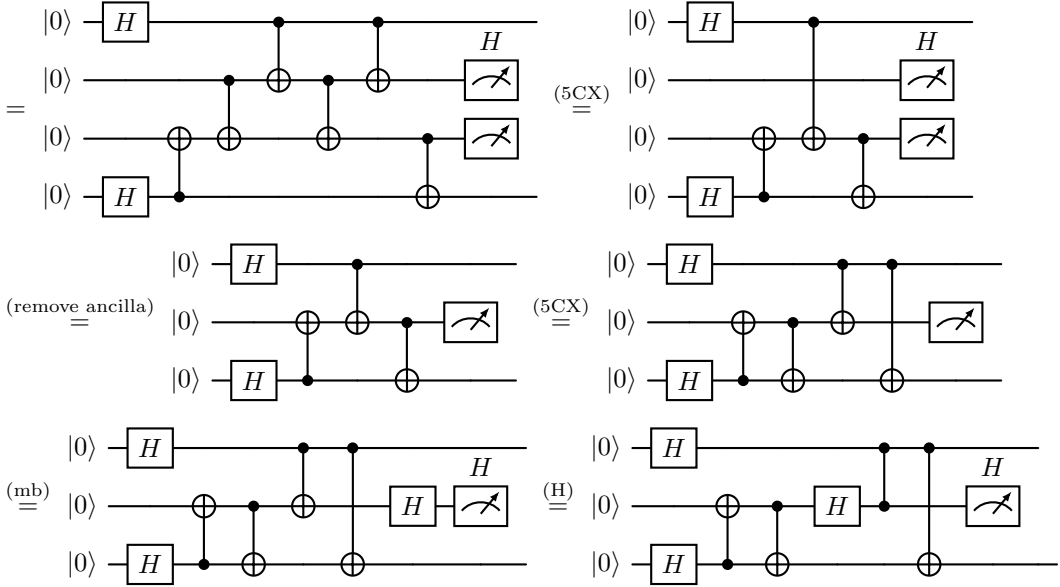
\begin{equation*}
\begin{aligned}
&=
\begin{quantikz}
\lstick{$\ket{0}$} & \gate{H} & & & \ctrl{1} & & \ctrl{1} & & & \\
\lstick{$\ket{0}$} & & & \ctrl{1} & \targ{} & \ctrl{1} & \targ{} & & \meter{H} \\
\lstick{$\ket{0}$} & & \targ{} & \targ{} & & \targ{} & & \ctrl{1} & \meter{} \\
\lstick{$\ket{0}$} & \gate{H} & \ctrl{-1} & & & & & \targ{} & & 
\end{quantikz}
\stackrel{\text{(5CX)}}{=}
\begin{quantikz}
\lstick{$\ket{0}$} & \gate{H} & & \ctrl{2} & & & \\
\lstick{$\ket{0}$} & & & & & \meter{H} \\
\lstick{$\ket{0}$} & & \targ{} & \targ{} & \ctrl{1} & \meter{} \\
\lstick{$\ket{0}$} & \gate{H} & \ctrl{-1} & & \targ{} & & 
\end{quantikz}
\\
&\stackrel{\text{(remove ancilla)}}{=}
\begin{quantikz}
\lstick{$\ket{0}$} & \gate{H} & & \ctrl{1} & & & \\
\lstick{$\ket{0}$} & & \targ{} & \targ{} & \ctrl{1} & \meter{} \\
\lstick{$\ket{0}$} & \gate{H} & \ctrl{-1} & & \targ{} & & 
\end{quantikz}
\stackrel{\text{(5CX)}}{=}
\begin{quantikz}
\lstick{$\ket{0}$} & \gate{H} & & & \ctrl{1} & \ctrl{2} & & \\
\lstick{$\ket{0}$} & & \targ{} & \ctrl{1} & \targ{} & & \meter{} \\
\lstick{$\ket{0}$} & \gate{H} & \ctrl{-1} & \targ{} & & \targ{} & & 
\end{quantikz}
\\
&\stackrel{\text{(mb)}}{=}
\begin{quantikz}
\lstick{$\ket{0}$} & \gate{H} & & & \ctrl{1} & \ctrl{2} & & & \\
\lstick{$\ket{0}$} & & \targ{} & \ctrl{1} & \targ{} & & \gate{H} & \meter{H} \\
\lstick{$\ket{0}$} & \gate{H} & \ctrl{-1} & \targ{} & & \targ{} & & & 
\end{quantikz}
\stackrel{\text{(H)}}{=}
\begin{quantikz}
\lstick{$\ket{0}$} & \gate{H} & & & & \ctrl{1} & \ctrl{2} & & \\
\lstick{$\ket{0}$} & & \targ{} & \ctrl{1} & \gate{H} & \phase{} & & \meter{H} \\
\lstick{$\ket{0}$} & \gate{H} & \ctrl{-1} & \targ{} & & & \targ{} & & &
\end{quantikz}
\end{aligned}
\end{equation*}

We now recognize the presence of an fSWAP, which can be freely removed because it fixes the input state $\ket{00}$, as proven in section A.2.3.

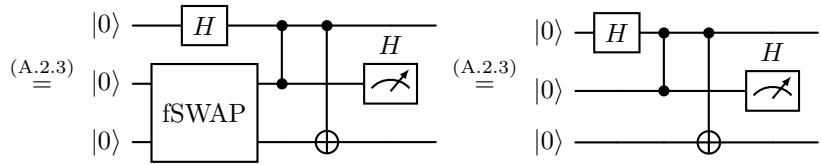
\begin{equation*}
\stackrel{\text{(A.2.3)}}{=}
\begin{quantikz}
\lstick{$\ket{0}$} & \gate{H} & \ctrl{1} & \ctrl{2} & & \\
\lstick{$\ket{0}$} & \gate[2]{\text{fSWAP}} & \phase{} & & \meter{H} \\
\lstick{$\ket{0}$} & \ghost{\text{fSWAP}} & & \targ{} & & 
\end{quantikz}
\stackrel{\text{(A.2.3)}}{=}
\begin{quantikz}
\lstick{$\ket{0}$} & \gate{H} & \ctrl{1} & \ctrl{2} & & \\
\lstick{$\ket{0}$} & & \phase{} & & \meter{H} \\
\lstick{$\ket{0}$} & & & \targ{} & & 
\end{quantikz}
\end{equation*}

Again we can use the state of an ancilla to remove a gate; this time a CZ that does not activate due to one of its inputs being in the state $\ket{0}$. Since the remaining circuit after removing the last non-interacting ancilla is a Bell state initialization, it's retroactively clear that this protocol generates a Bell pair between qubits 1 and 4.

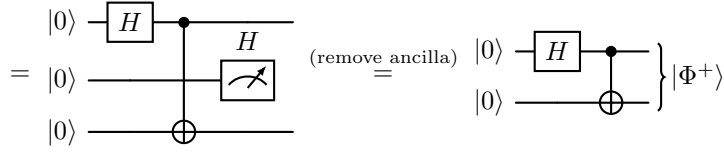
\begin{equation*}
\begin{aligned}
=
\begin{quantikz}
\lstick{$\ket{0}$} & \gate{H} & \ctrl{2} & & \\
\lstick{$\ket{0}$} & & & \meter{H} \\
\lstick{$\ket{0}$} & & \targ{} & & 
\end{quantikz}
\stackrel{\text{(remove ancilla)}}{=}
\begin{quantikz}
\lstick{$\ket{0}$} & \gate{H} & \ctrl{1} & \rstick[2]{$\ket{\Phi^+}$} \\
\lstick{$\ket{0}$} & & \targ{} & 
\end{quantikz}
\end{aligned}
\end{equation*}

\subsection{Algorithms}
Various algorithms are presented in the EQCD notation. Diagrammatic proofs of functionality are provided for the simpler algorithms.

\subsubsection{Phase kickback}
Phase kickback is more of an algorithmic primitive than an algorithm in its own right, but it's so ubiquitous that it deserves its own section. In EQCDs, phase kickback is a consequence of the fact that roots can slide freely along vertical lines.

Let $U = V\text{diag}(e^{i\pi\theta_j})V^\dagger$ be a generic unitary with eigenvectors $\ket{u_j}$ and associated eigenvalues $e^{i\pi\theta_j}$, acting on $n$ qubits collectively labelled $B$. The controlled unitary $CU$ with control qubit labelled $A$ is

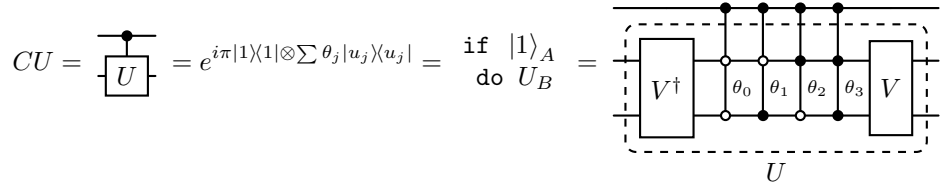
\begin{equation*}
CU = 
\begin{quantikz}[column sep=0.1cm,row sep=0.2cm]
    & \ctrl{1} & \\
    & \gate{U} & 
\end{quantikz}
=
e^{i\pi\dyad{1}\otimes\sum \theta_j \dyad{u_j}}= \begin{array}{l}
\verb|if | \ket{1}_A \\
\verb| do | U_B
\end{array} 
=
\begin{quantikz}[column sep=0.35cm,row sep=0.15cm]
    & \ghost{H} & \ctrl{1} & \ctrl{1} & \ctrl{1} & \ctrl{1} & & \\
    & \gate[2]{V^\dagger}\gategroup[2,steps=6,style={dashed,rounded corners, inner xsep=2pt, inner ysep=2pt},label style={label position=below, yshift=-0.5cm}]{$U$} & \octrl[wire style={"\theta_0"}]{1} & \octrl[wire style={"\theta_1"}]{1} & \ctrl[wire style={"\theta_2"}]{1} & \ctrl[wire style={"\theta_3"}]{1} & \gate[2]{V} & \\
    & & \ocontrol{} & \control{} & \ocontrol{} & \control{} & &
\end{quantikz} 
\end{equation*}

Sliding the eigenphases $\theta_j$ up to wire $A$ we get an equivalent form with a very different interpretation: the role of control and target have been reversed. Rather than a single unitary $U$ controlled by the qubit $A$, $CU$ now looks like a series of controlled phase gates $CZ^{\theta_j}$ controlled by $B$.

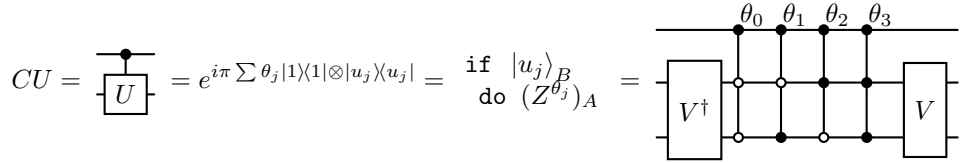
\begin{equation*}
CU = 
\begin{quantikz}[column sep=0.1cm,row sep=0.2cm]
    & \ctrl{1} & \\
    & \gate{U} & 
\end{quantikz}
=
e^{i\pi\sum\theta_j \dyad{1}\otimes \dyad{u_j}}= \begin{array}{l}
\verb|if | \ket{u_j}_B \\
\verb| do | (Z^{\theta_j})_A
\end{array}
=
\begin{quantikz}[column sep=0.15cm,row sep=0.15cm]
    & \ghost{H} & \phase{\theta_0} & \phase{\theta_1} & \phase{\theta_2} & \phase{\theta_3} & & \\
    & \gate[2]{V^\dagger} & \ocontrol{} & \ocontrol{} & \control{} & \control{} & \gate[2]{V} & \\
    & & \octrl{-2} & \ctrl{-2} & \octrl{-2} & \ctrl{-2} & &
\end{quantikz}
\end{equation*}

Here we've shown $CU$ controlled in the simplest and most common way, by the symbol $\bullet$ on a single wire, but this technique extends straightforwardly to any number of control qubits $A$ in any chosen basis. 

Some of the earliest-discovered examples of quantum algorithms with speedup over their classical counterparts use only phase kickback. We demonstrate three of these `simple kickback algorithms' now: Deutsch (1985), Deutsch-Jozsa (1992), and Bernstein-Vazirani (1997). The analysis given here using EQCDs procedes by translating everything into explicit circuits: we consider how the black boxes must look internally, notice that they contain specific arrangements of vertical lines and numbers, and then obtain a super-classical solution by sliding all of the numbers upward (the `phase kickback step').

\subsubsection{Deutsch}
The goal of Deutsch's algorithm is to distinguish two of the four one-bit functions -- the constant functions outputting 0 and 1 respectively -- from the other two, the balanced functions implementing an identity and bit flip respectively. Outputting a constant is not a unitary operation, so Deutsch's black box is implemented in the usual reversible way as a two-qubit unitary where the value of $f(x)$ is XOR-ed onto the second qubit. In other words, Deutsch's oracle has one of the four forms

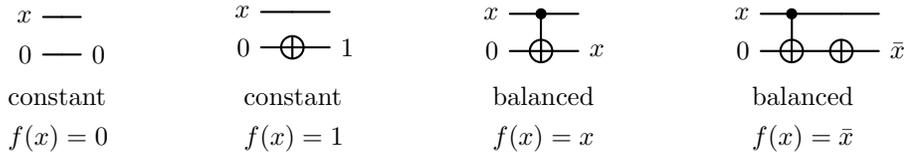
\begin{align*}
& \begin{quantikz}[row sep=0.5cm]
    \lstick{$x$} & & \\
    \lstick{$0$} & & \rstick{$0$}
\end{quantikz} & & \begin{quantikz}[row sep=0.3cm]
    \lstick{$x$} & & \\
    \lstick{$0$} & \targ{} & \rstick{$1$}
\end{quantikz} & & \begin{quantikz}
    \lstick{$x$} & \ctrl{1} & \\
    \lstick{$0$} & \targ{} & \rstick{$x$}
\end{quantikz} & & \begin{quantikz}
    \lstick{$x$} & \ctrl{1} & & \\
    \lstick{$0$} & \targ{} & \targ{} & \rstick{$\bar{x}$}
\end{quantikz} & \\
& \text{constant} & & \ \ \text{constant} & & \ \ \text{balanced} & & \quad \text{balanced} & \\
& f(x) = 0 & & \ \ f(x) = 1 & & \ \ f(x) = x & & \quad f(x) = \bar{x} &
\end{align*}

where the second input bit was fixed w.l.o.g. to 0 for simplicity of presentation. These four functions can be encapsulated in a single diagram indexed by the two bits $b$ and $c$, which -- diagrammatically speaking -- are powers of the gates they sit next to.

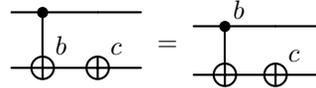
\begin{equation*}
\begin{quantikz}
    & \ctrl{1} & & \\
    & \targ{b} & \targ{c} & 
\end{quantikz}
=
\begin{quantikz}
    & \phase{b} & & \\
    & \trg{-1} & \targ{c} & 
\end{quantikz}
\end{equation*}

If $b=0$ then $f$ is constant with output $c$, and if $b=1$ then $f$ is balanced -- either the identity function $f(x) = x$ (when $c=0$) or the bit flip function $f(x) = \bar{x}$ (when $c=1$). Deutsch's problem is to determine the value of $b$ using as few calls to this oracle as possible. 

Inputting computational basis states is no different than using a classical strategy and requires two calls due to the inability to distinguish between $b$ and $c$ on the lower wire. Inputting $\ket{\pm}$ states, however, $c$ becomes a global phase on the lower wire, and $b$ can be read off the top wire by activating the CNOT with a $\ket{-}$ on the lower wire. If $b=0$ (i.e. there is no CNOT internal to the oracle and $f$ is constant) the top bit will come out unchanged, and if $b=1$ (i.e. there is a CNOT internal to the oracle and $f$ is balanced) the top bit will be flipped. 

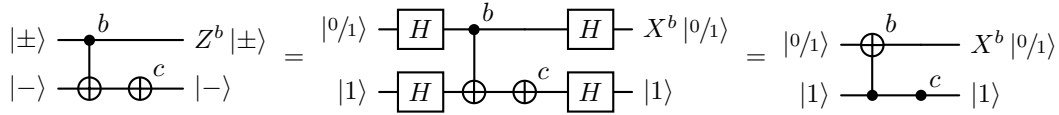
\begin{equation*}
\begin{quantikz}
    \lstick{$\ket{\pm}$} & \phase{b} & & \rstick{$Z^b\ket{\pm}$} \\
    \lstick{$\ket{-}$} & \trg{-1} & \targ{c} & \rstick{$\ket{-}$}
\end{quantikz}
=
\begin{quantikz}
    \lstick{$\ket{\sfrac{0}{1}}$} & \gate{H} & \phase{b} & & \gate{H} & \rstick{$X^b\ket{\sfrac{0}{1}}$} \\
    \lstick{$\ket{1}$} & \gate{H} & \trg{-1} & \targ{c} & \gate{H} & \rstick{$\ket{1}$}
\end{quantikz}
=
\begin{quantikz}
    \lstick{$\ket{\sfrac{0}{1}}$} & \targ{b} & & \rstick{$X^b\ket{\sfrac{0}{1}}$} \\
    \lstick{$\ket{1}$} & \ctrl{-1} & \phase{c} & \rstick{$\ket{1}$}
\end{quantikz}
\end{equation*}

These three diagrams of Deutsch's algorithm are equivalent. The first directly inputs $\ket{\pm}$ states, while the second inputs $\ket{\sfrac{0}{1}}$ followed by a Hadamard gate -- this is the usual solution as presented in textbooks. The third diagram is obtained by multiplying the $H$ gates directly into the oracle. It could be viewed as a modified `control-reversed' oracle, but it is of course locally equivalent to the normal oracle, and gives the clearest picture of why Deutsch's algorithm works. 

Note that either $\ket{0}$ or $\ket{1}$ would work as input on the top wire; typically Deutsch's algorithm is presented using specifically $\ket{0}$. A global phase $(-1)^c$ has been omitted from the final state in every diagram.

\subsubsection{Deutsch-Jozsa}
The second in this series of problems with nearly identical oracles and solutions, the goal here is to distinguish a constant n-bit to 1-bit function from a balanced one using fewer than the optimal classical value of $n/2$ calls to the oracle. 

What does the oracle look like this time? Being a classical binary function, it can only contain $X$ and CNOT gates, and since the $n$ input bits cannot change, it can only have the form

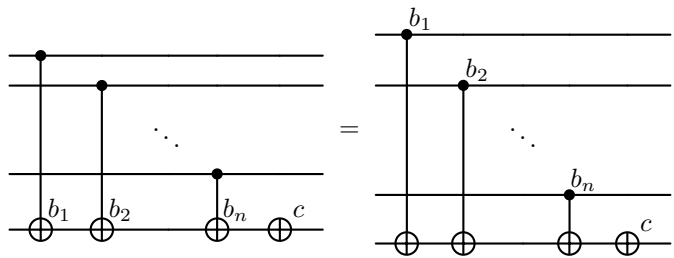
\begin{equation*}
\begin{quantikz}[wire types={q,q,n,q,q}]
    & \ctrl{4} & & & & & \\
    & & \ctrl{3} & & & & \\
    & & & \ddots{} & & & \\
    & & & & \ctrl{1} & & \\
    & \targ{b_1} & \targ{b_2} & & \targ{b_n} & \targ{c} & 
\end{quantikz}
=
\begin{quantikz}[wire types={q,q,n,q,q}]
    & \phase{b_1} & & & & & \\
    & & \phase{b_2} & & & & \\
    & & & \ddots & & & \\
    & & & & \phase{b_n} & & \\
    & \trg{-4} & \trg{-3} & & \trg{-1} & \targ{c} &
\end{quantikz}
\end{equation*}

If $f$ is constant it cannot depend on the input $x$, meaning every $b_k = 0$, i.e. the oracle contains no vertical lines connecting the bottom wire to any of the first $n$. If $f$ is balanced, at least one $b_k = 1$, meaning there exists at least one vertical line internal to the oracle.

The problem is solved in one call by the same strategy as before. Inputting $\ket{\pm}$ states on every wire is equivalent to calling a `control-reversed' oracle which outputs the string $b$ on the first $n$ wires when $\ket{1}$ is given as input.

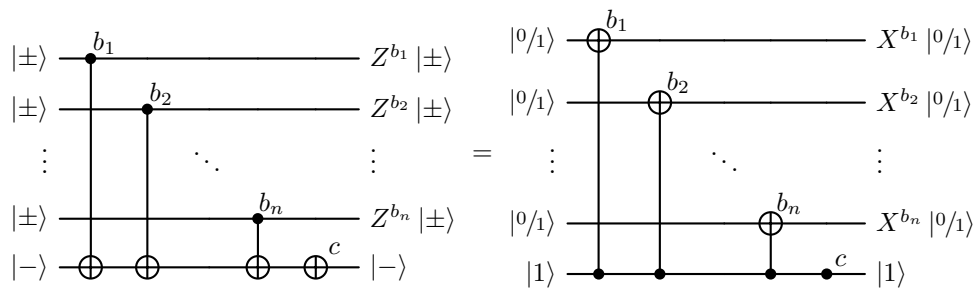
\begin{equation*}
\begin{quantikz}[wire types={q,q,n,q,q}]
    \lstick{$\ket{\pm}$} & \phase{b_1} & & & & & \rstick{$Z^{b_1}\ket{\pm}$} \\
    \lstick{$\ket{\pm}$} & & \phase{b_2} & & & & \rstick{$Z^{b_2}\ket{\pm}$} \\
    \lstick{\vdots} & & & \ddots{} & & & \rstick{\vdots} \\
    \lstick{$\ket{\pm}$} & & & & \phase{b_n} & & \rstick{$Z^{b_n}\ket{\pm}$} \\
    \lstick{$\ket{-}$} & \trg{-4} & \trg{-3} & & \trg{-1} & \targ{c} & \rstick{$\ket{-}$}
\end{quantikz}
=
\begin{quantikz}[wire types={q,q,n,q,q}]
    \lstick{$\ket{\sfrac{0}{1}}$} & \targ{b_1} & & & & & \rstick{$X^{b_1}\ket{\sfrac{0}{1}}$} \\
    \lstick{$\ket{\sfrac{0}{1}}$} & & \targ{b_2} & & & & \rstick{$X^{b_2}\ket{\sfrac{0}{1}}$} \\
    \lstick{\vdots} & & & \ddots & & & \rstick{\vdots} \\
    \lstick{$\ket{\sfrac{0}{1}}$} & & & & \targ{b_n} & & \rstick{$X^{b_n}\ket{\sfrac{0}{1}}$} \\
    \lstick{$\ket{1}$} & \ctrl{-4} & \ctrl{-3} & & \ctrl{-1} & \phase{c} & \rstick{$\ket{1}$} 
\end{quantikz}
\end{equation*}

Measuring the first $n$ wires in the $\pm$ basis (equivalently, the control-reversed oracle in the computational basis) returns the bit string $b = b_1b_2...b_n$, where $b_k = 1$ if the $k$th input state was flipped, and 0 otherwise. If $f$ is constant then $b$ is the all-zero string; if $f$ is balanced then $b$ is some other non-zero string. As before the typical input given in textbooks on the first $n$ wires is the all $+$ string, though any $\pm$ bitstring would work.

Although the Deutsch-Jozsa problem only asks one to distinguish $b=0$ from $b\neq 0$, this algorithm actually solves a slightly more general problem since it returns all of the bits of $b$. In the case that $b\neq 0$, this identifies specifically \textit{which} balanced function the oracle implements.\footnote{at least up to the single bit $c$, which can be determined with a second call to the oracle.}

\subsubsection{Bernstein-Vazirani}
The Bernstein-Vazirani problem is a special case of the Deutsch-Jozsa problem, although it is usually phrased in terms of extracting a `secret string' $s$ from an oracle which XORs $s$ with the given input string $x$. The Bernstein-Vazirani oracle is equivalent to the Deutsch-Jozsa oracle with $c=0$.

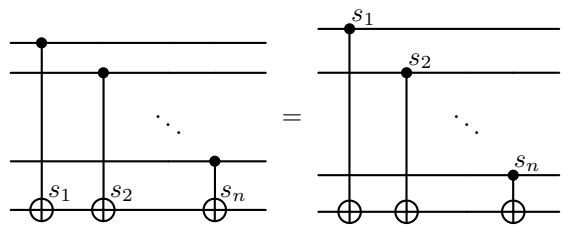
\begin{equation*}
\begin{quantikz}[wire types={q,q,n,q,q}]
    & \ctrl{4} & & & & \\
    & & \ctrl{3} & & & \\
    & & & \ddots & & \\
    & & & & \ctrl{1} & \\
    & \targ{s_1} & \targ{s_2} & & \targ{s_n} &
\end{quantikz}
=
\begin{quantikz}[wire types={q,q,n,q,q}]
    & \phase{s_1} & & & & \\
    & & \phase{s_2} & & & \\
    & & & \ddots & & \\
    & & & & \phase{s_n} & \\
    & \trg{-4} & \trg{-3} & & \trg{-1} &
\end{quantikz}
\end{equation*}

As with Deutsch's problem, if we are restricted to only classical inputs (i.e. computational basis states) this problem requires $n$ calls to the oracle to read off the string $s$. Inputting $X$ basis states with a $\ket{-}$ on the last wire will control all the CNOTs on the bottom wire, activating them all and implementing $Z^{s_k}$ on the $k$th wire. This conditionally flips each of the first $n$ bits and allows you to read off the string $s$ from these qubits using just one call to the oracle.

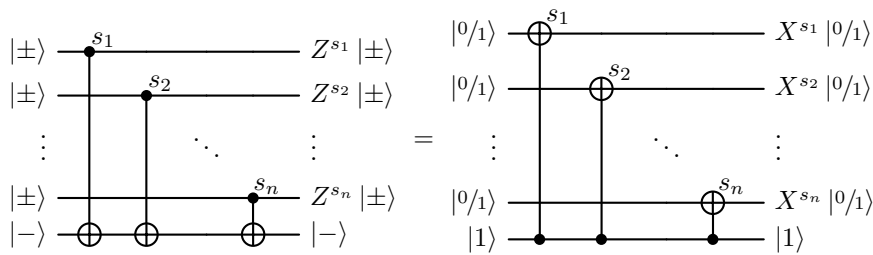
\begin{equation*}
\begin{quantikz}[wire types={q,q,n,q,q}]
    \lstick{$\ket{\pm}$} & \phase{s_1} & & & & \rstick{$Z^{s_1}\ket{\pm}$} \\
    \lstick{$\ket{\pm}$} & & \phase{s_2} & & & \rstick{$Z^{s_2}\ket{\pm}$} \\
    \lstick{\vdots} & & & \ddots & & \rstick{\vdots} \\
    \lstick{$\ket{\pm}$} & & & & \phase{s_n} & \rstick{$Z^{s_n}\ket{\pm}$} \\
    \lstick{$\ket{-}$} & \trg{-4} & \trg{-3} & & \trg{-1} & \rstick{$\ket{-}$}
\end{quantikz}
=
\begin{quantikz}[wire types={q,q,n,q,q}]
    \lstick{$\ket{\sfrac{0}{1}}$} & \targ{s_1} & & & & \rstick{$X^{s_1}\ket{\sfrac{0}{1}}$} \\
    \lstick{$\ket{\sfrac{0}{1}}$} & & \targ{s_2} & & & \rstick{$X^{s_2}\ket{\sfrac{0}{1}}$} \\
    \lstick{\vdots} & & & \ddots & & \rstick{\vdots} \\
    \lstick{$\ket{\sfrac{0}{1}}$} & & & & \targ{s_n} & \rstick{$X^{s_n}\ket{\sfrac{0}{1}}$} \\
    \lstick{$\ket{1}$} & \ctrl{-4} & \ctrl{-3} & & \ctrl{-1} & \rstick{$\ket{1}$} 
\end{quantikz}
\end{equation*}

Again, the simplest and typical presentation is to input $\ket{+}$ on the first $n$ wires, although any $\pm$ bitstring would also work.

\subsubsection{Grover}
Here we simply translate Grover's algorithm into the notation of EQCDs and do not prove its functionality. Grover's algorithm involves repeated calls to two unitaries: the phase oracle, whose internal parameters we are trying to recover, and the `diffusion operator' which is a fixed, known operation. The phase oracle applies a $-1$ phase to each among some set of marked computational basis states which we are searching for; here we will assume the hardest case where there is only one marked state with binary expansion $x = x_1x_2...x_n$. The Grover oracle has the form 

\begin{equation*}
\begin{quantikz}[wire types={q,q,n,n,q,q}]
    & \ctrl{5} & \\
    & \ocontrol{} & \\
    \lstick{\vdots} & & \\
    & & \\
    & \control{} & \\
    & \control{} & 
\end{quantikz}
= e^{i\pi\dyad{10...11}}
\end{equation*}

where the symbol on the $k$th wire is $\circ$ if $x_k = 0$ and $\bullet$ if $x_k = 1$. We can summarize this oracle as

\begin{equation*}
\begin{quantikz}[wire types={q,q,n,n,q}]
    & \measure{x_1}\vqw{4} & \\
    & \measure{x_2} & \\
    \lstick{\vdots} & & \\
    & & \\
    & \measure{x_n} & 
\end{quantikz}
= e^{i\pi\dyad{x_1x_2...x_n}} = e^{i\pi\dyad{x}} = \eye - 2\dyad{x}
\end{equation*}

where 
\begin{quantikz}
    \measure{0}
\end{quantikz}
=
\begin{quantikz}
    \ocontrol{}
\end{quantikz}
and
\begin{quantikz}
    \measure{1}
\end{quantikz}
=
\begin{quantikz}
    \control{}
\end{quantikz}.

The diffusion operator is $2\dyad{+^n} - \eye$, where $\ket{+^n} = \ket{+}^{\otimes n}$. This is equivalent to $\eye - 2\dyad{+^n}$ up to a global phase of $-1$, and we will use this operator instead because it has only one eigenvalue equal to $-1$ (namely, the one for $\ket{+^n}$) and $2^n-1$ eigenvalues equal to 1 rather than the other way around. The diffusion operator is

\begin{equation*}
\begin{quantikz}[wire types={q,q,n,n,q}]
    & \targm{}\vqw{1} & \\
    & \targm{}\vqw{3} & \\
    \lstick{\vdots} & & \\
    & & \\
    & \targm{} & 
\end{quantikz}
= e^{i\pi\dyad{+^n}} = \eye - 2\dyad{+^n}
\end{equation*}

Grover's algorithm is then

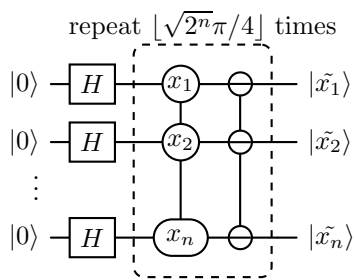
\begin{equation*}
\begin{quantikz}[wire types={q,q,n,n,q}]
    \lstick{$\ket{0}$} & \gate{H} & & \measure{x_1}\vqw{4}\gategroup[5,steps=2,style={dashed,rounded corners}]{repeat $\lfloor\sqrt{2^n}\pi/4\rfloor$ times} & \targm{}\vqw{1} & & \rstick{$\ket{\Tilde{x_1}}$} \\
    \lstick{$\ket{0}$} & \gate{H} & & \measure{x_2} & \targm{}\vqw{3} & & \rstick{$\ket{\Tilde{x_2}}$}\\
    \lstick{\vdots} & & & & \\
    & & & & \\
    \lstick{$\ket{0}$} & \gate{H} & & \measure{x_n} & \targm{} & & \rstick{$\ket{\Tilde{x_n}}$}
\end{quantikz}
\end{equation*}

where $\ket{\Tilde{x_k}}$ is an approximation to $\ket{x_k}$.

\subsubsection{Phase Estimation}
We present the phase estimate algorithm in the notation of EQCDs, and point out a small property its circuit has (that it can be partly `interleaved'), which is demonstrated with an example calculation. 

We begin with the standard phase estimation setup. Let $U\ket{u} = e^{2\pi i \phi}\ket{u}$. In this first step we remove the $U$-gadget by sliding the eigenphases upward (phase kickback). The $\ket{u}$ wire bundle can then be removed.

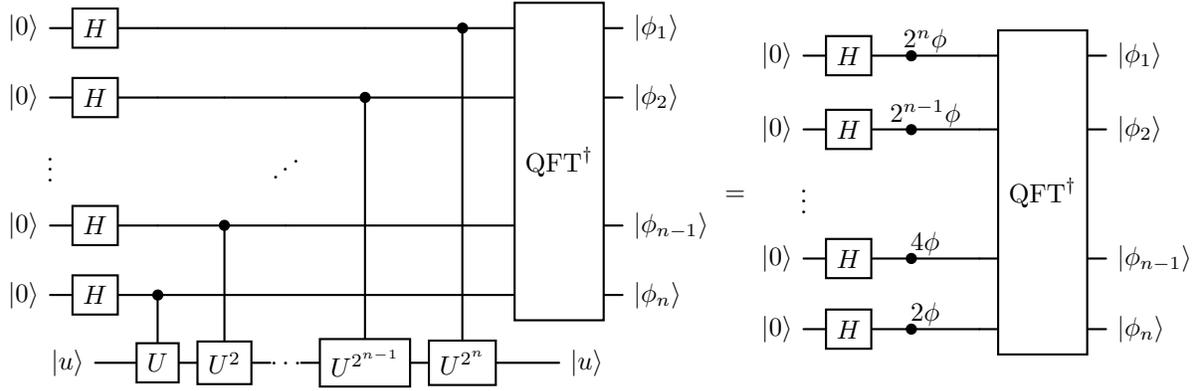
\begin{equation*}
\begin{quantikz}[wire types={q,q,n,q,q,n}]
\lstick{$\ket{0}$} & \gate{H} & & & & & \ctrl{5} & \gate[5]{\text{QFT}^\dagger} & \rstick{$\ket{\phi_1}$} \\
\lstick{$\ket{0}$} & \gate{H} & & & & \ctrl{4} & & \ghost{\text{QFT}^\dagger} & \rstick{$\ket{\phi_2}$} \\
\vdots & & & & \iddots & & & & \\
\lstick{$\ket{0}$} & \gate{H} & & \ctrl{2} & & & & \ghost{\text{QFT}^\dagger} & \rstick{$\ket{\phi_{n-1}}$} \\
\lstick{$\ket{0}$} & \gate{H} & \ctrl{1} & & & & & \ghost{\text{QFT}^\dagger} & \rstick{$\ket{\phi_n}$} \\
& \lstick{$\ket{u}$} & \gate{U}\setwiretype{q} & \gate{U^2} & \hdots & \gate{U^{2^{n-1}}} & \gate{U^{2^n}} & \rstick{$\ket{u}$}
\end{quantikz}
=
\begin{quantikz}[wire types={q,q,n,q,q}]
\lstick{$\ket{0}$} & \gate{H} & \phase{2^n\phi} & & \gate[5]{\text{QFT}^\dagger} & \rstick{$\ket{\phi_1}$} \\
\lstick{$\ket{0}$} & \gate{H} & \phase{2^{n-1}\phi} & & \ghost{\text{QFT}^\dagger} & \rstick{$\ket{\phi_2}$} \\
\vdots & & & & & & & & \\
\lstick{$\ket{0}$} & \gate{H} & \phase{4\phi} & & \ghost{\text{QFT}^\dagger} & \rstick{$\ket{\phi_{n-1}}$} \\
\lstick{$\ket{0}$} & \gate{H} & \phase{2\phi} & & \ghost{\text{QFT}^\dagger} & \rstick{$\ket{\phi_n}$}
\end{quantikz}
\end{equation*}

In the next steps we introduce the standard circuit implementation of the QFT and untwist the left side to remove the SWAPs.

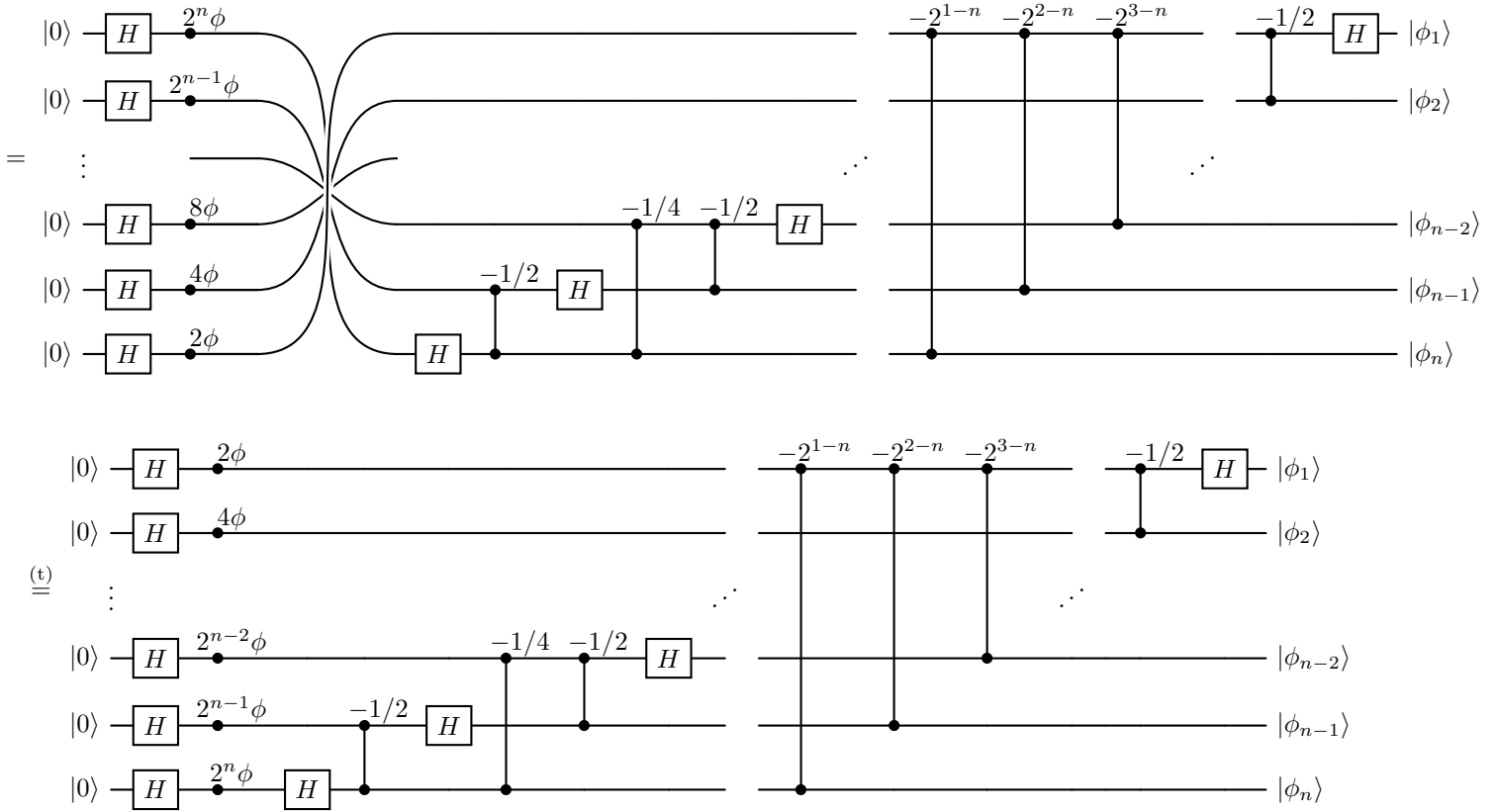
\begin{equation*}
\hspace{-2cm} 
=
\begin{quantikz}[wire types={q,q,n,q,q}]
\lstick{$\ket{0}$} & \gate{H} & \phase{2^n\phi} & \widepermute{6,5,4,3,2,1} & & & & & & & & \wireoverride{n} & \phase{-2^{1-n}} & \phase{-2^{2-n}} & \phase{-2^{3-n}} & & \wireoverride{n} & \phase{-1/2} & \gate{H} & \rstick{$\ket{\phi_1}$} \\
\lstick{$\ket{0}$} & \gate{H} & \phase{2^{n-1}\phi} & & & & & & & & & \wireoverride{n} & & & & & \wireoverride{n} & \ctrl{-1} & & \rstick{$\ket{\phi_2}$} \\
\vdots & & & & & & & & & & \iddots & & & & & \iddots & & & & & & & \\
\lstick{$\ket{0}$} & \gate{H} & \phase{8\phi} & & & & & \phase{-1/4} & \phase{-1/2} & \gate{H} & & \wireoverride{n} & & & \ctrl{-3} & & & & & \rstick{$\ket{\phi_{n-2}}$} \\
\lstick{$\ket{0}$} & \gate{H} & \phase{4\phi} & & & \phase{-1/2} & \gate{H} & & \ctrl{-1} & & & \wireoverride{n} & & \ctrl{-4} & & & & & & \rstick{$\ket{\phi_{n-1}}$} \\
\lstick{$\ket{0}$} & \gate{H} & \phase{2\phi} & & \gate{H} & \ctrl{-1} & & \ctrl{-2} & & & & \wireoverride{n} & \ctrl{-5} & & & & & & & \rstick{$\ket{\phi_n}$}
\end{quantikz}
\end{equation*}

\begin{equation*}
\hspace{-2cm} 
\stackrel{\text{(t)}}{=}
\begin{quantikz}[wire types={q,q,n,q,q}]
\lstick{$\ket{0}$} & \gate{H} & \phase{2\phi} & & & & & & & & \wireoverride{n} & \phase{-2^{1-n}} & \phase{-2^{2-n}} & \phase{-2^{3-n}} & & \wireoverride{n} & \phase{-1/2} & \gate{H} & \rstick{$\ket{\phi_1}$} \\
\lstick{$\ket{0}$} & \gate{H} & \phase{4\phi} & & & & & & & & \wireoverride{n} & & & & & \wireoverride{n} & \ctrl{-1} & & \rstick{$\ket{\phi_2}$} \\
\vdots & & & & & & & & & \iddots & & & & & \iddots & & & & & & & \\
\lstick{$\ket{0}$} & \gate{H} & \phase{2^{n-2}\phi} & & & & \phase{-1/4} & \phase{-1/2} & \gate{H} & & \wireoverride{n} & & & \ctrl{-3} & & & & & \rstick{$\ket{\phi_{n-2}}$} \\
\lstick{$\ket{0}$} & \gate{H} & \phase{2^{n-1}\phi} & & \phase{-1/2} & \gate{H} & & \ctrl{-1} & & & \wireoverride{n} & & \ctrl{-4} & & & & & & \rstick{$\ket{\phi_{n-1}}$} \\
\lstick{$\ket{0}$} & \gate{H} & \phase{2^n\phi} & \gate{H} & \ctrl{-1} & & \ctrl{-2} & & & & \wireoverride{n} & \ctrl{-5} & & & & & & & \rstick{$\ket{\phi_n}$}
\end{quantikz}
\end{equation*}

In the next step we combine all of the pairs of Hadamard gates, changing every sandwiched $\bullet$ into a $\oplus$.

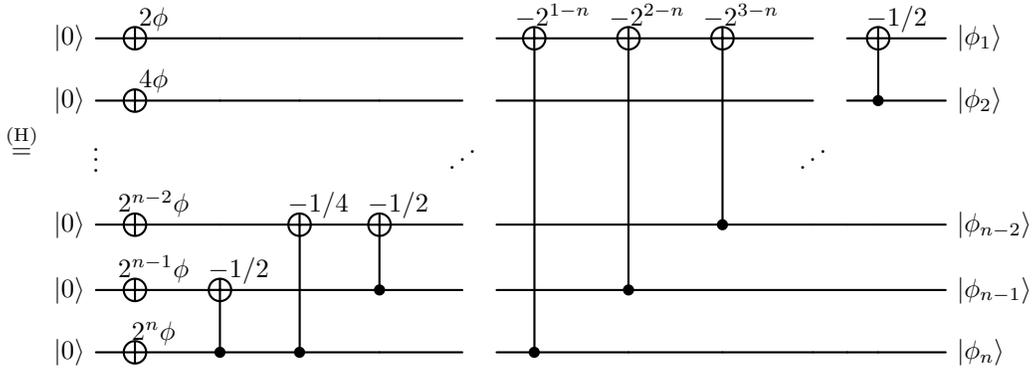
\begin{equation*}
\stackrel{\text{(H)}}{=}
\begin{quantikz}[wire types={q,q,n,q,q}]
\lstick{$\ket{0}$} & \targ{2\phi} & & & & & \wireoverride{n} & \targ{-2^{1-n}} & \targ{-2^{2-n}} & \targ{-2^{3-n}} & & \wireoverride{n} & \targ{-1/2} & \rstick{$\ket{\phi_1}$} \\
\lstick{$\ket{0}$}  & \targ{4\phi} & & & & & \wireoverride{n} & & & & & \wireoverride{n} & \ctrl{-1} & \rstick{$\ket{\phi_2}$} \\
\vdots & & & & & \iddots & & & & & \iddots & & & & & & \\
\lstick{$\ket{0}$} & \targ{2^{n-2}\phi} & & \targ{-1/4} & \targ{-1/2} & & \wireoverride{n} & & & \ctrl{-3} & & & & \rstick{$\ket{\phi_{n-2}}$} \\
\lstick{$\ket{0}$} & \targ{2^{n-1}\phi} & \targ{-1/2} & & \ctrl{-1} & & \wireoverride{n} & & \ctrl{-4} & & & & & \rstick{$\ket{\phi_{n-1}}$} \\
\lstick{$\ket{0}$} & \targ{2^n\phi} & \ctrl{-1} & \ctrl{-2} & & & \wireoverride{n} & \ctrl{-5} & & & & & & \rstick{$\ket{\phi_n}$}
\end{quantikz}
\end{equation*}

The stack of $\oplus^{2^k\phi}$ symbols on the left are our oracle calls, and the upward cascade of CNOTs is what remains of the inverse QFT. Note that since $\oplus$ commutes with itself but not with $\bullet$, the oracle calls can be pushed partly into the circuit until they encounter a $\bullet$. In this way the oracle calls can be partly interleaved with the inverse QFT. In the case that $\phi$ is representable in $n$ bits without truncation or rounding, this interleaving procedure results in each output qubit being set sequentially, starting with the lowest (least significant) qubit and proceeding upwards. We demonstrate this with an explicit example for the case $n=3$ and $\phi = 5/8$, which has an exact solution since $5/8 = 0.101$ terminates after 3 bits.

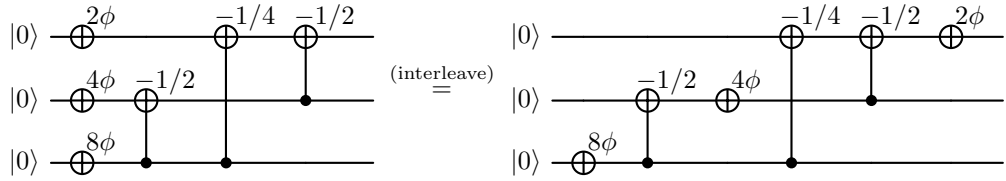
\begin{equation*}
\begin{quantikz}
\lstick{$\ket{0}$} & \targ{2\phi} & & \targ{-1/4} & \targ{-1/2} & \\
\lstick{$\ket{0}$} & \targ{4\phi} & \targ{-1/2} & & \ctrl{-1} & \\
\lstick{$\ket{0}$} & \targ{8\phi} & \ctrl{-1} & \ctrl{-2} & & 
\end{quantikz}
\stackrel{\text{(interleave)}}{=}
\begin{quantikz}
\lstick{$\ket{0}$} & & & & \targ{-1/4} & \targ{-1/2} & \targ{2\phi} & \\
\lstick{$\ket{0}$} & & \targ{-1/2} & \targ{4\phi} & & \ctrl{-1} & & \\
\lstick{$\ket{0}$} & \targ{8\phi} & \ctrl{-1} & & \ctrl{-2} & & & 
\end{quantikz}
\end{equation*}

Substituting in $\phi=5/8$ we notice that the oracle calls take the forms $\oplus^{2\phi} = \oplus^{5/4}$, $\oplus^{4\phi} = \oplus^{5/2} = \oplus^{1/2}$, and $\oplus^{8\phi} = \oplus^5 = \oplus$. The latter two identities hold because $\oplus^2 = \eye$. With this we begin calculating from the left. 

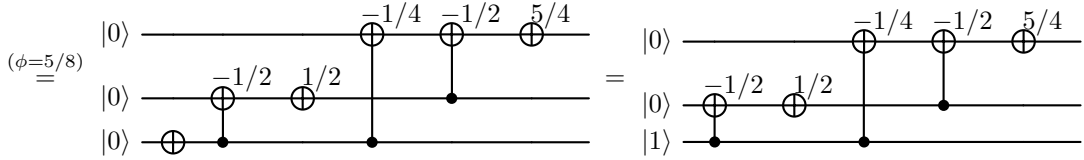
\begin{equation*}
\stackrel{(\phi=5/8)}{=}
\begin{quantikz}
\lstick{$\ket{0}$} & & & & \targ{-1/4} & \targ{-1/2} & \targ{5/4} & \\
\lstick{$\ket{0}$} & & \targ{-1/2} & \targ{1/2} & & \ctrl{-1} & & \\
\lstick{$\ket{0}$} & \targ{} & \ctrl{-1} & & \ctrl{-2} & & & 
\end{quantikz}
=
\begin{quantikz}
\lstick{$\ket{0}$} & & & \targ{-1/4} & \targ{-1/2} & \targ{5/4} & \\
\lstick{$\ket{0}$} & \targ{-1/2} & \targ{1/2} & & \ctrl{-1} & & \\
\lstick{$\ket{1}$} & \ctrl{-1} & & \ctrl{-2} & & & 
\end{quantikz}
\end{equation*}

The least significant bit is now set and will not change throughout the rest of the calculation since there are no more $\oplus$ symbols on its wire; there are only $\bullet$ symbols, which it is a (partial) eigenvector of. Now we consider the next two gates. Since the lowest wire is in the state $\ket{1}$, the first CNOT$^{-1/2}$ will deterministically apply a $\oplus^{-1/2}$ to the second wire. This is then immediately undone by the following $\oplus^{1/2}$.

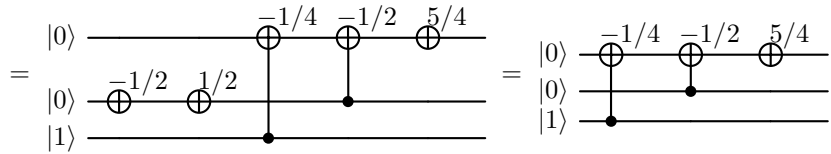
\begin{equation*}
=
\begin{quantikz}
\lstick{$\ket{0}$} & & & \targ{-1/4} & \targ{-1/2} & \targ{5/4} & \\
\lstick{$\ket{0}$} & \targ{-1/2} & \targ{1/2} & & \ctrl{-1} & & \\
\lstick{$\ket{1}$} & & & \ctrl{-2} & & & 
\end{quantikz}
=
\begin{quantikz}
\lstick{$\ket{0}$} & \targ{-1/4} & \targ{-1/2} & \targ{5/4} & \\
\lstick{$\ket{0}$} & & \ctrl{-1} & & \\
\lstick{$\ket{1}$} & \ctrl{-2} & & & 
\end{quantikz}
\end{equation*}

The second bit of $\phi$ is now set as well since there are no more $\oplus$ symbols on its wire either. For the most significant bit of $\phi$ we apply the same logic. The CNOT$^{-1/4}$ will activate but the CNOT$^{-1/2}$ will not, leaving a net $\oplus^{-1/4}$ on the top wire. Finally we combine this with the last oracle call $\oplus^{5/4}$, leaving a simple $\oplus$ on the top wire and completing the calculation.

\begin{equation*}
=
\begin{quantikz}
\lstick{$\ket{0}$} & \targ{-1/4} & \targ{5/4} & \\
\lstick{$\ket{0}$} & & & \\
\lstick{$\ket{1}$} & \ghost{H} & & 
\end{quantikz}
=
\begin{quantikz}
\lstick{$\ket{0}$} & \targ{} & \\
\lstick{$\ket{0}$} & & \\
\lstick{$\ket{1}$} & \ghost{H} & 
\end{quantikz}
=
\begin{quantikz}[wire types={n,n,n}]
\lstick{$\ket{1}$} & \\
\lstick{$\ket{0}$} & \ghost{H} \\
\lstick{$\ket{1}$} & 
\end{quantikz}
\end{equation*}

\section{Soundness proofs}\label{appendix:soundness}

In this section we provide soundness proofs of the rules and their flipflop symmetry principles.

\subsection{Soundness of the rules}

We provide algebraic proofs of the rules, using whichever mathematical formalism allows for the simplest proof (matrix or exponential). BCH denotes usage of the Baker-Campbell-Hausdorff formula, which here never introduces a commutator series because the involved Hamiltonians always commute. Arbitrary unitaries $U$ have Hamiltonians denoted by $H_U$ (not to be confused with Hadamard, denoted by $H$ sans subscript).

\begin{equation*}
\text{\textbf{(c)}:  }
\begin{bmatrix}
   1 & 0 \\
   0 & e^{i\pi\beta} 
\end{bmatrix}
\begin{bmatrix}
   1 & 0 \\
   0 & e^{i\pi\alpha} 
\end{bmatrix}
=
\begin{bmatrix}
   1 & 0 \\
   0 & e^{i\pi(\alpha+\beta)} 
\end{bmatrix}
=
\begin{bmatrix}
   1 & 0 \\
   0 & e^{i\pi\alpha} 
\end{bmatrix}
\begin{bmatrix}
   1 & 0 \\
   0 & e^{i\pi\beta} 
\end{bmatrix}
\end{equation*}

\begin{equation*}
\text{\textbf{(ac)}:  }
\begin{bmatrix}
   0 & 1 \\
   1 & 0 
\end{bmatrix}
\begin{bmatrix}
   1 & 0 \\
   0 & e^{i\pi\alpha} 
\end{bmatrix}
\begin{bmatrix}
   0 & 1 \\
   1 & 0 
\end{bmatrix}
=
\begin{bmatrix}
   e^{i\pi\alpha} & 0 \\
   0 & 1 
\end{bmatrix}
\end{equation*}

\begin{equation*}
\text{\textbf{(H)}:  } 
H e^{i\pi\alpha\dyad{1}} H = e^{i\pi\alpha H\dyad{1}H} = e^{i\pi\alpha \dyad{-}}
\text{  and  }
H e^{i\pi\alpha\dyad{0}} H = e^{i\pi\alpha H\dyad{0}H} = e^{i\pi\alpha \dyad{+}}
\end{equation*}

\begin{equation*}
\text{\textbf{(i)}:  } 
Z^2 = H^2 = \eye
\end{equation*}

\begin{equation*}
\text{\textbf{(C0)}:  } 
e^{i\pi(\zero \otimes H_U)} = e^{i\pi\zero} = e^\zero = \eye
\end{equation*}

\begin{equation*}
\text{\textbf{(d)}:  } 
\begin{bmatrix}
   \eye & \zero \\
   \zero & V
\end{bmatrix}
\begin{bmatrix}
   \eye & \zero \\
   \zero & U
\end{bmatrix}
=
\begin{bmatrix}
   \eye & \zero \\
   \zero & VU
\end{bmatrix}
\end{equation*}

\begin{equation*}
\begin{aligned}
\text{\textbf{(e)}:  } 
\eye \otimes e^{i\pi H_U} =
\begin{bmatrix}
   U & \zero \\
   \zero & U
\end{bmatrix}
= e^{i\pi \eye \otimes H_U} &= e^{i\pi (\dyad{0}+\dyad{1}) \otimes H_U} \stackrel{\text{(BCH)}}{=} e^{i\pi \dyad{0}\otimes H_U}e^{i\pi \dyad{1}\otimes H_U} \\
\Big( &= e^{i\pi (\dyad{+}+\dyad{-}) \otimes H_U} \stackrel{\text{(BCH)}}{=} e^{i\pi \dyad{+}\otimes H_U}e^{i\pi \dyad{-}\otimes H_U}\Big)
\end{aligned}
\end{equation*}

\begin{equation*}
\text{\textbf{(f)}:  } 
\begin{bmatrix}
   \eye & \zero \\
   \zero & V
\end{bmatrix}
\begin{bmatrix}
   U & \zero \\
   \zero & \eye
\end{bmatrix}
=
\begin{bmatrix}
   U & \zero \\
   \zero & V
\end{bmatrix}
=
\begin{bmatrix}
   U & \zero \\
   \zero & \eye
\end{bmatrix}
\begin{bmatrix}
   \eye & \zero \\
   \zero & V
\end{bmatrix}
\end{equation*}

Rule (g) is easily proved using matrices, but since it's derived from (d-f), each proven above, we will prove (g) graphically using (d-f). By assumption, 

\begin{equation}\label{rule g assumption}
\begin{quantikz}
    & \gate{U} & \gate{V} &
\end{quantikz}
=
\begin{quantikz}
    & \gate{V'} & \gate{U} &
\end{quantikz}
\end{equation}

The proof of rule (g) is then as follows.

\begin{equation*}
\begin{aligned}
\text{\textbf{(g)}:  } 
\begin{quantikz}
    & & \ctrl{1} & \\
    & \gate{U} & \gate{V} &
\end{quantikz}
&\stackrel{\text{(e)}}{=}
\begin{quantikz}
    & \octrl{1} & \ctrl{1} & \ctrl{1} & \\
    & \gate{U} & \gate{U} & \gate{V} &
\end{quantikz}
\stackrel{\text{(d)}}{=}
\begin{quantikz}
    & \octrl{1} & \ctrl{1} & \\
    & \gate{U} & \gate{U\circ V} &
\end{quantikz}
\stackrel{(\ref{rule g assumption})}{=}
\begin{quantikz}
    & \octrl{1} & \ctrl{1} & \\
    & \gate{U} & \gate{V'\circ U} &
\end{quantikz} \\
&\stackrel{\text{(d)}}{=}
\begin{quantikz}
    & \octrl{1} & \ctrl{1} & \ctrl{1} & \\
    & \gate{U} & \gate{V'} & \gate{U} &
\end{quantikz}
\stackrel{\text{(f)}}{=}
\begin{quantikz}
    & \ctrl{1} & \octrl{1} & \ctrl{1} & \\
    & \gate{V'} & \gate{U} & \gate{U} & 
\end{quantikz}
\stackrel{\text{(e)}}{=}
\begin{quantikz}
    & \ctrl{1} & & \\
    & \gate{V'} & \gate{U} & 
\end{quantikz}
\end{aligned}
\end{equation*}

\begin{equation*}
\text{\textbf{(s)}:  } 
\begin{bmatrix}
   1 & 0 & 0 & 0 \\
   0 & 0 & 1 & 0 \\
   0 & 1 & 0 & 0 \\
   0 & 0 & 0 & 1 
\end{bmatrix}
=
\begin{bmatrix}
   1 & 0 & 0 & 0 \\
   0 & 1 & 0 & 0 \\
   0 & 0 & 0 & 1 \\
   0 & 0 & 1 & 0 
\end{bmatrix}
\begin{bmatrix}
   1 & 0 & 0 & 0 \\
   0 & 0 & 0 & 1 \\
   0 & 0 & 1 & 0 \\
   0 & 1 & 0 & 0 
\end{bmatrix}
\begin{bmatrix}
   1 & 0 & 0 & 0 \\
   0 & 1 & 0 & 0 \\
   0 & 0 & 0 & 1 \\
   0 & 0 & 1 & 0 
\end{bmatrix}
\end{equation*}

Let $U_{AB} := e^{i\pi\sum\lambda_j\ket{u_j}_{AB}\bra{u_j}_{AB}}$ be a gate acting on qubits $A,B$ with eigenvectors $\ket{u_j}_{AB}$. Then SWAP has the action $\text{SWAP}\ket{u_j}_{AB} = \ket{u_j}_{BA}$, and
\begin{equation*}
\begin{aligned}
\text{\textbf{(t)}:  }
&\text{SWAP}\cdot e^{i\pi\sum\lambda_j\ket{u_j}_{AB}\bra{u_j}_{AB}} = e^{i\pi\sum\lambda_j\text{SWAP}\ket{u_j}_{AB}\bra{u_j}_{AB}} \\
&= e^{i\pi\sum\lambda_j\ket{u_j}_{BA}\bra{u_j}_{AB}} = e^{i\pi\sum\lambda_j\ket{u_j}_{BA}\bra{u_j}_{BA}\text{SWAP}} 
= e^{i\pi\sum\lambda_j\ket{u_j}_{BA}\bra{u_j}_{BA}}\cdot \text{SWAP}.
\end{aligned}
\end{equation*}

\begin{equation*}
\text{\textbf{(n+)}:  } 
e^{i\pi\beta H_U}e^{i\pi\alpha H_U} \stackrel{\text{(BCH)}}{=} e^{i\pi(\alpha+\beta) H_U}
\end{equation*}

\begin{equation*}
\text{\textbf{(n*)}:  } 
e^{i\pi\big(\alpha H_U \otimes \beta H_V\big)} = e^{i\pi\big(\alpha\beta H_U \otimes H_V\big)} = e^{i\pi\big(H_U \otimes \alpha\beta H_V\big)}
\end{equation*}

For soundness of rule (pdm) we refer the reader to the proofs given in the quantum computing literature \cite{nielsen_quantum_2010}. For rule (mb), let $U\ket{j} = \ket{u_j}$ and let the $U$-measurement box denote a projective measurement in the basis $\{\ket{u_j}\}$ with outcome probabilities $\braket{u_j}{\psi}\braket{\psi}{u_j}$ on the measured state $\ket{\psi}$. The input state on the RHS is $U^\dagger\ket{\psi}$, and the measurement basis is the computational basis $\{\ket{j}\}$, so the measurement outcome probabilities are $\bra{j}U^\dagger\ket{\psi}\bra{\psi}U\ket{j} = \braket{u_j}{\psi}\braket{\psi}{u_j}$, as claimed.

\subsection{Soundness of the symmetry principles}

The extra rules generated by the symmetry principles can be explicitly constructed by choosing $\bm{\oplus} = X$ and $\bm{\bullet} = Z$ as explicit one-qubit flip maps for flipping $Z$-type and $X$-type symbols respectively, and $H$ as the explicit one-qubit flop map. Flipped and flopped variants of each rule can then be constructed by conjugating said rule with these one-qubit flipflop operators on each individual wire containing symbols. 

The fact that $X$, $Z$, and $H$ are valid one-qubit flip and flop maps is a consequence of the one-qubit rules (c-i) which were proven above. To show this we start by noting that $X$, $Z$, and $H$ are Hermitian, and therefore the dagger that would normally appear when conjugating $U: A \mapsto UAU^\dagger$ can be dropped; we will therefore say `conjugate by' as a shorthand for `multiply both sides by'.

First we demonstrate that $H$ is a one-qubit flop map. The proof of rule (H) with $\alpha=1$ says
\begin{equation*}
    H(Z)H = X \quad \text{and} \quad H(-Z)H = -X.
\end{equation*}
Conjugating both equations by $H$ and appealing to rule (i) gives
\begin{equation*}
    H(X)H = Z \quad \text{and} \quad H(-X)H = -Z.
\end{equation*}
(Parentheses added for emphasis on what is being conjugated.) These four equations together say that conjugating by $H$ has the action $(\bm{\bullet}\leftrightarrow\bm{\oplus}$ and $\bm{\circ}\leftrightarrow\bm{\ominus})$ as desired.

For $X$, substituting $\alpha=1$ into the proof of rule (ac) above gives $X(Z)X=-Z$. Since $X$ is both unitary and Hermitian, it is also involutory, so conjugating by $X$ gives $Z = X(-Z)X$. $X$ being an involution also gives $X(X)X = X$, and together with rule (c) gives $X(-X)X = (-X)XX = -X$. That is,
\begin{equation*}
\begin{aligned}
    &X(Z)X= -Z \quad &\text{and} \quad \quad &X(-Z)X = Z \\
    &X(X)X = X \quad &\text{and} \quad \quad &X(-X)X = -X.
\end{aligned}
\end{equation*}
These four equations together say that conjugating by $X$ has the action $(\bm{\bullet}\leftrightarrow\bm{\circ})$ as desired.

For $Z$, conjugating all four $X$ equations above by $H$ and appealing to rule (H)\footnote{For example: $HXZXH = H(-Z)H \quad \stackrel{\text{(H)}}{\implies} \quad ZHZXH = -X \quad \stackrel{\text{(H)}}{\implies} \quad ZXHXH = -X \quad \stackrel{\text{(H)}}{\implies} \quad ZXZ = -X.$} gives
\begin{equation*}
\begin{aligned}
    &Z(Z)Z= Z \quad &\text{and} \quad \quad &Z(-Z)Z = -Z \\
    &Z(X)Z = -X \quad &\text{and} \quad \quad &Z(-X)Z = X.
\end{aligned}
\end{equation*}
These four equations together say that conjugating by $Z$ has the action $(\bm{\oplus}\leftrightarrow\bm{\ominus})$ as desired.

We can now apply these one-qubit flipflop maps to rules from the ruleset to generate flipflopped rules. This means conjugating every wire containing explicit symbols by $Z$, $X$, $H$, or any combination thereof. The result is exactly what you get from the flipflop symmetry principles as origininally stated. We show a few examples and cut the rest for brevity since the proofs are all identical -- conjugate, apply the identities above to push the flipflop operators through the circuit step by step, and then cancel them at the end.

\begin{equation*}
\begin{aligned}
\text{\textbf{(c')} (c, flopped)} \begin{quantikz}
    & \gate{H} & \phase{\alpha} & \ophase{\beta} & \gate{H} &
\end{quantikz}
&=
\begin{quantikz}
    & \gate{H} & \ophase{\beta} & \phase{\alpha} & \gate{H} &
\end{quantikz} \\
\begin{quantikz}
    & \targ{\alpha} & \gate{H} & \ophase{\beta} & \gate{H} &
\end{quantikz}
&=
\begin{quantikz}
    & \targm{\beta} & \gate{H} & \phase{\alpha} & \gate{H} &
\end{quantikz} \\
\begin{quantikz}
    & \targ{\alpha} & \targm{\beta} & \gate{H} & \gate{H} &
\end{quantikz}
&=
\begin{quantikz}
    & \targm{\beta} & \targ{\alpha} & \gate{H} & \gate{H} &
\end{quantikz} \\
=
\begin{quantikz}
    & \targ{\alpha} & \targm{\beta} & 
\end{quantikz}
&=
\begin{quantikz}
    & \targm{\beta} & \targ{\alpha} & 
\end{quantikz}
\end{aligned}
\end{equation*}

Rules involving vertical lines are made easy to flipflop by rule (g). The commutation properties of the flipflop operators given above, together with rule (g), mean that they simply ignore the presence of vertical lines when moving through a symbol joined vertically to something else. The symbol in question is flipflopped exactly as if it were not connected to anything by a vertical line.

\begin{equation*}
\begin{aligned}
\text{\textbf{(d')} (d, flipped)}
\begin{quantikz}[column sep=0.1cm,row sep=0.25cm]
    & \targ{} & \ctrl{1} & \ctrl{1} & \targ{} & \\
    & & \gate{U} & \gate{V} & &
\end{quantikz}
&=
\begin{quantikz}[column sep=0.1cm,row sep=0.3cm]
    & \targ{} & \ctrl{1} & \targ{} & \\
    & & \gate{U\circ V} & &
\end{quantikz} \\
\begin{quantikz}[column sep=0.1cm,row sep=0.25cm]
    & \octrl{1} & \targ{} & \ctrl{1} & \targ{} & \\
    & \gate{U} & & \gate{V} & &
\end{quantikz}
&=
\begin{quantikz}[column sep=0.1cm,row sep=0.3cm]
    & \octrl{1} & \targ{} & \targ{} & \\
    & \gate{U\circ V} & & &
\end{quantikz} \\
\begin{quantikz}[column sep=0.1cm,row sep=0.25cm]
    & \octrl{1} & \octrl{1} & \targ{} & \targ{} & \\
    & \gate{U} & \gate{V} & & &
\end{quantikz}
&=
\begin{quantikz}[column sep=0.1cm,row sep=0.3cm]
    & \octrl{1} & \\
    & \gate{U\circ V} &
\end{quantikz} \\
\begin{quantikz}[column sep=0.1cm,row sep=0.25cm]
    & \octrl{1} & \octrl{1} & \\
    & \gate{U} & \gate{V} &
\end{quantikz}
&=
\begin{quantikz}[column sep=0.1cm,row sep=0.3cm]
    & \octrl{1} & \\
    & \gate{U\circ V} &
\end{quantikz}
\end{aligned}
\end{equation*}

\begin{equation*}
\begin{aligned}
\text{\textbf{(d'')} (d, flopped then flipped)}
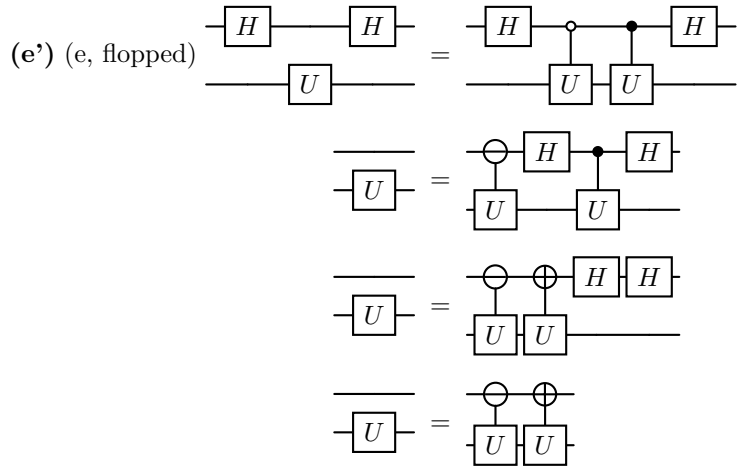
\begin{quantikz}[column sep=0.1cm,row sep=0.25cm]
    & \gate{H} & \octrl{1} & \octrl{1} & \gate{H} & \\
    & & \gate{U} & \gate{V} & &
\end{quantikz}
&=
\begin{quantikz}[column sep=0.1cm,row sep=0.3cm]
    & \gate{H} & \octrl{1} & \gate{H} & \\
    & & \gate{U\circ V} & &
\end{quantikz} \\
\begin{quantikz}[column sep=0.1cm,row sep=0.25cm]
    & \trgm{1} & \gate{H} & \octrl{1} & \gate{H} & \\
    & \gate{U} & & \gate{V} & &
\end{quantikz}
&=
\begin{quantikz}[column sep=0.1cm,row sep=0.3cm]
    & \trgm{1} & \gate{H} & \gate{H} & \\
    & \gate{U\circ V} & & &
\end{quantikz} \\
\begin{quantikz}[column sep=0.1cm,row sep=0.25cm]
    & \trgm{1} & \trgm{1} & \gate{H} & \gate{H} & \\
    & \gate{U} & \gate{V} & & &
\end{quantikz}
&=
\begin{quantikz}[column sep=0.1cm,row sep=0.3cm]
    & \trgm{1} & \\
    & \gate{U\circ V} &
\end{quantikz} \\
\begin{quantikz}[column sep=0.1cm,row sep=0.25cm]
    & \trgm{1} & \trgm{1} & \\
    & \gate{U} & \gate{V} &
\end{quantikz}
&=
\begin{quantikz}[column sep=0.1cm,row sep=0.3cm]
    & \trgm{1} & \\
    & \gate{U\circ V} &
\end{quantikz}
\end{aligned}
\end{equation*}

\begin{equation*}
\begin{aligned}
\text{\textbf{(e')} (e, flopped)}
\begin{quantikz}
    & \gate{H} & & \gate{H} & \\
    & & \gate{U} & &
\end{quantikz}
&=
\begin{quantikz}
    & \gate{H} & \octrl{1} & \ctrl{1} & \gate{H} & \\
    & & \gate{U} & \gate{U} & &
\end{quantikz} \\
\begin{quantikz}
    & & \\
    & \gate{U} &
\end{quantikz}
&=
\begin{quantikz}[column sep=0.1cm,row sep=0.25cm]
    & \trgm{1} & \gate{H} & \ctrl{1} & \gate{H} & \\
    & \gate{U} & & \gate{U} & &
\end{quantikz} \\
\begin{quantikz}
    & & \\
    & \gate{U} &
\end{quantikz}
&=
\begin{quantikz}[column sep=0.1cm,row sep=0.25cm]
    & \trgm{1} & \trg{1} & \gate{H} & \gate{H} & \\
    & \gate{U} & \gate{U} & & &
\end{quantikz} \\
\begin{quantikz}
    & & \\
    & \gate{U} &
\end{quantikz}
&=
\begin{quantikz}[column sep=0.1cm,row sep=0.25cm]
    & \trgm{1} & \trg{1} & \\
    & \gate{U} & \gate{U} &
\end{quantikz}
\end{aligned}
\end{equation*}

...and so on for all remaining flipflopped rules. The only rule that deserve special mention with regard to flipflop symmetry is rule (pdm) because its involves a measurement (as does (mb), but that rule has no symbols to flipflop). Rule (pdm) has only one symbol, and it sits on a wire that is later measured and hence destroyed. This rule can still be flipflopped by applying just one flipflop operator on the left and appealing to rule (mb), but the resulting flipflopped rules have no more semantic content than rule (pdm) and amount to a complete relabelling of the conditions and measurement outcomes. We show two examples here.

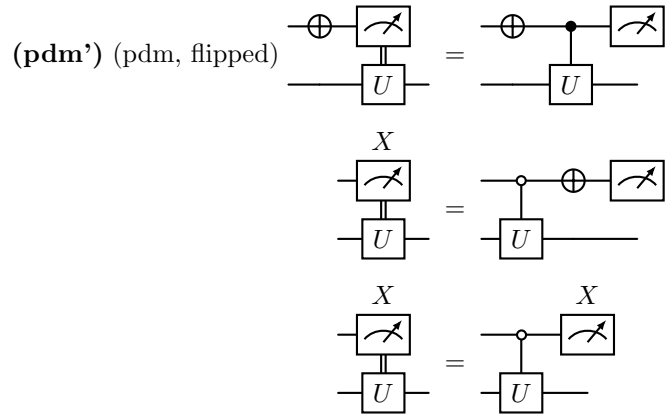
\begin{equation*}
\begin{aligned}
\text{\textbf{(pdm')} (pdm, flipped)}
\begin{quantikz}
    & \targ{} & \meter{}\vcw{1} \\
    & & \gate{U} &
\end{quantikz}
&=
\begin{quantikz}
    & \targ{} & \ctrl{1} & \meter{} \\
    & & \gate{U} & 
\end{quantikz}
\\
\begin{quantikz}
    & \meter{X}\vcw{1} \\
    & \gate{U} &
\end{quantikz}
&=
\begin{quantikz}
    & \octrl{1} & \targ{} & \meter{} \\
    & \gate{U} & & 
\end{quantikz}
\\
\begin{quantikz}
    & \meter{X}\vcw{1} \\
    & \gate{U} &
\end{quantikz}
&=
\begin{quantikz}
    & \octrl{1} & \meter{X} \\
    & \gate{U} & 
\end{quantikz}
\end{aligned}
\end{equation*}

Rule (pdm') gives a rule saying that measuring in the $\{\ket{1},\ket{0}\}$ basis followed by performing a classically controlled $U$ if the outcome is $\ket{0}$ is equivalent to first applying a quantum controlled operation implementing the logic \verb|if |$\ket{0}_A$\verb| do |$U_B$, followed by measuring in the $\{\ket{1},\ket{0}\}$ basis. This rule is rule (pdm) with all labels flipped $(0\leftrightarrow 1)$.

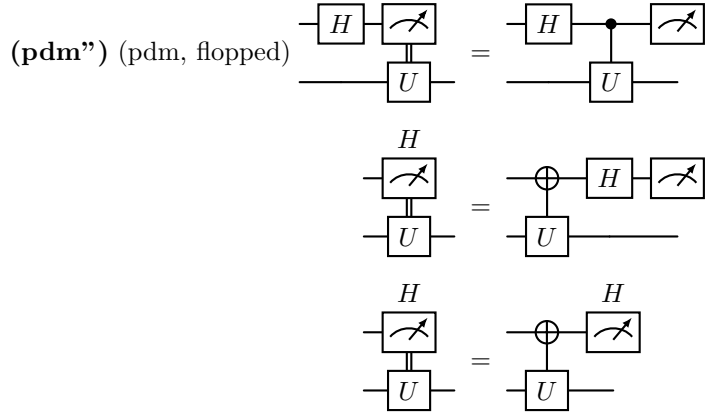
\begin{equation*}
\begin{aligned}
\text{\textbf{(pdm'')} (pdm, flopped)}
\begin{quantikz}
    & \gate{H} & \meter{}\vcw{1} \\
    & & \gate{U} &
\end{quantikz}
&=
\begin{quantikz}
    & \gate{H} & \ctrl{1} & \meter{} \\
    & & \gate{U} & 
\end{quantikz}
\\
\begin{quantikz}
    & \meter{H}\vcw{1} \\
    & \gate{U} &
\end{quantikz}
&=
\begin{quantikz}
    & \trg{1} & \gate{H} & \meter{} \\
    & \gate{U} & & 
\end{quantikz}
\\
\begin{quantikz}
    & \meter{H}\vcw{1} \\
    & \gate{U} &
\end{quantikz}
&=
\begin{quantikz}
    & \trg{1} & \meter{H} \\
    & \gate{U} & 
\end{quantikz}
\end{aligned}
\end{equation*}

Rule (pdm'') gives a rule saying that measuring in the $\{\ket{+},\ket{-}\}$ basis followed by performing a classically controlled $U$ if the outcome is $\ket{-}$ is equivalent to first applying a quantum controlled operation implementing the logic \verb|if |$\ket{-}_A$\verb| do |$U_B$, followed by measuring in the $\{\ket{+},\ket{-}\}$ basis. This rule is rule (pdm) with all labels flopped $(0\leftrightarrow +$ and $1\leftrightarrow -)$.

As a last note which is simply interesting, applying the flipflop operators to SWAP reveals that it takes the same form regardless which symbols are used to define it.

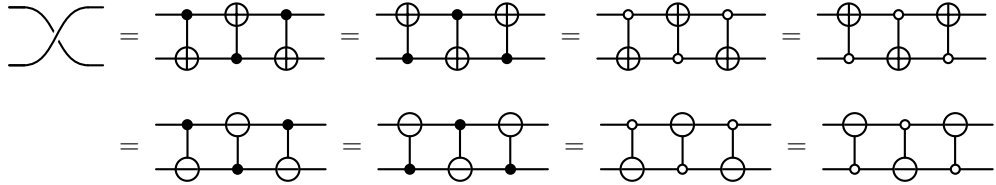
\begin{equation*}
\begin{aligned}
\begin{quantikz}[column sep=0.2cm]
& \permute{2,1} & \ghost{H} \\
& & \ghost{H}
\end{quantikz}
&=
\begin{quantikz}
& \ctrl{1} & \targ{} & \ctrl{1} & \\
& \targ{} & \ctrl{-1} & \targ{} &
\end{quantikz}
=
\begin{quantikz}
& \targ{} & \ctrl{1} & \targ{} & \\
& \ctrl{-1} & \targ{} & \ctrl{-1} &
\end{quantikz}
=
\begin{quantikz}
& \octrl{1} & \targ{} & \octrl{1} & \\
& \targ{} & \octrl{-1} & \targ{} &
\end{quantikz}
=
\begin{quantikz}
& \targ{} & \octrl{1} & \targ{} & \\
& \octrl{-1} & \targ{} & \octrl{-1} &
\end{quantikz} \\
&=
\begin{quantikz}
& \ctrl{1} & \targm{} & \ctrl{1} & \\
& \targm{} & \ctrl{-1} & \targm{} & 
\end{quantikz}
=
\begin{quantikz}
& \targm{} & \ctrl{1} & \targm{} & \\
& \ctrl{-1} & \targm{} & \ctrl{-1} &
\end{quantikz}
=
\begin{quantikz}
& \octrl{1} & \targm{} & \octrl{1} & \\
& \targm{} & \octrl{-1} & \targm{} & 
\end{quantikz}
=
\begin{quantikz}
& \targm{} & \octrl{1} & \targm{} & \\
& \octrl{-1} & \targm{} & \octrl{-1} &
\end{quantikz}
\end{aligned}
\end{equation*}

\section{On the notion of control}\label{appendix:onthenotionofcontrol}
The connection between controlled unitaries and tensor product Hamiltonians is elaborated on. A formal definition for a unitary gate being controlled is introduced, based on a notion we call partial eigenvectors, and some small proofs are provided.

\subsection{Semi-classical interpretations of controlled operations}

Typically in literature controlled unitaries are defined in the following way: a controlled unitary is a unitary that is conditionally applied to a set of target qubits depending on the state of one or more control qubits, denoted by $\bm{\bullet}$ (or $\bm{\circ}$). The CNOT, for example, has one control qubit and one target qubit, and has the action $\ket{a}\ket{b}\mapsto\ket{a}\ket{a\oplus b}$ on the computational basis. The first qubit is the control and does not change, while the second target qubit is where the action happens. While this simple description is sufficient to uniquely define CNOT, it neglects some subtleties CNOT has, like the fact that CNOT also has the seemingly reversed action $\ket{a}\ket{b}\mapsto\ket{a\oplus b}\ket{b}$ on the diagonal $\ket{+},\ket{-}$ basis, calling into question the meaning of `control' and `target'.

The interpretation of CNOT as a controlled-$X$ gate with a control wire and a target wire is as old as the gate itself, as this is how CNOT was first described when it was introduced \cite{feynman_quantum_1986}. We will call descriptions of the kind (i.e. descriptions of the kind $\verb|if | \ket{1}\verb| do | X$) `semi-classical interpretations' of quantum controlled gates; `classical' because they come originally from classical controlled operations $\verb|if x do y|$, but only `semi-classical' because they are not unique. Unlike their classical analogues, the notions of control and target are dependent on a local choice of basis, and quantum controlled unitaries always admit multiple such interpretations.


To motivate this section further, we show an explicit calculation converting CNOT's first semi-classical logic to its second. The fact that every step of this calculation is connected by an equals sign means that these two semi-classical interpretations are not different things, nor related by some kind of unitary transformation: they are \textit{mathematically equivalent}. Either one uniquely defines the action of CNOT, despite the fact that they interchange the roles of the control and target qubits.

CNOT $=\begin{array}{l}
\verb|if | \ket{1}_A \\
\verb| do | X_B
\end{array} = \dyad{0}\otimes\eye + \dyad{1}\otimes X = \dyad{0}\otimes(\dyad{+}+\dyad{-}) + \dyad{1}\otimes(\dyad{+}-\dyad{-})$
\hspace{2cm}$= (\dyad{0}+\dyad{1})\otimes\dyad{+} + (\dyad{0}-\dyad{1})\otimes\dyad{-} = \eye\otimes\dyad{+} + Z\otimes \dyad{-} = \begin{array}{l}
\verb|if | \ket{-}_B \\
\verb| do | Z_A
\end{array}$

It is interesting to note that the state activating the control wire ($\ket{1}$, or $\ket{-}$) appears in the CNOT's Hamiltonian as a projector, and that the resulting conditional operation performed on the other wire can be obtained from simply deleting the `control' state from the joint Hamiltonian. That is, 
\begin{equation*}
\begin{aligned}
    \text{CNOT} = e^{i\pi\big(\dyad{1}\otimes\dyad{-}\big)} &= \text{ project onto } \dyad{1}_A \text{ then apply } (e^{i\pi\dyad{-}})_B \\
    \Bigg(&= \text{ project onto } \dyad{-}_B \text{ then apply } (e^{i\pi\dyad{1}})_A\Bigg)
\end{aligned}
\end{equation*}
This cheater's way of reading off semi-classical interpretations from exponential forms (or their diagrams) turns out to be general. This is how all of the semi-classical interpretations shown in this paper were found. For each possible choice of control qubits, one simply `picks out' those qubits from the Hamiltonian and deletes them from the exponent. The resulting gate that is conditionally applied depending on the deleted piece is the exponential of whatever remains. In this way semi-classical interpretations can be easily read off from the associated exponential form, including seemingly complicated ones like the gate shown in section 2.2 that is a $\bm{\circ}$, $U$, and $\bm{\oplus}$ joined by a vertical line, which has seven complicated-looking semi-classical interpretations.

The notion of control qubits as those qubits that do not change (in some chosen basis) will form the base of our definition of a controlled operation. To formalize this we introduce the concept of a partial eigenbasis.\footnote{Everything shown in this section only applies to 2-qubit operations. It remains for future work to extend these definitions to $n$ qubits.}

\subsection{Partial eigenbases}

Let $U$ be a two-qubit gate acting on qubits $A$ and $B$, and let $\ket{\psi\phi}$ be a separable state. $\ket{\psi}_A$ is a \textit{partial eigenvector} of $U$ if
\begin{align*}
    U\ket{\psi\phi}_{AB} = \ket{\psi}_A \otimes U_\psi\ket{\phi}_B \quad \quad \forall\ket{\phi},
\end{align*}
where $U_\psi$ is a single qubit unitary acting on qubit $B$ which may depend on $\ket{\psi}_A$.


Now let $\{\ket{a_0},\ket{a_1}\}$ be a basis on qubit $A$. If both $\ket{a_0}$ and $\ket{a_1}$ are partial eigenvectors of $U$, we say $\{\ket{a_0},\ket{a_1}\}$ is a \textit{partial eigenbasis} of $U$. This means
\begin{align*}
    U\ket{a_0\phi} = \ket{a_0} \otimes U_{a_0}\ket{\phi} \quad \textbf{and} \quad U\ket{a_1\phi} = \ket{a_1} \otimes U_{a_1}\ket{\phi},
\end{align*}
meaning for a generic state $\ket{\psi} = \alpha\ket{a_0b_0}+\beta\ket{a_0b_1}+\gamma\ket{a_1b_0}+\delta\ket{a_1b_1}$ in an arbitrary basis $\{\ket{b_0},\ket{b_1}\}$ on qubit $B$,
\begin{align*}
    U\ket{\psi} &= \ket{a_0} \otimes U_{a_0}(\alpha \ket{b_0} + \beta\ket{b_1}) + \ket{a_1} \otimes U_{a_1}(\gamma \ket{b_0} + \delta \ket{b_1}).
\end{align*}
This means $U$ acts like a block diagonal matrix with respect to $\{\ket{a_0},\ket{a_1}\}$, by which we mean $U$ has the form
\begin{align}
U = (V\otimes\eye)\begin{bmatrix}
    U_{a_0} & 0 \\
    0 & U_{a_1}
\end{bmatrix}(V^\dagger\otimes\eye) \quad \text{where} \quad V = \begin{bmatrix}
    | & | \\
    \ket{a_0} & \ket{a_1} \\
    | & |
\end{bmatrix}.\label{eq1}
\end{align}
Now we can formally define a controlled operation. We'll say $U$ is \textit{controlled} if a partial eigenbasis of $U$ exists. In particular, $U$ is controlled by qubit $A$ in the basis $\{\ket{a_0},\ket{a_1}\}$ if $\{\ket{a_0},\ket{a_1}\}$ is a partial eigenbasis of $U$. $U$ is \textit{uncontrolled} if no partial eigenbasis of $U$ exists. Let's look at an example. 

Claim 1: CNOT$_{AB}$ is controlled (in the computational basis $\{\ket{0},\ket{1}\}$ by qubit $A$). 
\begin{proof}
To demonstrate this we must show that $\{\ket{0},\ket{1}\}$ on qubit $A$ is a partial eigenbasis of CNOT.
\begin{align*}
    \text{CNOT}\ket{0\phi} = \ket{0} \otimes \eye\ket{\phi} \quad \textbf{and} \quad \text{CNOT}\ket{1\phi} = \ket{1} \otimes X\ket{\phi} \quad \forall\ket{\phi}
\end{align*}
This completes the proof. In terms of equation (\ref{eq1}), $V = \eye$, $U_0 = \eye$, and $U_1 = X$.
\begin{align*}
\text{CNOT} = (\eye\otimes\eye)\begin{bmatrix}
    \eye & 0 \\
    0 & X
\end{bmatrix}(\eye\otimes\eye).
\end{align*}
\end{proof}

Claim 2: CNOT$_{AB}$ is controlled in the $X$ basis $\{\ket{+},\ket{-}\}$ by qubit $B$.
\begin{proof}
To demonstrate this we must show that $\{\ket{+},\ket{-}\}$ on qubit $B$ is a partial eigenbasis of CNOT. It will be easier to show this with the upside-down CNOT, CNOT$_{BA}$, since our definitions -- particularly the ease of seeing a block diagonal form -- use the first qubit. Of course the definition extends in the obvious way to the second qubit if you prefer not to flip your gates upside-down.
\begin{align*}
    \text{CNOT}_{BA}\ket{+\phi} = \ket{+} \otimes \eye\ket{\phi} \quad \textbf{and} \quad \text{CNOT}_{BA}\ket{-\phi} = \ket{-} \otimes Z\ket{\phi} \quad \forall\ket{\phi}
\end{align*}
This completes the proof, since $\{\ket{+},\ket{-}\}$ on qubit $B$ being a partial eigenbasis of CNOT$_{BA}$ is equivalent to $\{\ket{+},\ket{-}\}$ on qubit $B$ being a partial eigenbasis of CNOT$_{AB}$. In terms of equation \ref{eq1}, $V = H$, $U_+ = \eye$, and $U_- = Z$.
\begin{align*}
\text{CNOT}_{BA} = (H\otimes\eye)\begin{bmatrix}
    \eye & 0 \\
    0 & Z
\end{bmatrix}(H\otimes\eye).
\end{align*}

\end{proof}

\subsection{Various proofs}

In this section provide various small proofs regarding the properties of controlled operations. 

\begin{lemma}
$V$ and $W$ are partial eigenbases of $U$ on qubits 1 and 2 respectively iff $V\otimes W$ is an eigenbasis of $U$.
\end{lemma}
\begin{proof}
$(\implies)$ Let $V$ and $W$ be partial eigenbases of $U$ on qubits 1 and 2.

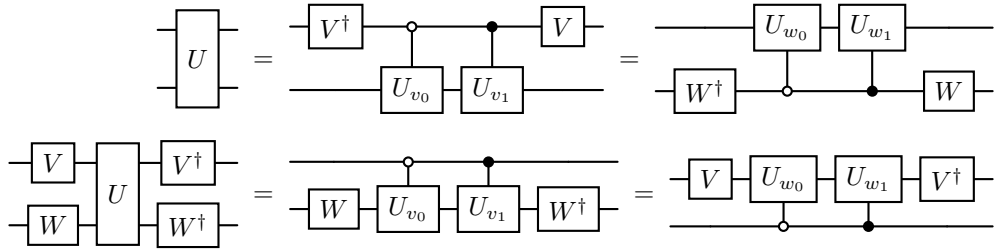
\begin{equation*}
\begin{aligned}
\begin{quantikz}
    & \gate[2]{U} & \\
    & \ghost{U} & 
\end{quantikz}
&=
\begin{quantikz}
    & \gate{V^\dagger} & \octrl{1} & \ctrl{1} & \gate{V} & \\
    & & \gate{U_{v_0}} & \gate{U_{v_1}} & &
\end{quantikz}
=
\begin{quantikz}
    & & \gate{U_{w_0}} & \gate{U_{w_1}} & & \\
    & \gate{W^\dagger} & \octrl{-1} & \ctrl{-1} & \gate{W} &
\end{quantikz}
\\
\begin{quantikz}
    & \gate{V} & \gate[2]{U} & \gate{V^\dagger} & \\
    & \gate{W} & \ghost{U} & \gate{W^\dagger} & 
\end{quantikz}
&=
\begin{quantikz}
    & & \octrl{1} & \ctrl{1} & & \\
    & \gate{W} & \gate{U_{v_0}} & \gate{U_{v_1}} & \gate{W^\dagger} &
\end{quantikz}
=
\begin{quantikz}
    & \gate{V} & \gate{U_{w_0}} & \gate{U_{w_1}} & \gate{V^\dagger} & \\
    & & \octrl{-1} & \ctrl{-1} & &
\end{quantikz}
\end{aligned}
\end{equation*}

Both sides of the last equals slign are block diagonal matrices, but in different rows and columns. The first is block diagonal in row/column $(1,2)$ and $(3,4)$, while the second is block diagonal in row/column $(1,3)$ and $(2,4)$. Explicitly,
\begin{align*}
(\eye\otimes W^\dagger)\begin{bmatrix}
(U_{v_0})_{00} & (U_{v_0})_{01} & 0 & 0 \\
(U_{v_0})_{10} & (U_{v_0})_{11} & 0 & 0 \\
0 & 0 & (U_{v_1})_{00} & (U_{v_1})_{01} \\
0 & 0 & (U_{v_1})_{10} & (U_{v_1})_{11}
\end{bmatrix}(\eye\otimes W) 
\\
=
(V^\dagger\otimes\eye)\begin{bmatrix}
(U_{w_0})_{00} & 0 & (U_{w_0})_{01} & 0 \\
0 & (U_{w_1})_{00} & 0 & (U_{w_1})_{01} \\
(U_{w_0})_{10} & 0 & (U_{w_0})_{11} & 0 \\
0 & (U_{w_1})_{10} & 0 & (U_{w_1})_{11}
\end{bmatrix}(V\otimes\eye)
\end{align*}

This implies all off-diagonal elements in both matrix products are zero, implying the off-diagonal elements of $(V^\dagger\otimes W^\dagger)U(V\otimes W)$ are zero, meaning the columns of $V\otimes W$ are eigenvectors of $U$. That is,

\begin{equation*}
\begin{quantikz}
    & \gate[2]{U} & \\
    & \ghost{U} & 
\end{quantikz}
=
\begin{quantikz}
    & \gate{V^\dagger} & \octrl[wire style={"\alpha"}]{1} & \octrl[wire style={"\beta"}]{1} & \ctrl[wire style={"\gamma"}]{1} & \ctrl[wire style={"\delta"}]{1} & \gate{V} & \\
    & \gate{W^\dagger} & \ocontrol{} & \control{} & \ocontrol{} & \control{} & \gate{W} &
\end{quantikz}
=
(V\otimes W)
\begin{bmatrix}
    e^{i\pi\alpha} & & & \\
    & e^{i\pi\beta} & & \\
    & & e^{i\pi\gamma} & \\
    & & & e^{i\pi\delta}
\end{bmatrix}
(V^\dagger\otimes W^\dagger)
\end{equation*}

i.e. $V\otimes W$ is an eigenbasis of $U$.


$(\impliedby)$ Let $V\otimes W$ be an eigenbasis of $U$.

\begin{equation*}
\begin{quantikz}
    & \gate[2]{U} & \\
    & \ghost{U} & 
\end{quantikz}
=
\begin{quantikz}
    & \gate{V^\dagger} & \octrl[wire style={"\alpha"}]{1} & \octrl[wire style={"\beta"}]{1} & \ctrl[wire style={"\gamma"}]{1} & \ctrl[wire style={"\delta"}]{1} & \gate{V} & \\
    & \gate{W^\dagger} & \ocontrol{} & \control{} & \ocontrol{} & \control{} & \gate{W} &
\end{quantikz}
\end{equation*}

Inserting $W^\dagger W$ on the lower wire and interpreting the eigenphases as sitting on the lower wire, we find that $V$ that is a partial eigenbasis of $U$ with corresponding conditional operations $U_{v_0} := W\begin{bmatrix}e^{i\pi\alpha} & \\ & e^{i\pi\beta}\end{bmatrix}W^\dagger$ and $U_{v_1} := W\begin{bmatrix}e^{i\pi\gamma} & \\ & e^{i\pi\delta}\end{bmatrix}W^\dagger$.

\begin{equation*}
\begin{quantikz}
    & \gate{V^\dagger} & \octrl{1} & \octrl{1} & & & & \ctrl{1} & \ctrl{1} & \gate{V} & \\
    & \gate{W^\dagger}\gategroup[1,steps=4,style={dashed,rounded corners},label style={label position=below}]{$=: U_{v_0}$} & \ophase{\alpha} & \phase{\beta} & \gate{W} & & \gate{W^\dagger}\gategroup[1,steps=4,style={dashed,rounded corners},label style={label position=below}]{$=: U_{v_1}$} & \ophase{\gamma} & \phase{\delta} & \gate{W} &
\end{quantikz}
=
\begin{quantikz}
    & \gate{V^\dagger} & \octrl{1} & \ctrl{1} & \gate{V} & \\
    & & \gate{U_{v_0}} & \gate{U_{v_1}} & &
\end{quantikz}
\end{equation*}

Similarly for the top wire: inserting $V^\dagger V$ and regrouping the diagonal gates, we find that $W$ is a partial eigenbasis of $U$ with corresponding conditional operations $U_{w_0} := W\begin{bmatrix}e^{i\pi\alpha} & \\ & e^{i\pi\gamma}\end{bmatrix}W^\dagger$ and $U_{w_1} := W\begin{bmatrix}e^{i\pi\beta} & \\ & e^{i\pi\delta}\end{bmatrix}W^\dagger$.

\begin{equation*}
\begin{quantikz}
    & \gate{V^\dagger}\gategroup[1,steps=4,style={dashed,rounded corners}]{$=: U_{w_0}$} & \ophase{\alpha} & \phase{\gamma} & \gate{V} & & \gate{V^\dagger}\gategroup[1,steps=4,style={dashed,rounded corners}]{$=: U_{w_1}$} & \ophase{\beta} & \phase{\delta} & \gate{V} & \\
    & \gate{W^\dagger} & \octrl{-1} & \octrl{-1} & & & & \ctrl{-1} & \ctrl{-1} & \gate{W} &
\end{quantikz}
=
\begin{quantikz}
    & & \gate{U_{w_0}} & \gate{U_{w_1}} & & \\
    & \gate{W^\dagger} & \octrl{-1} & \ctrl{-1} & \gate{W} &
\end{quantikz}
\end{equation*}

\end{proof}

\begin{lemma}
    $V\otimes W$ is an eigenbasis of $U$ iff $V$ is a partial eigenbasis of $U$ and the conditional operations $U_{v_0},U_{v_1}$ it implements satisfy $[U_{v_0},U_{v_1}]=0$.
\end{lemma}
\begin{proof}
$(\implies)$ Let $V\otimes W$ be an eigenbasis of $U$. 
\begin{align*}
U = (V\otimes W)
\begin{bmatrix}
    e^{i\pi\alpha} & & & \\
    & e^{i\pi\beta} & & \\
    & & e^{i\pi\gamma} & \\
    & & & e^{i\pi\delta}
\end{bmatrix}
(V^\dagger\otimes W^\dagger) \\
= (V\otimes \eye)
\begin{bmatrix}
    W\begin{bmatrix}
        e^{i\pi\alpha} & \\
        & e^{i\pi\beta}
    \end{bmatrix}W^\dagger & \\
    & W\begin{bmatrix}
        e^{i\pi\alpha} & \\
        & e^{i\pi\beta}
    \end{bmatrix}W^\dagger
\end{bmatrix}
(V^\dagger\otimes \eye)
\end{align*}

Defining $U_{v_0} := W\begin{bmatrix}e^{i\pi\alpha} & \\ & e^{i\pi\beta}\end{bmatrix}W^\dagger$ and $U_{v_1} := W\begin{bmatrix}e^{i\pi\gamma} & \\ & e^{i\pi\delta}\end{bmatrix}W^\dagger$ we see that $V$ is a partial eigenbasis of $U$. Furthermore, since $W$ simultaneously diagonalizes $U_{v_0}$ and $U_{v_1}$, $U_{v_0}$ and $U_{v_1}$ commute and so satisfy $[U_{v_0},U_{v_1}]=0$.

$(\impliedby)$ Let $V$ be a partial eigenbasis of $U$, and let $[U_{v_0},U_{v_1}]=0$. 
\begin{align*}
U = (V\otimes\eye)\begin{bmatrix}
    U_{v_0} & 0 \\
    0 & U_{v_1}
\end{bmatrix}(V^\dagger\otimes\eye)
\end{align*}

Since $U_{v_0}$ and $U_{v_1}$ commute, there exists a matrix which simultaneously diagonalizes them. Let this matrix be called $W$.
\begin{align*}
U = (V\otimes \eye)
\begin{bmatrix}
    W\begin{bmatrix}
        e^{i\pi\alpha} & \\
        & e^{i\pi\beta}
    \end{bmatrix}W^\dagger & \\
    & W\begin{bmatrix}
        e^{i\pi\gamma} & \\
        & e^{i\pi\delta}
    \end{bmatrix}W^\dagger
\end{bmatrix}
(V^\dagger\otimes \eye) \\
= (V\otimes W)
\begin{bmatrix}
    e^{i\pi\alpha} & & & \\
    & e^{i\pi\beta} & & \\
    & & e^{i\pi\gamma} & \\
    & & & e^{i\pi\delta}
\end{bmatrix}
(V^\dagger\otimes W^\dagger)
\end{align*}
meaning $V\otimes W$ is an eigenbasis of $U$.

\end{proof}

\begin{corollary}
    $V$ and $W$ are partial eigenbases of $U$ on qubits 1 and 2 respectively iff $V$ is a partial eigenbasis of $U$ and the conditional operations $U_{v_0},U_{v_1}$ it implements satisfy $[U_{v_0},U_{v_1}]=0$.
\end{corollary}
\begin{proof}
    Direct consequence of two preceding Lemmas.
\end{proof}

\begin{corollary}
    If $V$ is a partial eigenbasis of $U$ on qubit 1, then a partial eigenbasis $W$ of $U$ on qubit 2 exists iff $[U_{v_0},U_{v_1}]=0$.
\end{corollary}
\begin{proof}
    Corollary 1 is a statement of the form $P \wedge Q \iff P \wedge R$. If $P$ is true (namely, if $V$ is a partial eigenbasis of $U$ on qubit 1) then Corollary 1 reduces to $Q \iff R$, which is the statement of Corollary 2.
\end{proof}

Corollary 2 gives us an economical way to check when a controlled operation can have its control switched to the other wire; one simply has to check whether the conditional operations it implements commute with one another. In particular, it tells us that CNOT and CZ can have their controls reversed because one of the two conditional operations they implement is $\eye$, which commutes with everything. Any standard controlled unitary on two qubits, i.e. a two qubit gate which implements the logic $\verb|if | \ket{\phi}_A \verb| do | U_B$
has a reversible control for the same reason; the conditional operation implicitly happening in the orthogonal $\ket{\phi_\perp}_A$ is $\eye$.

\subsection{Classes of control}

An interesting feature of this definition of controlled operations is that it classifies two-qubit operations into three classes fairly naturally. The first are gates like CNOT which are controlled on both wires; as we've seen in the previous section, these gates have tensor product Hamiltonians. This class includes all of the gates usually thought of as controlled like CNOT and CZ. The second class is an interesting intermediary class of gates that are controlled on one wire but not the other. These gates are lopsided/assymetric in the way that CNOT is sometimes mistakenly taken to be; unlike CNOT, they distinguish unequivocally between the two qubits in a way that doesn't depend on a local choice of basis. The third class are the uncontrolled operations, which includes SWAP and its siblings fSWAP, iSWAP.

\begin{proposition}
    Controlled operations can be classified into two classes: (1) controlled by both wires and (2) controlled by one wire. Uncontrolled operations (3) make up a third class.
\end{proposition}

(1) Ex. CNOT = 
\begin{quantikz}
    & \ctrl{1} & \\
    & \targ{} &
\end{quantikz}
$\begin{array}{lll}
    = \dyad{0}\otimes\eye + \dyad{1}\otimes X & & \text{(controlled by A in 0/1 basis)} \\
    = \eye\otimes\dyad{+} + Z\otimes\dyad{-} & & \text{(controlled by B in +/- basis)}
\end{array}$

(2) Ex. Bell = 
\begin{quantikz}
    & \gate{H} & \ctrl{1} & \\
    & & \targ{} &
\end{quantikz}
$\begin{array}{lll}
    \neq \dyad{\phi}\otimes U_\phi + \dyad{\phi_\perp}\otimes U_{\phi_\perp} & \forall\phi & \text{(not controlled by A)} \\
    = H\otimes\dyad{+} + HZ\otimes\dyad{-} & & \text{(controlled by B in +/- basis)}
\end{array}$

(3) Ex. SWAP = 
\begin{quantikz}[column sep=0.2cm]
& \permute{2,1} & \ghost{H} \\
& & \ghost{H}
\end{quantikz}
$\begin{array}{lll}
    \neq \dyad{\phi}\otimes U_\phi + \dyad{\phi_\perp}\otimes U_{\phi_\perp} & \forall\phi & \text{(not controlled by A)} \\
    \neq U_\phi\otimes\dyad{\phi} + U_{\phi_\perp}\otimes\dyad{\phi_\perp} & \forall\phi & \text{(not controlled by B)} \\
\end{array}$

(1) CNOT is controlled by wire A in the 0/1 basis and controlled by wire B in the +/- basis.

(2) Bell is controlled on wire B in the +/- basis, but the operations it implements on wire A do not commute: $\comm{H}{HZ}\neq 0$. By Corollary 2 no partial eigenbasis exists on wire A, meaning A is not controlled on wire A.

(3) SWAP has the action $\ket{\phi\psi} \mapsto \ket{\psi\phi} \forall \phi,\psi$ so no partial eigenbasis can exist on either wire.

\begin{lemma}
    $U_{AB}$ uncontrolled $\implies$ there exists an entangled eigenvector of $U_{AB}$.
\end{lemma}
\begin{proof}
    (By contradiction) Assume $U_{AB}$ is a two-qubit unitary with no entangled eigenvectors, i.e. it has separable eigenvectors denoted $\{\ket{u_1},\ket{u_2},\ket{u_3},\ket{u_4}\} = \{\ket{a_1b_1},\ket{a_2b_2},\ket{a_3b_3},\ket{a_4b_4}\}$. Being an eigenbasis of $U_{AB}$, these eigenvectors satisfy $\braket{u_i}{u_j} = \delta_{ij}$, and in particular they must be pairwise orthogonal:
    

    \begin{align*}
    \braket{u_1}{u_2} = 0 \quad &\implies \quad \braket{a_1}{a_2} = 0 \text{  or } \braket{b_1}{b_2} = 0 \\
    \braket{u_1}{u_3} = 0 \quad &\implies \quad \braket{a_1}{a_3} = 0 \text{  or } \braket{b_1}{b_3} = 0 \\
    \braket{u_1}{u_4} = 0 \quad &\implies \quad \braket{a_1}{a_4} = 0 \text{  or } \braket{b_1}{b_4} = 0 \\
    \braket{u_2}{u_3} = 0 \quad &\implies \quad \braket{a_2}{a_3} = 0 \text{  or } \braket{b_2}{b_3} = 0 \\
    \braket{u_2}{u_4} = 0 \quad &\implies \quad \braket{a_2}{a_4} = 0 \text{  or } \braket{b_2}{b_4} = 0 \\
    \braket{u_3}{u_4} = 0 \quad &\implies \quad \braket{a_3}{a_4} = 0 \text{  or } \braket{b_3}{b_4} = 0 
    \end{align*}

    Each orthogonality contraint must be satisfied in either the $A$ subspace or the $B$ subspace, or both. Within each subspace we are further constrained by the fact that $A$ and $B$ are qubits, meaning if $\ket{a_1}\perp\ket{a_2}$ and $\ket{a_1}\perp\ket{a_3}$ then $\ket{a_2} = \ket{a_3}$.

    If 5 or 6 of these constraints are satisfied on either qubit we immediately reach a contradiction: the constraints contain an orthogonality cycle of odd length (e.g. $\ket{a_1}\perp\ket{a_2}\perp\ket{a_3}\perp\ket{a_1}$), meaning at least one $\ket{u_i}$ must be orthogonal to itself, which is a contradiction. If 0 or 1 constraints are satisfied on either qubit, the other must satisfy 5 or 6 constraints and we reach the same contradiction.

    If 3 constraints are satisfied on one qubit which do not create an odd-length orthogonality cycle, the fourth constraint is implied, e.g. $(\ket{a_1}\perp\ket{a_2} \text{ and } \ket{a_1}\perp\ket{a_3} \text{ and } \ket{a_2}\perp\ket{a_4}) \implies \ket{a_3}\perp\ket{a_4}$.
    
    This leaves us with only two cases to consider: 4 or 2 constraints being satisfied on either qubit. If 4 constraints are satisfied on one qubit, there is partial eigenbasis on that qubit, e.g. $(\ket{a_1}\perp\ket{a_2} \text{ and } \ket{a_1}\perp\ket{a_3} \text{ and } \ket{a_2}\perp\ket{a_4} \text{ and } \ket{a_1}\perp\ket{a_2}) \implies (\ket{a_1}=\ket{a_4} \text{ and } \ket{a_2}=\ket{a_3})$, meaning $\{\ket{a_1},\ket{a_2}\}$ is a partial eigenbasis of $U_AB$.
    If only 2 constraints are satisfied, no partial eigenbasis exists on that qubit.

    If both $A$ and $B$ satisfy 4 constraints each, we are in case (1) where partial eigenbases exist on both qubits. If $A$ satisfies 4 constraints and $B$ satisfies 2 constraints, we are in case (2) where only one partial eigenbasis exists on qubit $A$. For $U_{AB}$ to be uncontrolled it must satisfy just 2 constraints on both qubits, but since it must satisfy 6 total constraints and $2+2<6$, no such uncontrolled $U_{AB}$ can exist. Therefore $U$ is controlled.
\end{proof}

\section{Recap on Hamiltonians for the Computer Science reader}\label{appendix:hamiltonians}

Although unitary gates like $X$, $T$, or CNOT are the standard fare in quantum computing for manipulating a register of qubits, the time evolution of a generic quantum system $\ket{\psi(t)}$ is sometimes given in terms of another quantity called the \textit{Hamiltonian} via the Schrödinger equation
\begin{equation*}
    H\ket{\psi(t)} = i\hbar\frac{d}{dt}\ket{\psi(t)}.
\end{equation*}
The Hamiltonian $H$ is a Hermitian operator, and physically speaking is the observable describing the system's total energy. When $H$ is itself not changing with time, the Schrödinger equation has the solution
\begin{equation*}
    \ket{\psi(t)} = e^{-itH/\hbar}\ket{\psi(0)}
\end{equation*}
where and $e^{-itH/\hbar} =: U$ is the unitary operator evolving $\ket{\psi(0)}$ in time. Although $H$ and $U$ can be infinite dimensional in general (for example when they act on a continuous region of space), in quantum computing we only consider finite dimensional Hamiltonians, since our systems are finite collections of qubits.\footnote{or other discrete, finite systems like dual-rail encoded photons.}

The minus sign and the $\hbar$ in the exponent of $U$ are absorbed into $H$ in this paper to simplify the notation.\footnote{The $\hbar$ has physical units of energy$\cdot$time, so the remaining Hamiltonian in the exponent of $U=e^{i\pi tH}$ has units of time$^{-1}$ as needed to make the exponent dimensionless. Mathematicians studying Lie groups prefer to also absorb the $i$ into $H$, making it skew-Hermitian and leaving behind just $U = e^H$. Physicists typically leave the $i$ out so that $H$ has real eigenvalues instead of imaginary ones, as is expected of an observable.} In addition, a $\pi$ is factored out and left in the exponent in this paper so that common operations like $X$, $Z$, $H$, and CNOT (i.e. $\pi$-pulses, half-rotations of the Bloch sphere for the one-qubit gates) appear in the diagrammatic notation with a power of 1 rather than $\pi$, eliminating the need to write $\pi$ everywhere when working with diagrams. The numbers that appear in the diagrams can thus be interpreted as the time of application of the associated Hamiltonian, as measured in number of $\pi$-pulses. This leaves the functional form
\begin{equation*}
    \ket{\psi(t)} = e^{i\pi tH}\ket{\psi(0)}
\end{equation*}
for all gates. Simple gates ($t=1$) are $\pi$-pulses, and gates with a vertical line have a tensor product Hamiltonian along the line. For example, the symbol $\bm{\bullet}$ represents the Pauli matrix $Z$, and has Hamiltonian $\dyad{1}$. In 1 timestep this operator evolves a state as $\ket{\psi} \mapsto Z\ket{\psi}$. 

\begin{equation*}
\dyad{1} = 
\begin{bmatrix}
   0 \\
   1
\end{bmatrix}
\begin{bmatrix}
   0 & 1
\end{bmatrix}
=
\begin{bmatrix}
   0 & 0 \\
   0 & 1
\end{bmatrix}
\end{equation*}
\begin{equation*}
\begin{quantikz}
& \phase{} &
\end{quantikz} = e^{i\pi\dyad{1}} = \exp(i\pi\begin{bmatrix}
   0 & 0 \\
   0 & 1
\end{bmatrix})
= \exp(\begin{bmatrix}
   0 & 0 \\
   0 & i\pi
\end{bmatrix})
= \begin{bmatrix}
   e^0 & 0 \\
   0 & e^{i\pi}
\end{bmatrix}
= \begin{bmatrix}
   1 & 0 \\
   0 & -1
\end{bmatrix}
= Z
\end{equation*}

As a final note on the assumption of time-independence, it is always assumed in the circuit model of quantum computing that it is physically possible to apply a sequence of unitaries of the above form (i.e. gates) to the qubits, meaning that their associated time-independent Hamiltonians can be switched on and off instantaneously. Discontinuously jumping between different constant Hamiltonians instantaneously is clearly unphysical since it would require infinite energy, so the real statement is rather that we assume the action of the desired unitaries can be well-approximated by sufficiently good tuning of the mechanisms implementing the associated Hamiltonians; perhaps switching a laser pulse on and off as sharply as possible, and keeping the laser output as constant as possible during the pulse. This is the work of experimentalists and quantum hardware engineers, and the practical utility of all theory work done in the circuit model relies implicitly upon it.

\newpage
\section{Ruleset reference sheet}\label{appendix:referencesheet}

\begin{center}
\textbf{Notation}
\end{center}
\noindent\fbox{%
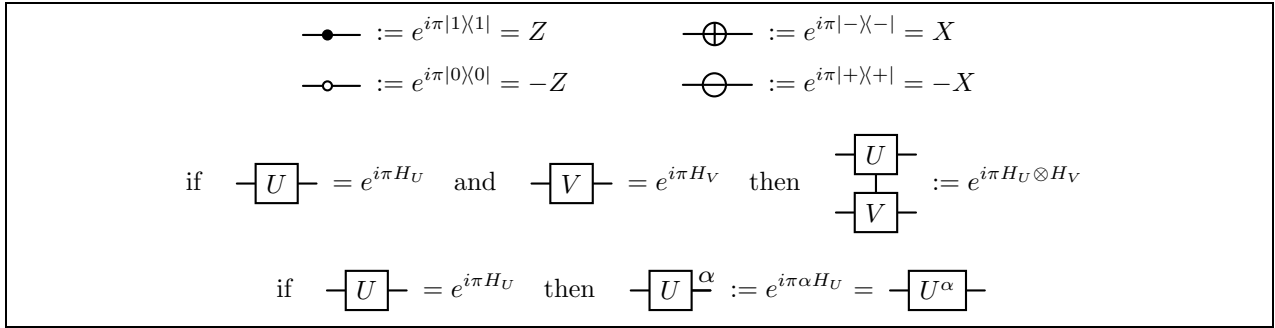
\begin{minipage}{\linewidth}
\begin{equation*}
\begin{aligned}
&\begin{quantikz}
    & \phase{} &
\end{quantikz}
:= e^{i\pi\dyad{1}} = Z
\quad \quad \quad
&\begin{quantikz}
    & \targ{} &
\end{quantikz}
&:= e^{i\pi\dyad{-}} = X \\
&\begin{quantikz}
    & \ophase{} &
\end{quantikz}
:= e^{i\pi\dyad{0}} = -Z
\quad \quad \quad
&\begin{quantikz}
    & \targm{} &
\end{quantikz}
&:= e^{i\pi\dyad{+}} = -X
\end{aligned}
\end{equation*}

\begin{equation*}
\text{if} \quad
\begin{quantikz}
    & \gate{U} &
\end{quantikz}
= e^{i\pi H_U} \quad \text{and} \quad
\begin{quantikz}
    & \gate{V} &
\end{quantikz}
= e^{i\pi H_V} \quad \text{then} \quad
\begin{quantikz}
    & \gate{U}\vqw{1} & \\
    & \gate{V} &
\end{quantikz}
:= e^{i\pi H_U \otimes H_V}
\end{equation*}

\begin{equation*}
\text{if} \quad
\begin{quantikz}
    & \gate{U} &
\end{quantikz}
= e^{i\pi H_U} \quad \text{then} \quad
\begin{quantikz}
    & \gate{U} & \wire[l][1]["\alpha"{above,pos=0.2,scale=1.3}]{a}
\end{quantikz}
:=
e^{i\pi \alpha H_U}
=
\begin{quantikz}
    & \gate{U^\alpha} &
\end{quantikz}
\end{equation*}
\end{minipage}
} 

Where all eigenvalues of $H_U,H_V$ are contained in the interval $[0,2)$ without loss of generality. \\

\begin{center}
\textbf{Ruleset}
\end{center}
\noindent\fbox{%

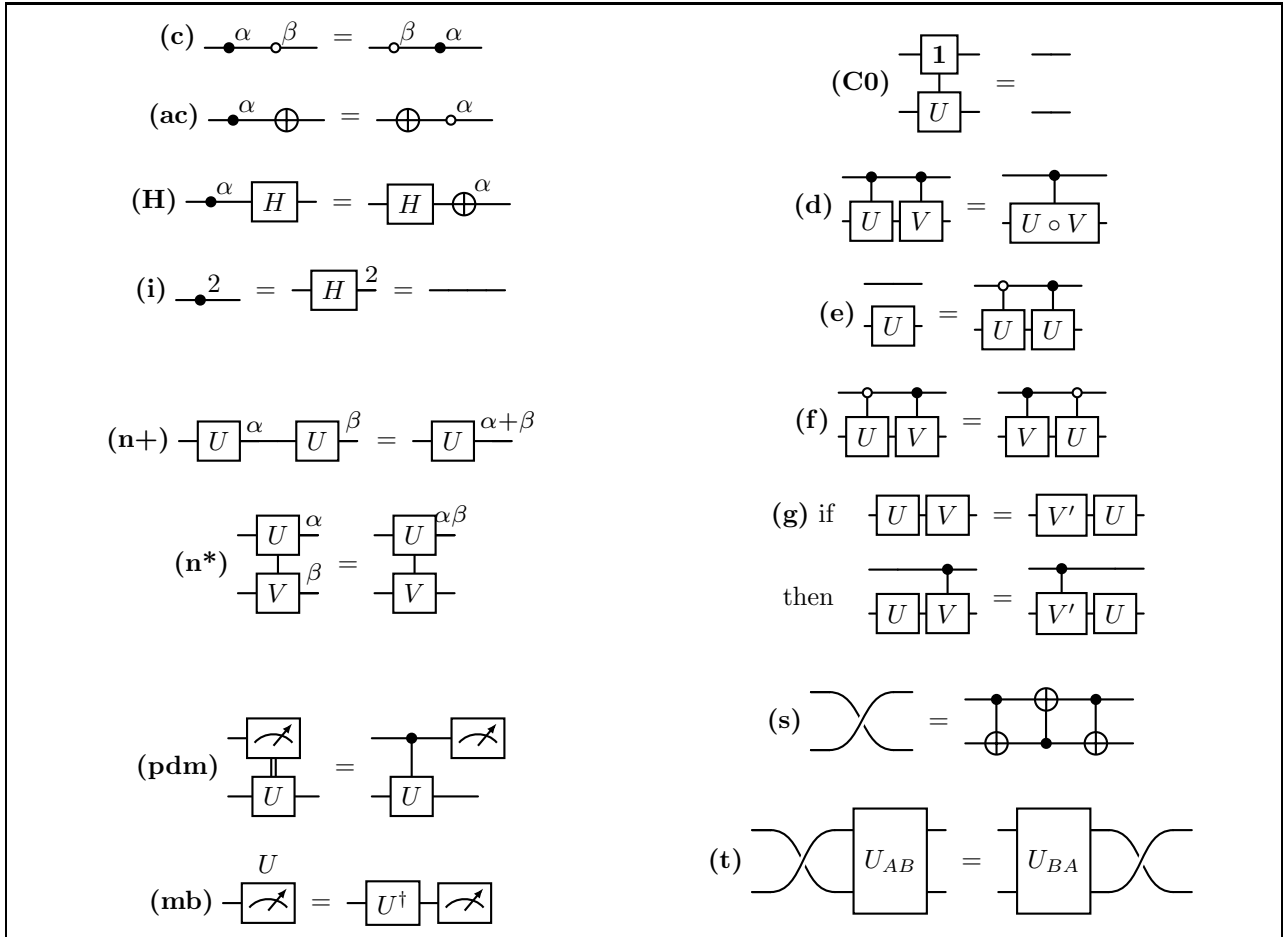
\begin{minipage}{0.5\linewidth}
\begin{equation*}
\textbf{(c)}
\begin{quantikz}
    & \phase{\alpha} & \ophase{\beta} &
\end{quantikz}
=
\begin{quantikz}
    & \ophase{\beta} & \phase{\alpha} &
\end{quantikz}
\end{equation*}

\begin{equation*}
\textbf{(ac)}
\begin{quantikz}
    & \phase{\alpha} & \targ{} &
\end{quantikz}
=
\begin{quantikz}
    & \targ{} & \ophase{\alpha} &
\end{quantikz}
\end{equation*}

\begin{equation*}
\textbf{(H)}
\begin{quantikz}
    & \phase{\alpha} & \gate{H} &
\end{quantikz}
=
\begin{quantikz}
    & \gate{H} & \targ{\alpha} &
\end{quantikz}
\end{equation*}

\begin{equation*}
\textbf{(i)}
\begin{quantikz}
    & \phase{2} & 
\end{quantikz}
=
\begin{quantikz}
    & \gate{H} & \wire[l][1]["2"{above,pos=0.2,scale=1.3}]{a}
\end{quantikz}
=
\begin{quantikz}
    & & & &
\end{quantikz}
\end{equation*}

\vspace{0.75cm}

\begin{equation*}
\textbf{(n+)}
\begin{quantikz}
& \gate{U} & \wire[l][1]["\alpha"{above,pos=0.2,scale=1.3}]{a} & & \gate{U} & \wire[l][1]["\beta"{above,pos=0.2,scale=1.3}]{a}
\end{quantikz}
=
\begin{quantikz}
& \gate{U} & & \wire[l][1]["\alpha+\beta"{above,pos=0.2,scale=1.3}]{a}
\end{quantikz}
\end{equation*}

\begin{equation*}
\textbf{(n*)}
\begin{quantikz}
& \gate{U}\vqw{1} & \wire[l][1]["\alpha"{above,pos=0.2,scale=1.3}]{a} \\
& \gate{V} & \wire[l][1]["\beta"{above,pos=0.2,scale=1.3}]{a}
\end{quantikz}
=
\begin{quantikz}
& \gate{U}\vqw{1} & \wire[l][1]["\alpha\beta"{above,pos=0.2,scale=1.3}]{a} \\
& \gate{V} &
\end{quantikz}
\end{equation*}

\vspace{0.75cm}

\begin{equation*}
\textbf{(pdm)}
\begin{quantikz}
    & \meter{}\vcw{1} \\
    & \gate{U} &
\end{quantikz}
=
\begin{quantikz}
    & \ctrl{1} & \meter{} \\
    & \gate{U} & 
\end{quantikz}
\end{equation*}

\begin{equation*}
\textbf{(mb)}
\begin{quantikz}
    & \meter{U}
\end{quantikz}
=
\begin{quantikz}
    & \gate{U^\dagger} & \meter{}
\end{quantikz}
\end{equation*}

\end{minipage}

\begin{minipage}{0.5\linewidth}
\begin{equation*}
\textbf{(C0)}
\begin{quantikz}
& \gate{\eye}\vqw{1} & \\
& \gate{U} & 
\end{quantikz}
=
\begin{quantikz}
& \ghost{H} & \\
& \ghost{H} & 
\end{quantikz}
\end{equation*}

\begin{equation*}
\textbf{(d)} 
\begin{quantikz}[column sep=0.1cm,row sep=0.25cm]
    & \ctrl{1} & \ctrl{1} & \\
    & \gate{U} & \gate{V} &
\end{quantikz}
=
\begin{quantikz}[column sep=0.1cm,row sep=0.3cm]
    & \ctrl{1} & \\
    & \gate{U\circ V} &
\end{quantikz}
\end{equation*}

\begin{equation*}
\textbf{(e)}
\begin{quantikz}[column sep=0.1cm,row sep=0.3cm]
    & & \\
    & \gate{U} &
\end{quantikz}
=
\begin{quantikz}[column sep=0.1cm,row sep=0.25cm]
    & \octrl{1} & \ctrl{1} & \\
    & \gate{U} & \gate{U} &
\end{quantikz}
\end{equation*}

\begin{equation*}
\textbf{(f)}
\begin{quantikz}[column sep=0.1cm,row sep=0.25cm]
    & \octrl{1} & \ctrl{1} & \\
    & \gate{U} & \gate{V} &
\end{quantikz}
=
\begin{quantikz}[column sep=0.1cm,row sep=0.25cm]
    & \ctrl{1} & \octrl{1} & \\
    & \gate{V} & \gate{U} &
\end{quantikz}
\end{equation*}

\begin{equation*}
\begin{aligned}
\textbf{(g)}
\text{ if} \quad
\begin{quantikz}[column sep=0.1cm,row sep=0.25cm]
    & \gate{U} & \gate{V} &
\end{quantikz}
=
\begin{quantikz}[column sep=0.1cm,row sep=0.25cm]
    & \gate{V'} & \gate{U} &
\end{quantikz} \\
\quad \text{then} \quad
\begin{quantikz}[column sep=0.1cm,row sep=0.25cm]
    & & \ctrl{1} & \\
    & \gate{U} & \gate{V} &
\end{quantikz}
=
\begin{quantikz}[column sep=0.1cm,row sep=0.25cm]
    & \ctrl{1} & & \\
    & \gate{V'} & \gate{U} &
\end{quantikz}
\end{aligned}
\end{equation*}

\begin{equation*}
\textbf{(s)}
\begin{quantikz}
& \permute{2,1} & \ghost{H} \\
& & \ghost{H}
\end{quantikz}
=
\begin{quantikz}
& \ctrl{1} & \targ{} & \ctrl{1} & \\
& \targ{} & \ctrl{-1} & \targ{} &
\end{quantikz}
\end{equation*}

\begin{equation*}
\textbf{(t)}
\begin{quantikz}
& \permute{2,1} & \gate[2]{U_{AB}} & \\
& & \ghost{U_{AB}} &
\end{quantikz}
=
\begin{quantikz}
& \gate[2]{U_{BA}} & \permute{2,1} & \\
& \ghost{U_{BA}} & &
\end{quantikz}
\end{equation*}

\end{minipage}

} 

All rules hold under arbitrary permutation of symbols within types (flip) and/or joint interchange of symbol types (flop). Applying the flip maps $(\text{flip}_\text{z}\text{: }\bm{\bullet}\leftrightarrow\bm{\circ})$ and/or $(\text{flip}_\text{x}\text{: }\bm{\oplus}\leftrightarrow\bm{\ominus})$, and/or the flop map $(\text{flop: }\bm{\bullet}\leftrightarrow\bm{\oplus}$ and $\bm{\circ}\leftrightarrow\bm{\ominus})$ to any rule yields another valid rule.

\end{document}